\documentclass[11pt]{article}
\usepackage[T1]{fontenc}
\usepackage[latin9]{inputenc}
\usepackage{geometry}
\usepackage{amsmath}
\usepackage{amsfonts}
\usepackage{amssymb}
\usepackage{amsthm}
\usepackage{comment}
\usepackage{algorithm}
\usepackage{algorithmicx}
\usepackage{algpseudocode}
\usepackage{enumerate}
\usepackage{mathtools}
\usepackage{thm-restate}
\usepackage{mathrsfs}
\usepackage{csquotes}
\usepackage{apptools}
\usepackage{graphicx}
\usepackage{verbatim}
\usepackage{microtype}
\usepackage{float}
\usepackage{chngcntr}
\usepackage{xr}
\usepackage{bm}
\usepackage{bbm}
\usepackage{todonotes}
\usepackage[shortlabels]{enumitem}
\usepackage{tikz}
\usepackage{pdfsync}
\usepackage{nicefrac}
\usepackage{crossreftools}
\usepackage[hypertexnames=false]{hyperref}
\usepackage[capitalize, nameinlink]{cleveref}

\makeatletter

\numberwithin{equation}{section}
\numberwithin{figure}{section}

\theoremstyle{plain}
\newtheorem{theorem}{Theorem}[section]

\newtheorem{corollary}[theorem]{Corollary}
\newtheorem{lemma}[theorem]{Lemma}

\newtheorem{claim}[theorem]{Claim}

\theoremstyle{remark}
\newtheorem{remark}[theorem]{Remark}

\theoremstyle{definition}
\newtheorem{definition}[theorem]{Definition}

\makeatother

 \newcommand{\TODO}[1]{\textbf{\color{red}[TODO: #1]}}

\algnewcommand{\LeftComment}[1]{\(\triangleright\) #1}

\global\long\def\defeq{\stackrel{\mathrm{{\scriptscriptstyle def}}}{=}}%
\global\long\def\norm#1{\left\Vert #1\right\Vert }%
\def\eps{\varepsilon}
\global\long\def\R{\mathbb{R}}%
\global\long\def\diag{\mathrm{diag}}%

\global\long\def\ot{\overline{t}}%

\global\long\def\ct{\mathcal{T}}%
\global\long\def\cs{\mathcal{S}}%

\global\long\def\treeop{\mathbf{\Delta}}%
\global\long\def\itreeop{\mathbf{\nabla}}%

\global\long\def\sc{\mathbf{Sc}}%
\global\long\def\sketchlen{w}

\newcommand{\hE}{\mathcal{E}}

\newcommand{\zprev}{{\vz^{(\mathrm{step})}}}

\newcommand{\zsum}{{\vz^{(\mathrm{sum})}}}

\global\long\def\mzero{\mathbf{0}}%

\global\long\def\ma{\mathbf{A}}%
\global\long\def\mb{\mathbf{B}}%
\global\long\def\md{\mathbf{D}}%
\global\long\def\mi{\mathbf{I}}%
\global\long\def\ml{\mathbf{L}}%
\global\long\def\mm{\mathbf{M}}%
\global\long\def\ms{\mathbf{S}}%
\global\long\def\mmu{\mathbf{U}}%
\global\long\def\mv{\mathbf{V}}%
\global\long\def\mw{\mathbf{W}}%
\global\long\def\mx{\mathbf{X}}%
\global\long\def\mpi{\mathbf{\Pi}}%

\global\long\def\mga{\mathbf{\Gamma}}%
\global\long\def\mphi{\mathbf{\Phi}}

\global\long\def\mproj{\mathbf{P}}%
\global\long\def\omd{\overline{\mathbf{D}}}%

\global\long\def\vx{\bm{x}}%
\global\long\def\vz{\bm{z}}%
\global\long\def\vd{\bm{d}}%
\global\long\def\vf{\bm{f}}%
\global\long\def\vq{\bm{q}}%
\global\long\def\vb{\bm{b}}%
\global\long\def\vc{\bm{c}}%
\global\long\def\vf{\bm{f}}%
\global\long\def\vs{\bm{s}}%
\global\long\def\vv{\bm{v}}%
\global\long\def\vw{\bm{w}}%
\global\long\def\vt{\bm{t}}%
\global\long\def\vv{\bm{v}}%
\global\long\def\vy{\bm{y}}%
\global\long\def\vl{\bm{l}}%
\global\long\def\vu{\bm{u}}%
\global\long\def\vx{\bm{x}}%

\global\long\def\ox{\overline{\vx}}%
\global\long\def\os{\overline{\vs}}%
\global\long\def\pf{\bm{f}^{\perp}}%

\global\long\def\O{\widetilde{O}}%

\global\long\def\new{{(\mathrm{new})}}
\global\long\def\old{{(\mathrm{old})}}

\global\long\def\init{{(\mathrm{init})}}

\global\long\def\vzero{\bm{0}}%
\global\long\def\vone{\bm{1}}%
\DeclareMathOperator{\nnz}{nnz}

\global\long\def\region{H}
\newcommand{\atH}{^{(H)}}
\newcommand{\bdry}[1]{\partial #1}
\newcommand{\elim}[1]{{F_{#1}}}
\newcommand{\sep}[1]{{S({#1})}}
\newcommand{\skel}[1]{F_{#1} \cup \partial #1}
\newcommand{\pathT}[1]{\mathcal{P}_{\mathcal{T}}(#1)}
\global\long\def\collN{\mathcal{H}}

\newcommand{\nfrac}[2]{\nicefrac{1}{2}}

\let\ref\cref

\author{
	Sally Dong\thanks{Sally Dong was supported by NSERC PGS-D 557770-2021,  and NSF awards CCF-1749609, DMS-1839116, DMS-2023166, CCF-2105772.} \\ University of Washington \\ sallyqd@uw.edu 
	\and 
	Gramoz Goranci \\ University of Vienna \\ gramoz.goranci@univie.ac.at 
	\and
	Lawrence Li\thanks{Lawrence Li was supported by NSERC
          Discovery Grant RGPIN-2018-06398 and Ontario Early Researcher
          Award (ERA) ER21-16-283 awarded to SS.} \\ University of Toronto \\ lawrenceli@cs.toronto.edu
	\and
	Sushant Sachdeva\thanks{Sushant Sachdeva was supported by an NSERC Discovery Grant RGPIN-2018-06398, an Ontario Early Researcher Award (ERA) ER21-16-283, and a Sloan Research Fellowship.} \\ University of Toronto \\ sachdeva@cs.toronto.edu 
	\and 
	Guanghao Ye\thanks{Guanghao Ye was supported by NSF awards CCF-1955217 and DMS-2022448.} \\ Massachusetts Institute of Technology \\ ghye@mit.edu
}

\begin{document}
	\title{Fast Algorithms for Separable Linear Programs}
	\date{\today}
	\maketitle
        \thispagestyle{empty}
	\begin{abstract}
		In numerical linear algebra, considerable effort has been devoted to obtaining faster algorithms for linear systems whose underlying matrices exhibit structural properties. A prominent success story is the method of generalized nested dissection~[Lipton-Rose-Tarjan'79] for separable matrices. On the other hand, the majority of recent developments in the design of efficient linear program (LP) solves do not leverage the ideas underlying these faster linear system solvers nor consider the separable structure of the constraint matrix.
		
		We give a faster algorithm for separable linear programs. Specifically, we consider LPs of the form $\min_{\mathbf{A}\bm{x}=\bm{b}, \bm{\ell} \leq \bm{x} \leq \bm{u}} \bm{c}^\top\bm{x}$, where the graphical support of the constraint matrix $\mathbf{A} \in \mathbb{R}^{n\times m}$ is $O(n^\alpha)$-separable. These include flow problems on planar graphs and low treewidth matrices among others. We present an $\widetilde{O}((m+m^{1/2 + 2\alpha}) \log(1/\epsilon))$ time algorithm for these LPs, where $\epsilon$ is the relative accuracy of the solution.
		
		Our new solver has two important implications: for the $k$-multicommodity flow problem on planar graphs, we obtain an  algorithm running in $\widetilde{O}(k^{5/2} m^{3/2})$ time in the high accuracy regime; and when the support of $\mathbf{A}$ is $O(n^\alpha)$-separable with $\alpha \leq 1/4$, our algorithm runs in $\widetilde{O}(m)$ time, which is nearly optimal. The latter significantly improves upon the natural approach of combining interior point methods and nested dissection, whose time complexity is lower bounded by $\Omega(\sqrt{m}(m+m^{\alpha\omega}))=\Omega(m^{3/2})$, where $\omega$ is the matrix multiplication constant. Lastly, in the setting of low-treewidth LPs, we recover the results of \cite{treeLP} and \cite{GS22} with significantly simpler data structure machinery.
	\end{abstract}
	
	\newpage
        \thispagestyle{empty}
	\tableofcontents
	\newpage

        \setcounter{page}{1}
	\section{Introduction}

Linear programming (LP) is a widely used technique for solving a broad range of problems that emerge in optimization, operations research, and computer science, among others. LP solvers have been a subject of research for many years, both from theoretical as well as practical perspectives. This has led to the development of several algorithmic gems such as the Simplex algorithm~\cite{dantzig1951maximization}, ellipsoid algorithm~\cite{khachiyan1980polynomial} and interior point method~\cite{karmarkar1984new}, to name a few.

Fast solvers for LPs via interior point methods have received considerable attention recently, especially in the theoretical computer science community. A series of improvements culminated in the recent breakthrough work of Cohen, Lee and Song~\cite{CohenLS21}, which shows that any linear program  $\min_{\ma \vx = \vb, \vl \leq \vx \leq \vu} \vc^\top \vx$ with $n$ constraints and $m$ variables can be solved in $\O(m^\omega \log (1/\eps))$ time, where $\eps$ is the accuracy parameter and $\omega \approx 2.3715$ is the matrix multiplication exponent~\cite{duan2022faster,williams2023new}. When $\ma$ is a dense matrix, their running time is almost optimal as it nearly matches the $O(m^\omega)$ algorithm for solving a linear system $\ma \vx = \vb$, which is a sub-problem of linear programming. However, the case when $\ma$ is a sparse matrix is equally important, since the constraint matrices of many LP instances that arise in practical applications happen to be sparse.

A widely-used method for identifying structures in a sparse matrix $\ma$ involves associating a graph with its non-zero pattern, which captures the interactions between the equations in the system. In this paper, we are interested in when said graph is \emph{separable}; we use a weighted-version of the definition as is common in literature, such as~\cite{henzinger1997faster}:

\begin{definition}[Separable graphs]
	\label{defn:separable-graph}
	A (hyper-)graph $G=(V, \hE)$ is \emph{$n^\alpha$-separable} for some $\alpha \in [0,1]$ if there exists constants $b\in (0, 1)$ and $c > 0$, such that for any vertex weight assignment $w$, the vertices of $G$ can be partitioned into $S$, $A$ and $B$ such that $|S| \leq c \cdot |V|^\alpha$, there are no edges between $A$ and $B$, and $\max \{ w(A), w(B)\}  \le b \cdot w(V)$.
	We call $S$ the \textit{($b$-)balanced vertex separator} of $H$ (with respect to $w$).
\end{definition}
A notable case is $\alpha = 1/2$, which includes the family of planar and bounded-genus graphs~\cite{lipton1979separator}. It has also been empirically observed that road networks have separators of size $n^{1/3}$\cite{DibbletSW14,SchildS15}. 

Building upon the seminal work of George~\cite{george1973nested}, Lipton Tarjan and Rose~\cite{lipton1979generalized} introduced the
generalized nested dissection algorithm, which solves the linear system $\ma \vx = \vb$ in $O(m + m^{\alpha \omega})$ time when $\ma$ is a symmetric-positive definite matrix and the associated graph is $O(n^\alpha)$-separable. 
When $\alpha < 1$, this algorithm outperforms the canonical $O(m^\omega)$-time algorithm for general linear systems.
Motivated by this, we ask the natural question of how to leverage the structures in the constraints to speed up linear programming:

\begin{center}
\emph{Are there faster LP solvers for the class of problems where the constraint matrix $\ma$ can be represented by an $O(n^\alpha)$-separable graph?}
\end{center}

Given the constraint matrix $\ma$, \cite{lipton1979generalized} associates with it the unique graph whose adjacency matrix has the same non-zero pattern as $\ma$.
In the context of linear programs, we define the \emph{dual graph} $G_{\ma}$ of a constraint matrix $\ma \in \R^{n \times m}$ to be the hypergraph with vertex set $\{1,\ldots,n\}$ corresponding to the rows of $\ma$ and hyper-edges $\{e_1,\dots, e_m\}$, such that vertex $i$ is in hyperedge $e_j$ if $\ma_{i,j} \neq 0$.

In this paper, we present a faster solver for LPs whose dual graph is separable.
\begin{theorem}\label{thm:main}
	Given a linear program $\min \; \{ \vc^{\top} \vx \;:\; \ma \vx = \vb, \vl \leq  \vx \leq \vu\}$, where $\ma \in \R^{n \times m}$ is a full-rank matrix with $n \leq m$,
	suppose the dual graph $G_\ma$ is $O(n^\alpha)$-separable with $\alpha \in [0,1]$, and a balanced separator is computable in $T(n)$ time.
	
	Suppose that $r$ is the inner radius of the polytope, namely, there is $\vx$ such that $\ma \vx= \vb$ and $\vl + r \leq \vx \leq \vu - r$.
	Let $L = \norm{\vc}_2$ and $R = \norm{\vu-\vl}_2$. 
	Then, for any $0<\eps\leq1/2$, we can find a feasible $\vx$ with high probability such
	that
	\[
		\vc^{\top} \vx \leq \min_{ \ma \vx=\vb,\, \vl \leq \vx \leq \vu} \vc^{\top} \vx+\eps\cdot L R,
	\]
	in time
	\[
	\widetilde{O}\left((m+m^{1/2 + 2\alpha})\cdot \log(R/(r \eps)) + T(n) \right).
	\]
\end{theorem}

Our result should be compared against the natural $\O(m^{1/2}(m+m^{\alpha \omega}))$ runtime, which directly follows from the fact that IPM-based methods require $\O(\sqrt{m})$ iterations, each of which can be implemented in $O(m + m^{\alpha \omega})$ time using the nested dissection~\cite{lipton1979generalized,alon2010solving} algorithm. 
For linear programs whose dual graph are $O(n^{\alpha})$-separable with $\alpha \leq 1/4$, our algorithm achieves $\O(m \log (1/\eps))$ time, which is optimal up to poly-logarithmic factors.

We would like to emphasize that $\O(m+m^{1/2 + 2\alpha})$ represents a natural barrier for the (robust) IPM-based approaches. At a high level, each iteration of IPM involves performing matrix operations using the inverse of an $O(m^{\alpha}) \times O(m^{\alpha})$ matrix\footnote{For $O(n^\alpha)$-separable graph, where $\alpha<1$, it is known that $m=O(n)$, see e.g., \cite{lipton1979generalized}.}. Even if one is given access to said inverse, multiplying a vector against it takes at least $\Omega(m^{2\alpha})$, showing that improving upon the $m^{2\alpha}$ factor will require significantly new ideas in the design and analysis of robust Interior Point Methods. Obtaining an LP solver whose time complexity is $\O(m+m^{\alpha \omega})$, which would in turn nearly match the time complexity for solving linear systems with recursively separable structure, remains an outstanding open problem~\cite{GohbergKK86}.

An immediate application of \cref{thm:main} is a faster
algorithm for solving the (fractional) $k$-commodity flow problem on
planar graphs to high accuracy. For general sparse graphs, an $\O((km)^{\omega})$ time algorithm for this problem follows by the recent linear program solvers that run in matrix multiplication time~\cite{CohenLS21,van2020deterministic}. It is known that solving the $k$-commodity flow problem is as hard as linear programming~\cite{Itai78,DingK022}, suggesting that additional structural assumptions on the input graph are necessary to obtain faster algorithms. As shown in
the theorem below, our result achieves a polynomial speed-up when
the input graph is planar.
\begin{theorem} \label{thm:k-multicommodity-flow}
	Given a minimum-cost $k$-multicommodity flow problem on a planar graph $G = (V,E)$ on $n$ vertices and $m$ edges, with edge-vertex incidence matrix $\mb$, integer edge capacities $\vu \in \R_{\geq 0}^E$, integer costs $\vc_1, \dots, \vc_k \in \R^E$ and integer demands $\vd_1, \dots, \vd_k \in \R^E$ for each commodity, we can solve the LP
	\begin{equation}\label{eq:k-commodity-LP}
	\begin{split}
		\min \quad \sum_{i=1}^k \; & \vc_i^\top \vf_i \\
		\text{s.t} \quad \mb^\top \vf_i &= \vd_i \qquad \forall i \in [k] \\
		\sum_{i=1}^k \vf_i &\leq \vu \\
		\vf_i &\geq \vzero \qquad\; \forall i \in [k]
	\end{split}
	\end{equation}
	to $\epsilon$ accuracy in $\O(k^{2.5} m^{1.5} \log (M/\eps))$ time, where $M$ is an upper on the absolute values of $\vu, \vc, \vd$.
\end{theorem}

Our main result also has the important advantage of recovering and simplifying the recent work by Dong, Lee and Ye~\cite{treeLP} and Gu and Song~\cite{GS22} who obtain fast solvers for LPs whose constraint matrix has bounded treewidth. 

\begin{theorem} \label{thm:treewidth}
	Suppose we have a linear program with the same setup as \cref{thm:main}, and we are given a tree-decomposition of the dual graph $G_{\ma}$\footnote{We can view the hypergraph $G_{\ma}$ as a graph, where we interpret each hyper-edge as a clique, and consider its treewidth as usual.} of width $\tau$. Then we can solve the linear program in  time 
	\[
	\O(m \tau^2 \log (R/(\eps r))) \text{ or } \O(m \tau^{(\omega+1)/2} \log (R/(\eps r))).
	\]

\end{theorem}

\subsection{Previous work}

It is known that $2$-commodity flow is as hard as linear programming~\cite{Itai78}. Recently, \cite{DingK022} showed a linear-time reduction from linear programs to \emph{sparse} $k$-commodity flow instance, indicating that sparse $k$-commodity flow instances are hard to solve. 
This has led to renewed interest in solving $k$-commodity flow in restricted settings, with the authors of~\cite{van2023faster} making progress on dense graphs. 

\paragraph{Linear programming solvers.} The quest for understanding the computational complexity of linear programming has a long and rich history in computer science and mathematics. Since the seminal works of Khachiyan~\cite{khachiyan1980polynomial} and later Karmarkar~\cite{karmarkar1984new}, who were the first to prove that LPs can be solved in polynomial time, the interior point method and its subsequent variants have become the central methods for efficiently solving linear programs with provable guarantees. This has led to a series of refined and more efficient IPM-based solvers~\cite{renegar1988polynomial, vaidya1996new,nesterov1991acceleration, abs-1910-08033, lee2019solving, CohenLS21, Jiang0WZ21}, which culminated in the recent breakthrough work of Cohen, Lee, and Song~\cite{CohenLS21} who showed that an LP solver whose running time essentially matches the matrix multiplication cost, up to small low-order terms. In a follow-up work, Brand~\cite{van2020deterministic} managed to derandomize their algorithm while retaining the same time complexity.

A problem closely related to this paper is solving LPs when the support of the constraint matrix has bounded treewidth $\tau$. Dong, Lee and Ye~\cite{treeLP} showed that such structured LPs can be solved in $\O(m \tau^2)$, which is near-linear when $\tau$ is poly-logarithmic in the parameters of the input. 

\paragraph{High-accuracy and approximate multi-commodity flow.} As mentioned above, it is known that we can solve multicommodity flow in the high-accuracy regime using linear programming. For a graph with $n$ nodes, $m$ edges, and $k$ commodities, the underlying constraint matrix has $km$ variables and $kn+m$ equality constraints. Thus, using the best-known algorithms for solving linear programs~\cite{CohenLS21,van2020deterministic}, one can achieve a runtime time complexity of $\O((km)^{\omega})$ for solving multi-commodity flow. In the special case of dense graphs, Brand and Zhang~\cite{van2023faster} recently showed an improved algorithm achieving $\O(k^{2.5} \sqrt{m}n^{\omega-1/2})$ runtime.

In the approximate regime, Leighton et al.~\cite{leighton1991fast} show that $(1+\epsilon)$ multi-commodity flow on \emph{undirected graphs} can be solved in $\O(kmn)$, albeit with a rather poor dependency on $\epsilon$. This result led to several follow-up improvements in the low-accuracy regime~\cite{garg2007faster, fleischer2000approximating, madry2010faster}. Later on, breakthrough works in approximating single commodity max flow in nearly-linear time were also extended to the $k$-commodity flow problem on undirected graphs~\cite{kelner2014almost,sherman2013nearly,peng2016approximate}, culminating in the work of Sherman~\cite{sherman2017area} who achieved an $\O(mk \epsilon^{-1})$ time algorithm for the problem.

\paragraph{Multi-commodity flow on planar graphs.}
The multi-commodity flow problem on planar graphs was studied in the 1980s, but there has not been much interest in it until most recently. Results in the past focused on finding conditions under which solutions existed~\cite{OkamuraS81}, or finding simple algorithms in even more restricted settings, with the authors of~\cite{MatsumotoNS85} demonstrating that the problem could be solved in $O(kn + n^2(\log{n})^{1/2})$ time if the sources and sinks were all on the outer face of the graph. More recently, \cite{KawarabayashiK13} studied the all-or-nothing version of planar multi-commodity flow, where flows have to be integral, and demonstrate that an $O(1)$-approximation could be achieved in polynomial time. 

\paragraph{Max flow and min-cost flow on general graphs.} In what follows, we will focus on surveying only \emph{exact} algorithms for max-flow and min-cost flow on general graphs. For earlier developments on these problems, including fast approximation algorithms, we refer the reader to the following works~\cite{king1994faster,ahuja1988network,christiano2011electrical,sherman2013nearly,kelner2014almost,peng2016approximate,sidford2018coordinate,bernstein2021deterministic}, and the references therein.

An important view, unifying almost all recent max-flow or min-cost flow developments, is interpreting max-flow as the problem of finding one unit of $s$-$t$ flow that minimizes the $\ell_\infty$ congestion of the flow vector. Motivated by the near-linear Laplacian solver of Spielman and Teng~\cite{spielman2004nearly} (which in turn can be used to solve the problem of finding one unit of $s$-$t$ flow that minimizes the $\ell_2$ congestion), and the fact that the gap between $\ell_\infty$ and $\ell_2$ is roughly $O(\sqrt{m})$, Daitch and Spielman~\cite{daitch2008faster} showed how to implement the IPM for solving min-cost flows in $\O(m^{3/2})$ time. 

Follow-up works initially made progress on the case of unit capacitated graphs, with the work of Madry~\cite{madry2013navigating} achieving an $\O(m^{10/7})$ time algorithm for max flow and thus being the first to break the 3/2-exponent barrier in the runtime. The running time was later improved to $O(m^{4/3+o(1)})$ and it was generalized to the min-cost flow problem~\cite{AMV22,KathuriaLS20}.

For general, polynomially bounded capacities, Brand et al.~\cite{BrandLLSSSW21} gave an improved algorithm for dense graphs that runs in $\O(m+n^{3/2})$. In the sparse graph regime, Gao, Liu and Peng~\cite{GaoLP21:arxiv} were the first to break the 3/2-exponent barrier by giving an $\O(m^{3/2 -1/128})$ time algorithm, which was later improved to $\O(m^{3/2 -1/58})$~\cite{BGJLLPS21}. Very recently, the breakthrough work of Chen et al.~\cite{ChenKLPGS22} shows that the min-cost flow problem can be solved in $\O(m^{1+o(1)})$, which is optimal up to the subpolynomial term.

\paragraph{Max flow and min-cost flow on planar graphs.} The study of flows on planar graphs dates back to the celebrated work of Ford and Fulkerson~\cite{ford1956maximal} who showed that for the case of $s,t$-planar graphs\footnote{planar graphs where $s$ and $t$ lie on the same face}, there is an $O(n^2)$ time algorithm for max flow. This was subsequently improved to $O(n \log n)$ by Itai and Shiloach~\cite{itai1979maximum} and finally to $O(n)$ by Henzinger et al~\cite{henzinger1997faster}, the latter building upon a prior work of Hassin~\cite{Hassin}.

For general planar graphs, there have been two lines of work focusing on the undirected and the directed version of the problem respectively. In the first setting, Reif~\cite{reif1983minimum} (and later Hassin and Johnson~\cite{hassin1985n}) gave an $O(n \log^2 n)$ time algorithm. The state-of-the-art algorithm is due to Italiano et al.~\cite{INSW11} and achieves $O(n \log \log n)$ runtime. Weihe~\cite{weihe1997maximum} gave the first speed-up for directed planar max flow running in $O(n \log n)$ time. However, his algorithm required some assumptions on the connectivity of the input graph. Later on, Borradaile and Klein~\cite{borradaile2008exploiting} gave an $O(n \log n)$ algorithm for general planar directed graphs. Generalization of planar graphs, e.g., graphs of bounded genus have also been studied in the context of the max flow problem. The work of Chambers et al.~\cite{CEFN19:arxiv} showed that these graphs also admit near-linear time max flow algorithms.

Imai and Iwano~\cite{II90} obtained an $O(n^{1.594} \log M)$ min-cost flow algorithm for graphs that are $O(\sqrt{n})$ recursively separable. For the min-cost flow problem on planar graphs with unit capacities, Karczmarz and Sankowski~\cite{KS19} gave an $O(n^{4/3})$ algorithm. Very recently, Dong et al.~\cite{dong2022nested} showed that the min-cost flow on planar directed graphs with polynomially bounded capacities admits an $\O(n)$ time algorithm, which is optimal up to polylogarithmic factors.

\subsection{Technical overview}
Our algorithm framework builds on the work of
Dong-Gao-Goranci-Lee-Peng-Sachdeva-Ye on planar min-cost flow~\cite{dong2022nested}.  
We solve our linear program using the
robust interior point method used in \cite{treeLP,
  dong2022nested}, where we maintain feasible primal and dual
solutions $\vx$ and $\vs$ to the linear program that converge to the
optimal solution over $\O(\sqrt{m})$-many steps of IPM.  At every
step, we want to move our solutions in the direction of steepest descent of the objective function. 
To stay close to the central path and avoid violating the capacity lower and upper bounds,
the IPM controls the weights $\mw$ on the variables and the step direction $\vv$, 
in order to limit the magnitude of the update to a variable as it approach its bounds.
Both $\mw$ and $\vv$ are defined to be entry-wise dependent on the current solution $\vx$ and $\vs$.
To maintain feasibility of the solutions, we apply the weighted projection
$\mproj_{\vw} \defeq \mw^{1/2} \ma^\top (\ma \mw \ma^\top)^{-1} \ma \mw^{1/2}$ matrix to the desired step direction $\vv$,
which ensures the resulting $\vx$ and $\vs$ after a step remain in their respective feasible subspaces.
In robust IPMs, we also maintain entry-wise
approximations $\ox, \os$ to $\vx$ and $\vs$, and use
these approximations to compute $\vw, \vv$, and $\mproj_{\vw}$ at
every step.  By limiting the updates to $\ox, \os$, robust IPMs
achieve efficient runtimes.

The key challenge in the RIPM framework is to implement each step
efficiently, specifically, computing the projection
$\mproj_{\vw} \vv$, as well as updating $\ox, \os$.  Similar to
\cite{dong2022nested}, we use a \emph{separator tree} to recursively
factor the term $\ma \mw \ma^\top$ in $\mproj_{\vw}$ via nested
dissection and recursive Schur complements.
However, there are several challenges in applying the framework from
\cite{dong2022nested} to general linear programs: In flow problems,
the constraint matrix $\ma$ is the vertex-edge incidence matrix of the
underlying graph, and therefore $\ma \mw \ma^\top$, as well as all
recursive Schur complements along the separator tree, are weighted
Laplacian matrices, for which we have efficient nearly-linear-time
solvers~\cite{spielman2004nearly} and sparse
approximations to Schur complements~\cite{KyngS16,GHP18}.  This allows \cite{dong2022nested}
to work with an approximate
$\widetilde{\mproj}_{\vw} \approx \mproj_{\vw}$ efficiently, with implicit access via a collection of approximate Schur complements that can be viewed as sparse Laplacians.

In the context of general separable linear programs, we do
not have fast solvers or sparse approximate Schur complements, so
instead, we must maintain the collection of Schur complements and their inverses
explicitly, by computing them in a bottom-up fashion using the separator
tree. 
To bound the update time, we show that a rank-$k$
update to $\ma \mw \ma^\top$ induced by changes in $\mw$ corresponds to
rank-$k$ updates to all the recursive Schur complements.

Our second contribution is the dynamic data structures to maintain the implicit representations of $\vx, \vs$,
which can be viewed as a significant refinement of those from \cite{dong2022nested}.
We recall the notion of \emph{tree operators} introduced in \cite{dong2022nested} 
and define an analogous \emph{inverse tree operator}, and give simplified modular data structures
to maintain $\vx, \vs$ using the tree and inverse tree operator.
Specifically, we demonstrate more cleanly the power of nested dissection and the recursive subgraph structure
in supporting efficient lazy updates to the IPM solutions.

Our third technical contribution is the definition of a fine-grained separator tree which we call the $(a,b,\lambda)$-separator tree. 
The parameters are defined based on the parameters of separable graphs, but they also capture important characteristics of other classes such as low-treewidth graphs.
These trees guarantee that at any node, we are able to separate not only the associated graph region, but also the boundary of the region. 
We use them to maintain the tree operators from the implicit representations,
and a careful analysis of node and boundary sizes allows us to conclude that the maintenance can be performed efficiently. 

Finally, we note that this work recovers the treewidth LP result of \cite{treeLP} and \cite{GS22} with significantly lower technical complexity.
Whereas \cite{treeLP} constructs an \emph{elimination tree} 
to directly compute the Cholesky factorization of $\ma \mw \ma^\top = \ml \ml^\top$,
we use a \emph{separator tree} to recursively factor $\ma \mw \ma^\top$.
There is a key difference in the two tree constructions, which we believe this paper is correct in:
To construct an elimination tree, \cite{treeLP} finds a balanced vertex separator $S$ of $G_{\ma}$, \emph{remove $S$ from $G_{\ma}$ yielding two disconnected subgraphs $H_1, H_2$}, recursively construct the elimination tree for $H_1$ and $H_2$, and attach them as children to a vertical path of length $|S|$ corresponding to the vertices of $S$.
When the treewidth of $\ma \in \R^{n \times m}$ is $t$, this process results in an elimination tree of height $\O(t)$ where each node corresponds to a vertex of $G_{\ma}$, which can then be used to identify explicit coordinates in the Cholesky factor $\ml$ to update when $\mw$ changes.
Next, an extremely involved transformation using \emph{heavy-light decomposition} is needed to turn the elimination tree into a \emph{sampling tree} of height $O(\log n)$,
in order to facilitate the sampling of entries from some implicit vector of the form $\ma^\top \ml^{-\top} \vz$ (required for maintaining $\ox, \os$).
%
In contrast, to construct our separator tree, we find a balanced
vertex separator $S$ of $G_{\ma}$, but \emph{include $S$ in both
  subgraphs that are recursed on}, and partition the hyperedges in $\hE(S)$ arbitrarily between the two subgraphs. 
The resulting separator tree has height $O(\log n)$, where each node corresponds to a subgraph of $G_{\ma}$, and each level of the tree gives a partition of the columns of $\ma$. 
This recursive partitioning gives a cleaner, recursive rather than brute-force factorization of $\ma \mw \ma^\top$,
and leads to a significant difference in the data structures.
When $\mw$ changes at a step, \cite{treeLP} updates the Cholesky factorization by processing one changed coordinate at a time (\cite{GS22} processes one block at a time), so that the data structure update time is linear in the number of new coordinates. On the other hand, our separator tree allows us to update $\mw$ in one pass through the tree and yields a sublinear dependence on the number of new coordinates.
Moreover, as each node in our separator tree naturally corresponds to a subset of columns of $\ma$, 
we can use it in a much more straight-forward manner to sample coordinates of $\ma^\top \ml^{-\top} \vz$.


	\section{Preliminaries}

\paragraph{General Notations.}
\emph{We assume all matrices and vectors in an expression have matching dimensions.} That is, we will trivially pad matrices and vectors with zeros when necessary. This abuse of notation is unfortunately unavoidable as we will be considering lots of submatrices and subvectors.

An event holds with high probability if it holds with probability at least $1-n^c$ for arbitrarily large constant $c$. The choice of $c$ affects guarantees by constant factors. 

We use boldface lowercase variables to denote vectors, and boldface uppercase variables to denote matrices.
We use $\|\vv\|_2$ to denote the 2-norm of vector $\vv$ and $\|\vv\|_{\mm}$ to denote $\sqrt{\vv^\top \mm\vv}$.
We use $\nnz(\vv)$ to denote the number of non-zero entries in the vector $\vv$, equivalently, it is the zero-norm.
For any vector $\vv$ and scalar $x$, we define $\vv+ x$ to be the vector obtained by adding $x$ to each coordinate of $\vv$
and similarly $\vv-x$ to be the vector obtained by subtracting $x$ from each coordinate of $\vv$. 
We use $\vzero$ for all-zero vectors and matrices where dimensions are determined by context. 

For an index set $A$, we use $\vone_{A}$ for the vector with value $1$ on coordinates in $A$ and $0$ everywhere else. 
We use $\mi$ for the identity matrix and $\mi_{S}$ for the identity matrix in $\mathbb{R}^{S \times S}$. 
For any vector $\vx \in \mathbb{R}^{S}$,
$\vx|_{C}$ denotes the sub-vector of $\vx$ supported on $C\subseteq S$; 
\emph{more specifically, $\vx|_C \in \R^S$, where $\vx_i = 0$ for all $i \notin C$.}

For any matrix $\mm \in \mathbb{R}^{A\times B}$, 
we use the convention that $\mm_{C, D}$ denotes the sub-matrix of $\mm$ supported on $C\times D$ where $C\subseteq A$ and $D\subseteq B$. 
When $\mm$ is not symmetric and only one subscript is specified, as in $\mm_D$, this denotes the sub-matrix of $\mm$ supported on $A \times D$.
To keep notations simple, $\mm^{-1}$ will denote the inverse of $\mm$ if it is an invertible matrix and the Moore-Penrose pseudo-inverse otherwise.

For any vector $\vv$, we use the corresponding capitalized letter $\mv$ to denote the diagonal matrix with $\vv$ on the diagonal.

For two positive semi-definite matrices $\ml_1$ and $\ml_2$, we write $\ml_1 \approx_t \ml_2$ if $e^{-t} \ml_1\preceq \ml_2 \preceq e^{t} \ml_1$, where $\ma\preceq \mb$ means $\mb-\ma$ is positive semi-definite. 
Similarly we define $\geq_t$ and $\leq_t$ for scalars, that is, $x\leq_t y$ if $e^{-t}x\le y\le e^t x$. 

When multiplying two matrices of differing sizes, say an $m \times n$ matrix with an $n \times k$ matrix, we decompose both matrices into blocks of size $\min \{m,n,k\}$. We then perform block matrix multiplication, with fast matrix multiplication used for the multiplication of two blocks. For example, multiplying a $m \times n$ matrix with an $n \times n$ matrix, with $m \geq n$, takes $(m/n)(n^{\omega})$ time. 

\paragraph{Trees.} 
For a tree $\ct$, we write $H \in \ct$ to mean $H$ is a node in $\ct$.
We write $\ct_H$ to mean the complete subtree of $\ct$ rooted at $H$.
We say a node $A$ is an ancestor of $H$ and $H$ is a descendant of $A$ if $H$ is in the subtree rooted at $A$, and $H \neq A$.
Given a set of nodes $\collN$, we use $\pathT{\collN}$ to denote the set of all nodes in $\ct$ that are ancestors of some node in $\collN$ unioned with $\collN$.

The \emph{level} of a node in a tree has the following properties: the root is at level 0; the maximum level is one less than the height of the tree; and the level of a node must be at least one greater than the level of its parent, but this difference does not have to be equal to one. We may assign levels to nodes arbitrarily as long as the above is satisfied.
We use $\ct(i)$ to denote the collection of all nodes at level $i$ in tree $\ct$.

\paragraph{IPM data structures.} 

When we discuss data structures in the context of the IPM, step 0 means the initialization step. For $k > 0$, step $k$ means the $k$-th iteration of the while-loop in \textsc{Solve} (\cref{alg:IPM_centering}); that is, it is the $k$-th time we update the current solutions.
For any vector or matrix $\vx$ used in the IPM, we use $\vx^{(k)}$ to denote the value of $\vx$ at the end of the $k$-th step.

In all procedures in these data structures, we assume inputs are given by the set of changed coordinates and their values, 
\emph{compared to the previous input}. 
Similarly, we output a vector by the set of changed coordinates and their values, compared to the previous output. 
This can be implemented by checking memory for changes.


	\section{Overview of RIPM framework} \label{sec:ipm-framework}

In this section, we set up the general framework for solving a linear program using a 
robust IPM. 
We show that if the projection matrix from the IPM can be maintained efficiently based on the structure of its sparsity pattern,
then the overall IPM can be implemented efficiently.



\subsection{Robust interior point method}

\begin{restatable}[RIPM]{theorem}{RIPM}
	\label{thm:IPM}
	Consider the linear program 
	\[
	\min_{\ma \vx=\vb,\; \vl\leq\vx\leq\vu}\vc^{\top}\vx
	\]
	with $\ma\in\R^{n \times m}$. We are given a scalar $r>0$ such that \emph{there exists} some interior point $\vx_{\circ}$ satisfying 
	$\ma \vx_{\circ}=\vb$ and $\vl+r\leq \vx_{\circ} \leq \vu-r.$
	Let $L=\|\vc\|_{2}$ and $R=\|\vu-\vl\|_{2}$. 
	For any $0<\eps\leq1/2$,
	the algorithm $\textsc{RIPM}$ (\cref{alg:IPM_centering}) finds $\vx$ such that $\ma \vx=\vb$,
	$\vl\leq\vx\leq\vu$ and
	\[
	\vc^{\top}\vx\leq\min_{\ma \vx=\vb,\; \vl\leq\vx\leq\vu}\vc^{\top}\vx+\eps LR.
	\]
	Furthermore, the algorithm has the following properties:
	\begin{itemize}
		\item Each call of \textsc{Solve} involves $O(\sqrt{m}\log m\log(\frac{mR}{\epsilon r}))$-many steps, and $\ot$ is only updated $O(\log m\log(\frac{mR}{\epsilon r}))$-many
		times.
		\item In each step of \textsc{Solve}, the coordinate $i$ in $\vw,\vv$ changes only if $\ox_{i}$
		or $\os_{i}$ changes.
		\item In each step of \textsc{Solve}, $h\|\vv\|_{2}=O(\frac{1}{\log m})$.
		\item \cref{line:step_given_begin} to \cref{line:step_given_end} takes $O(K)$
		time in total, where $K$ is the total number of coordinate changes
		in $\ox,\os$.
	\end{itemize}
\end{restatable}
%

We note that this algorithm only requires access to
$(\ox,\os)$, but not $(\vx,\vs)$ during the main while loop.
Hence, $(\vx,\vs)$ can be implicitly maintained via any data structure.  
We only require $(\vx,\vs)$ explicitly when returning the 
approximately optimal solution at the end of the algorithm \cref{line:step_user_output}.

\subsection{Projection operators} \label{subsec:overview_projection_operator}

At step $k$ of \textsc{Solve} with step direction $\vv^{(k)}$ and weights $\vw$ (we drop its superscript $^{(k)}$ for convenience), 
recall we define the projection matrix
\[
 \mproj_{\vw} \defeq \mw^{1/2}\ma^\top(\ma\mw\ma^\top)^{-1}\ma\mw^{1/2}.
\]
We want to make the primal and dual updates
\begin{align*}
	\vx &\leftarrow\vx+ h^{(k)} \mw^{1/2} \vv^{(k)} - h^{(k)} \mw^{1/2} \mproj_{\vw} \vv^{(k)},\\
	\vs &\leftarrow\vs+\ot h^{(k)} \mw^{-1/2}\mproj_{\vw} \vv^{(k)}.
\end{align*}

The first term for the primal update is straightforward to maintain,
so we may ignore it without loss of generality.
After this reduction, we see that the primal and dual updates are analogous.
In the remainder of this section, we show how to maintain $\vx$ undergoing the update
\begin{equation*} 
	\vx \leftarrow \vx + h^{(k)} \mw^{1/2} \mproj_{\vw} \vv^{(k)}.
\end{equation*}

First, observe that $\mw^{1/2} \mproj_{\vw}$ is an operator dependent on the dynamic weights $\vw$,
which motivates us to formalize this problem setting:
\begin{definition}[Dynamic linear operator, update complexity]
	Let $\vw$ be a dynamic vector. 
	We say $\mm$ is a \emph{dynamic linear operator dependent on $\vw$} if $\mm$ is a function of $\vw$. 
	Let $\vw^{(k)}$ be the value of $\vw$ at step $k$,
	then we use $\mm^{(k)}$ to denote the corresponding value of $\mm$ at step $k$.
	
	Suppose exists a data structure that dynamically maintains $\mm$ and $\vw$, such that at every step $k$, if $\vw^{(k-1)}$ and $\vw^{(k)}$ differ on $K$ coordinates, then the data structure can update $\mm^{(k-1)}$ to $\mm^{(k)}$ in $f(K)$ time. Then we say \emph{$\mm$ has update complexity $f$}.
\end{definition}

Next, we define two types of dynamic operators dependent on the weights $\vw$ from the IPM: 
the \emph{inverse tree operator} $\itreeop$ and the \emph{tree operator} $\treeop$.
For linear programs with separable structures, they should crucially combine so that throughout algorithm, we have
\begin{equation} \label{eq:projection-as-tree-ops}
	\mw^{1/2} \mproj_{\vw}  = \treeop \itreeop.
\end{equation}

\subsubsection{Operators on a tree}

In this section, we fix $\ct$ to be a \emph{constant-degree} rooted tree with root node $G$, called the \emph{operator tree}.
Let each node $H \in \ct$ be associated with a set $F_{H}$, where
the $F_H$'s are pairwise disjoint.
Let each \emph{leaf} node $L \in \ct$ be further associated with a \emph{non-empty} set $E(L)$, 
where the $E(L)$'s are pairwise disjoint.
For a non-leaf node $H$, define $E(H) \defeq \bigcup_{\text{leaf } L \in \ct_H} E(L)$.
Finally, define $E \defeq E(G) = \bigcup_{\text{leaf } L \in \ct} E(L)$ and 
$V \defeq \bigcup_{H\in \ct} F_H$.

We define two special classes of linear operators that build on the structure of $\ct$.
The advantage of these operators lie in their decomposability, which allows them to be efficiently maintained.

\begin{definition}[Inverse tree operator] \label{defn:inverse-tree-operator}
Let $\ct$ be an operator tree with the associated sets as above.
We say a linear operator $\itreeop : \R^E \mapsto \R^V$ is an \emph{inverse tree operator supported on $\ct$} 
if there exists a linear \emph{edge operator} $\itreeop_{H}$ for each non-root node $H$ in $\ct$, corresponding to the edge from $H$ to its parent,
such that $\itreeop$ can be decomposed as
\[
\itreeop = \sum_{\text{leaf $L$, node $H$} \;:\; L \in \ct_H} 
\mi_{F_H} \itreeop_{H \leftarrow L},
\]
where $\itreeop_{H \leftarrow L}$ is defined as follows:
If $L = H$, then $\itreeop_{H \leftarrow L} \defeq \mi$;
otherwise, suppose the path in $\ct$ from leaf $L$ to node $H$ is given by $(H_t \defeq L, H_{t-1}, \dots, H_{1}, H_0 \defeq H)$, then
\[
\itreeop_{H \leftarrow L} \defeq \itreeop_{H_{1}} \cdots \itreeop_{H_{t-1}}  \itreeop_{H_{t}}.
\]
To maintain $\itreeop$, it will suffice to maintain $\itreeop_{H}$ at each non-root node $H$ in $\ct$.
\end{definition}

Intuitively, when applying an inverse tree operator to a vector $\vv \in \R^E$, $\vv$ is partitioned according to the leaves of $\ct$, and then the edge operators are applied sequentially along the tree edges in a bottom-up fashion.
It is natural to then also define the opposite process, where edge operators are applied along the tree edges in a top-down fashion.

\begin{definition}[Tree operator]
	\label{def:tree-operator} 
	Let $\ct$ be an operator tree with the associated sets as above.
	We say a linear operator $\itreeop : \R^V \mapsto \R^E$ is \emph{tree operator supported on $\ct$} 
	if there exists a linear edge operator $\itreeop_{H}$ for each non-root node $H$ in $\ct$, corresponding to the edge from $H$ to its parent,
	such that $\itreeop$ can be decomposed as
	\[
		\treeop \defeq \sum_{\text{leaf $L$, node $H$} \;:\; L \in \ct_H} \treeop_{L \leftarrow H}\mi_{F_{H}}.
	\]
	where $\treeop_{H \leftarrow L}$ is defined as follows:
	If $L = H$, then $\treeop_{L \leftarrow H} \defeq \mi$.
	Otherwise, suppose the path in $\ct$ from node $H$ to leaf $L$  is given by $(H_t \defeq L, H_{t-1}, \dots, H_0 \defeq H)$, then
	\[
	\treeop_{L \leftarrow H} \defeq \treeop_{H_{t}} \cdots \treeop_{H_{2}}  \treeop_{H_{1}}.
	\]
\end{definition}

We define the complexity of a tree (and inverse tree) operator to be parameterized by the number of edge operators applied.

\begin{definition}[Query complexity]
	\label{def:tree-operator-complexity}
	Let $\treeop \defeq \{ \treeop_H : H \in \ct\}$ be a tree (or inverse tree) operator on tree $\ct$.
	Suppose for any set $\collN$ of $K$ distinct non-root nodes in $\ct$, and any two families of $K$ vectors indexed by $\collN$, $\{\vu_H : H \in \collN \}$ and $\{\vv_H : H \in \collN \}$,
	the total time to compute
	$\{\vu_H^{\top}\treeop_H : H \in \collN\}$ and $\{\treeop_H \vv_H : H \in \collN\}$ is bounded by $f(K)$. 
	Then we say $\treeop$ has query complexity $f$ for some function $f$.
	
	Without loss of generality, we may assume $f(0) = 0$, $f(k) \geq k$, and $f$ is concave.
\end{definition}

By examining the definition of the inverse tree and tree operator, we see they are related.
\begin{lemma} \label{lem:treeop-transpose}
	If $\treeop$ is a tree operator on $\ct$, then $\treeop^\top$ is an inverse tree operator on $\ct$, where its edge operators are obtained from $\treeop$'s edge operators by taking a transpose.  
	Furthermore, $\treeop$ and $\itreeop$ have the same query and update complexity.
	\qed
\end{lemma}

\subsection{Implicit representations of the solution}
\label{subsec:overview_representation}

Assuming we have dynamic inverse tree and tree operators $\itreeop$ and $\treeop$ on tree $\ct$ dependent on $\vw$ such that $\mw^{1/2} \mproj_{\vw}  = \treeop \itreeop$,
we can now state how to abstractly maintain the implicit representation of the solutions throughout \textsc{Solve} (\cref{alg:IPM_centering}).
Specifically, we want to maintain the solution $\vx$, and at every step $k$, 
carry out an update of the form
\begin{equation} \label{eq:simplified-x-update}
	\vx \leftarrow \vx + h^{(k)} \mw^{1/2} \mproj_{\vw} \vv^{(k)}.
\end{equation}

We design a data structure \textsc{MaintainRep} to accomplish this, by:
\begin{itemize}
	\item At the start of \textsc{Solve}, initializing the data structure using the procedure \textsc{Initialize} with $\vx = \vx^{\init}$,
	\item At each step $k$, updating the weights $\vw$ in the data structure using the procedure \textsc{Reweight}, followed by updating $\vx$ according to \cref{eq:simplified-x-update} using the procedure \textsc{Move},
	\item At the end of \textsc{Solve}, outputing the final $\vx$ using the procedure \textsc{Exact}.
\end{itemize}

The key to designing an efficient data structure is to make use of the structure of the operators.
Due to their decomposition along $\ct$, we can update the operators and apply them to vectors without exploring all of $\ct$ every time.

\begin{restatable}[Implicit representation maintenance] {theorem}{MaintainRepresentation} \label{thm:maintain_representation}
	Let $\vw$ be the weights changing at every step of \textsc{Solve} (\cref{alg:IPM_centering}).
	Suppose there exists dynamic inverse tree and tree operators $\itreeop$ and $\treeop$ on tree $\ct$ both dependent on $\vw$ such that $\mw^{1/2} \mproj_{\vw}  = \treeop \itreeop$ throughout the IPM.
	Let $Q$ be the max of the query complexity of the tree and inverse tree operator,
	and let $U$ be the max of the update complexity of the two operators.
	Suppose $\ct$ has constant degree and height $\eta$.
	Then there is a data structure \textsc{MaintainRep}
	that satisfies the following invariants at the end of step $k$:
	\begin{itemize}
		\item It explicitly maintains the dynamic weights $\vw$ and step direction $\vv$ from the current step.
		
		\item It explicitly maintains scalar $c$ and vectors $\zprev, \zsum$, which together represent the implicitly-maintained vector $\vz \defeq c \zprev + \zsum$.
		At the end of step $k$, $\zprev = \itreeop^{(k)} \vv^{(k)}$, and
		\[
		\vz = \sum_{i=1}^{k} h^{(i)} \itreeop^{(i)} \vv^{(i)}.
		\]
		\item 
		It implicitly maintains $\vx$ so that at the end of step $k$, 
		\[
		\vx = \vx^{\init} + \sum_{i=1}^{k} h^{(i)} \treeop^{(i)} \itreeop^{(i)} \vv^{(i)},
		\]
		where $\vx^{\init}$ is some initial value set at the start of \textsc{Solve}.
	\end{itemize}
	The data structure supports the following procedures and runtimes:
	\begin{itemize}
		\item $\textsc{Initialize}(\treeop, \itreeop, \vv^\init \in\R^{m},\vw^\init \in\R_{>0}^{m}, \vx^{\init} \in \R^m)$:
		Preprocess and set $\vx \leftarrow \vx^{\init}$.
		
		The procedure runs in $O(U(m) + Q(m))$ time.
		
		\item $\textsc{Reweight}(\delta_{\vw} \in\R_{>0}^{m}$):
		Update the weights to $\vw \leftarrow \vw + \delta_{\vw}$.
		
		The procedure runs in 
		$O(U(\eta K) + Q(\eta K))$ total time,
		where $K = \nnz(\delta_{\vw})$.
		
		\item $\textsc{Move}(h \in \R$, $\delta_{\vv} \in \R^{m}$):
		Update the current step direction to $\vv \leftarrow \vv + \delta_{\vv}$.
		Update the implicit representation of $\vx$ to reflect the following change in value:
		\[
		\vx \leftarrow \vx + h \treeop \itreeop \vv.
		\]
		The procedure runs in $O(Q(\eta K))$ time,
		where $K = \nnz(\delta_{\vv})$.
		
		\item $\textsc{Exact}$:
		Output the current exact value of $\vx$ in $O(Q(m))$ time. 
	\end{itemize}
\end{restatable}

\subsection{Solution approximation}
\label{subsec:overview_approx}

In the IPM, one key operation is to maintain the solution vector $\ox$ that is close to $\vx$ throughout the algorithm.
(Analogously for the slack $\os$ close to $\vs$.)
Since we have implicit representations of the solution $\vx$ from \textsc{MaintainRep},
we now show how to maintain $\ox$ close to $\vx$.
To accomplish this, we use a meta data structure that solves this in
a more general setting introduced in \cite{dong2022nested}.

\begin{restatable}[Approximate vector maintenance with tree operator \cite{dong2022nested}]{theorem}{solnApprox} \label{thm:soln-approx}
Let $0 < \rho < 1$ be a failure probability.
Suppose $\treeop \in \R^{m \times n}$ is a tree operator with query complexity $Q$ and supported on a constant-degree tree $\ct$ with height $\eta$.
There is a randomized data structure \textsc{MaintainApprox} that
takes as input the dynamic weights $\vw$
and the dynamic $\vx$ implicitly maintained according to \cref{thm:maintain_representation} at every step,
and explicitly maintains the approximation $\ox$ to $\vx$ satisfying $\norm{\mw^{-1/2} (\vx - \ox)}_\infty \leq \delta$ at every step with probability $1 - \rho$.

Suppose $\|{\mw^{(k)}}^{-1/2}(\vx^{(k)}-\vx^{(k-1)})\|_2 \leq \beta$ \emph{for all steps $k$}.
Furthermore, suppose $\vw$ is a function of $\ox$ coordinate-wise.
Then, for each $\ell \geq 0$, $\ox$ admits $2^{2\ell}$ coordinate changes every $2^\ell$ steps.
Over $N$ total steps, the total cost of the data structure is
\begin{equation} \label{eq:coordinate-changes-bound}
	\O(\eta^3 (\beta/\delta)^2 \log^3(mN/\rho)) \left( Q(m) + \sum_{k=1}^N Q(S^{(k)}) +  \sum_{\ell=0}^{\log N} \frac{N}{2^\ell} \cdot Q(2^{2\ell}) \right),
\end{equation}
where $S^{(k)}$ is the number of nodes $H$ where $\treeop_H$ or $\vu_H$ in the implicit representation of $\vx$ changed at step $k$.
\end{restatable}


\subsection{Main theorem for the RIPM framework}

We are now ready to state and prove the main result in this framework.

\begin{theorem}[RIPM framework] \label{thm:ripm-main}
Consider an LP of the form
\begin{equation}\label{eq:LP-main}
	\min_{\vx\in\mathcal{P}} \; \vc^{\top}\vx\quad\text{where}\quad\mathcal{P}=\{\ma \vx=\vb,\; \vl\leq\vx\leq\vu\}
\end{equation}
where $\ma\in\mathbb{R}^{n\times m}$.
For any vector $\vw$, let 
$\mproj_{\vw} \defeq \mw^{1/2}\ma^\top(\ma\mw\ma^\top)^{-1}\ma\mw^{1/2}$, 
and suppose there exists dynamic tree and inverse tree operators $\treeop$ and $\itreeop$  dependent on $\vw$, 
such that $\mw^{1/2} \mproj_{\vw}  = \treeop \itreeop$.
Let $U$ be the update complexity of $\treeop$ and $\itreeop$, and let $Q$ be their query complexity.
Let $r$ and $R = \norm{\vu-\vl}_2$ be the inner and outer radius of $\mathcal{P}$, and let $L = \norm{\vc}_2$. 
Then, there is a data structure to solve \cref{eq:LP-main} to $\eps L R$ accuracy with probability $1-2^{-m}$ in time
\[
\O \left( \eta^4 \sqrt{m} \log( \frac{R}{\eps r}) \cdot \sum_{\ell=0}^{\frac12 \log m} \frac{U(2^{2\ell}) + Q(2^{2\ell})}{2^\ell} \right).
\]
\end{theorem}

\begin{proof}[Proof of \cref{thm:ripm-main}]
We implement the IPM algorithm using the data structures from \cref{subsec:overview_representation,subsec:overview_approx},
and bound the cost of each operations of the data structures.
For simplicity, we only discuss the primal variables in this proof, but the slack variables are analogous.
We use one copy of \textsc{MaintainRep} to maintain $\vx$, and one copy of \textsc{MaintainApprox} to maintain $\ox$.
At each step, we perform the implicit update of $\vx$ using \textsc{Move} and update $\vw$ using \textsc{Reweight} in \textsc{MaintainRep}.
We construct the explicit approximations $\ox$ using \textsc{Approximate} in \textsc{MaintainApprox}.

\cref{thm:soln-approx} shows that throughout the IPM, for each $\ell \geq 0$, there are $2^{2\ell}$ coordinate changes to $\ox$ every $2^\ell$ steps. 
Since $\vw$ is a function of $\ox$ coordinate-wise, there are also $2^{2\ell}$ coordinate changes in $\vw$ every $2^\ell$ steps. 
Similarly, we observe that $\vv$ is defined as a function of $\ox$ and $\os$ coordinate-wise,  so there are $O(2^{2\ell})$ coordinate changes to $\vv$ every $2^\ell$ steps.
Then \cref{thm:maintain_representation} shows that the total runtime over $N$ steps for the \textsc{MaintainRep} data structure is
\begin{equation} \label{eq:maintainrep-overall-runtime}
\O(U(m) + Q(m)) + \O \left(  \sum_{\ell=0}^{\log N} \frac{N}{2^\ell} \cdot \left( U(\eta \cdot 2^{2\ell}) + Q(\eta \cdot 2^{2\ell}) \right)  \right).
\end{equation}

\cref{thm:soln-approx} shows that the total runtime over $N$ steps for \textsc{MaintainApprox} is
\begin{equation} \label{eq:maintainapprox-overall-runtime}
	\O(\eta^3 (\beta/\delta)^2 \log^3(mN/\rho)) \left( Q(m) + \sum_{k=1}^N Q(S^{(k)}) +  \sum_{\ell=0}^{\log N} \frac{N}{2^\ell} \cdot Q(2^{2\ell}) \right),
\end{equation}
where the variables are defined as in the theorem statement.
By examining \cref{thm:maintain_representation}, we see that when a coordinate of $\vw$ or $\vv$ changes, the implicit representation of $\vx$ admits updates at $O(\eta)$-many nodes.
Combined with the concavity of $Q$, we can bound
\[
	\sum_{k=1}^N Q(S^{(k)}) \leq O(\eta) \cdot \sum_{\ell = 0}^{ \log N} \frac{N}{2^\ell} \cdot Q(2^{2\ell}).
\]

\cref{thm:IPM} guarantees that there are $N=\sqrt{m}\log m\log(\frac{mR}{\eps r})$ total IPM steps, and at each step $k$, we have $h^{(k)} \norm{{\mw^{(k)}}^{-1/2} ( \vx^{(k)} - \vx^{(k-1)}) }_2 = h^{(k)} \norm{\vv^{(k)} - \mproj_{\vw} \vv^{(k)} }_2 \leq O(\frac{1}{\log m})$, so we can set $\beta = O(\frac{1}{\log m})$. 
By examining \cref{alg:IPM_centering}, we see it suffices to set $\delta = O(\frac{1}{\log m})$. 
We choose the failure probability $\rho$ to be appropriately small, e.g.\ $2^{-m}$.
Finally, we conclude that the overall runtime of the IPM framework is
\[
\O \left( \eta^4 \sqrt{m} \log( \frac{R}{\eps r}) \cdot \sum_{\ell=0}^{\frac12 \log m} \frac{U(2^{2\ell}) + Q(2^{2\ell})}{2^\ell} \right),
\]
where the terms for intialization times have been absorbed.

\end{proof}


	\section{From separator tree to projection operators} 

In this section, we explore the separable structure of the dual graph $G_{\ma}$ of the LP constraint matrix $\ma$, 
and use these properties to help define and maintain the tree operator and inverse tree operator as needed for the IPM framework from \cref{sec:ipm-framework}.

Throughout this section, we fix $\ma \in \R^{n \times m}$, so that the dual graph $G_{\ma} = (V,E)$ has $n$ vertices, $m$ hyperedges. 
Additionally, let $\rho$ denote the max hyperedge size in $G_{\ma}$; equivalently, $\rho$ is the column sparsity of $\ma$.

\subsection{Separator tree}

The notion of using a \emph{separator tree} to represent the recursive decomposition of a separable graph is well-established in literature, c.f \cite{eppstein1996separator, henzinger1997faster}.
In our work, we use the following definition:

\begin{definition} [Separator tree] \label{defn:separator-tree}
	Let $G$ be a hypergraph with $n$ vertices, $m$ hyperedges, and max hyperedge size $\rho$.
	A \emph{separator tree} $\cs$ for $G$ is a \emph{constant-degree} tree whose nodes represent a recursive decomposition of $G$ based on balanced separators.
	
	Formally, each node of $\cs$ is a \emph{region} (edge-induced subgraph) $\region$ of $G$; we denote this by $\region \in \cs$. 
	At a node $\region$, we define subsets of vertices $\bdry{\region}, \sep{\region}, \elim{\region}$, 
	where $\bdry{\region}$ is the set of \emph{boundary vertices} of $H$, i.e.\ vertices with neighbours outside $\region$ in $G$;
	$\sep{\region}$ is a balanced vertex separator of $\region$;
	and $\elim{\region}$ is the set of \emph{eliminated vertices} at $\region$. 
	Furthermore, let $E(H)$ denote the edges contained in $H$.

	The nodes and associated vertex sets are defined in a top-down manner as follows: 
	
	\begin{enumerate}
		\item The root of $\cs$ is the node $\region = G$, with $\bdry{\region} = \emptyset$ and $\elim{\region} = \sep{\region}$.
		\item A non-leaf node $\region \in \cs$ has a constant number of children whose union is $H$.
		The children form a edge-disjoint partition of $\region$, and the intersection of their vertex sets is a balanced separator $\sep{\region}$ of $\region$.
		Define the set of eliminated vertices at $H$ to be $\elim{\region} \defeq \sep{\region} \setminus \bdry{\region}$.
		
		The set $\skel{H}$ consists of all vertices in the boundary and separator, which can intuitively be interpreted as the \emph{skeleton} of $H$.  In later sections, we recursively construct graphs (matrices) on $\skel{H}$ which capture compressed information about all of $H$.
		
		By definition of boundary vertices, for a child $D$ of $H$, we have $\bdry{D} \defeq (\bdry{\region} \cup \sep{\region}) \cap V(D)$.
		
		\item At a leaf node $H$, we define $\sep{\region} = \emptyset$ and $\elim{\region} = V(\region) \setminus \bdry{\region}$. (This convention allows leaf nodes to exist at different levels in $\cs$.)
		The leaf nodes of $\cs$ partition the edges of $G$.
		
	\end{enumerate}
	
	We use $\eta$ to denote the height of $\cs$.
\end{definition}

For a separator tree to be meaningful, the leaf node regions should be sufficiently small, to indicate that we have a good overall decomposition of the graph.
Additionally, for our work, we want a more careful bound on the sizes of the skeleton of regions.
This motivates the following refined definition:
\begin{definition}[$(a, b, \lambda)$-separator tree]
	Let $G$ be a graph with $n$ vertices, $m$ edges, and max hyperedge size $\rho$.
	Let $a \in [0,1]$ and $b \in (0,1)$ be constants, and $\lambda \geq 1$ be an expression in terms of $m,n,\rho$.
	An $(a, b, \lambda)$-separator tree $\cs$ for $G$ is a separator tree satisfying the following additional properties:
	\begin{enumerate}
		\item There are at most $O(b^{-i})$ nodes at level $i$ in $\cs$,
		\item any node $H$ at level $i$ satisfies $|\skel{H}| \leq O(\lambda \cdot b^{a i})$,
		\item a node $H$ at level $i$ is a leaf node if and only if $|V(H)| \leq O(\rho)$.
	\end{enumerate}
\end{definition}

\begin{figure}
\includegraphics[scale=0.5]{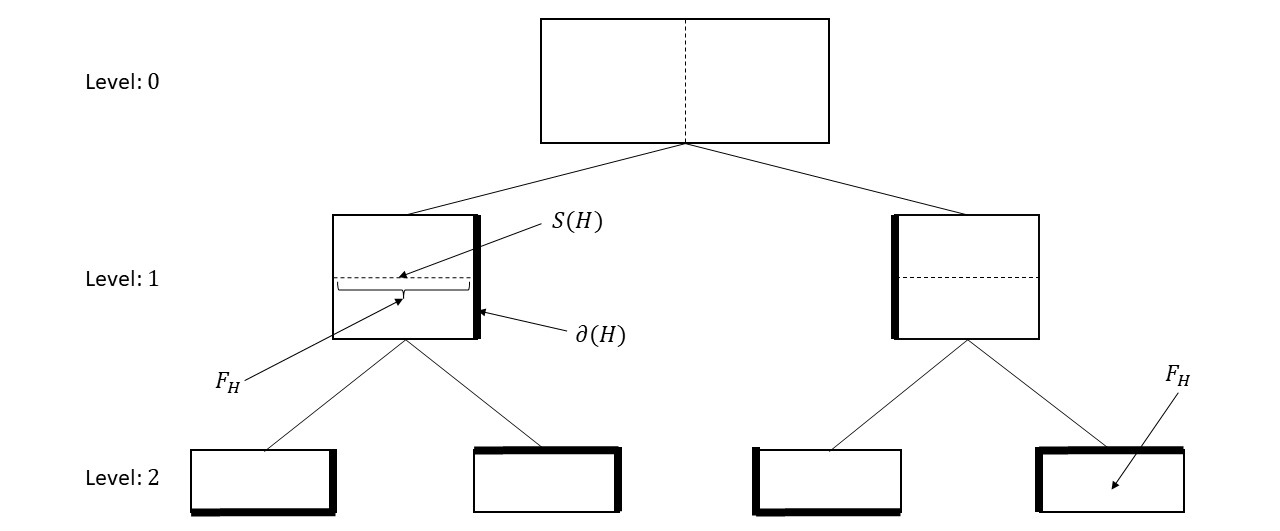}
\caption{An example of a separator tree. The bold edges denote the boundary of each component, $\protect\bdry{\region}$ while the dotted lines denote the separators $\protect\sep{\region}$. Note that $\protect\elim{\region} = \sep{\region} \setminus \bdry{\region}$ is defined differently on the leaves. }
\end{figure}

Intuitively, $a$ and $b$ come from the separability parameters of $G$, and $\lambda$ is a scaling factor for node sizes in $\cs$. 
Since there could be hyperedges of size $\rho$, regions of size $\rho$ are not necessarily separable, so we set the region as a leaf.

We make extensive use of these properties in subsequent sections when computing runtimes.

\subsection{Nested dissection using a separator tree} \label{subsec:nested-dissection}

Let $\cs$ be any separator tree for $G_{\ma}$. 
In this section, we show how to use $\cs$ to factor the matrix $\ml^{-1} \defeq (\ma \mw \ma^\top)^{-1}$  recursively:

\begin{definition}[Block Cholesky decomposition]
	The \emph{block Cholesky decomposition} of a symmetric matrix $\ml$ with blocks indexed by $F$ and $C$ is:
	\begin{equation}\label{eq:basic_chol}
		\ml
		= 
		\left[
		\begin{array}{cc}
			\mi
			& \mzero \\
			\ml_{C,F} (\ml_{F,F})^{-1}
			& \mi
		\end{array}
		\right] 
		\left[
		\begin{array}{cc}
			\ml_{F,F}
			& \mzero \\
			\mzero
			& \sc(\ml, C)
		\end{array}
		\right] 
		\left[
		\begin{array}{cc}
			\mi
			& (\ml_{F,F})^{-1}\ml_{F,C} \\
			\mzero & \mi
		\end{array}\right],
	\end{equation}
	where the middle matrix in the decomposition is a block-diagonal matrix with blocks indexed by $F$ and $C$, with the lower-right block being the \emph{Schur complement} $\sc(\ml,C)$ of $\ml$ onto $C$:
	\begin{equation} \label{eq:schur-complement}
	\sc(\ml,C) \defeq \ml_{C,C} - \ml_{C,F}\ml_{F,F}^{-1}\ml_{F,C}. 
	\end{equation}
\end{definition}

Since $\sc(\ml,C)$ is a symmetric matrix, we can recursively apply the decomposition \cref{eq:basic_chol} to it.
By choosing the index sets $F,C$ for each recursive step according to $\cs$, we get a recursive decomposition of $\ml^{-1}$:

\begin{theorem}[$\ml^{-1}$ factorization, c.f.~\cite{dong2022nested} Theorem 33] \label{thm:L-inv-factorization}
	
	Let $\cs$ be the separator tree of $G_{\ma}$ with height $\eta$.
	For each node $H \in \cs$ with hyperedges $E(H)$, let $\ma_H \in \R^{n \times m}$ denote the matrix $\ma$ restricted to columns indexed by $E(H)$. 
	Define 
	\begin{align}\label{eq:L^H-defn}
		\ml[H] &\defeq \ma_H \mw {\ma_H}^\top, \text{and} \\
		\ml^{(\region)} & \defeq \sc(\ml[\region], \skel{\region}).
	\end{align}
	Then, we have
	\begin{equation} \label{eq:L-inv-factorization}
		\ml^{-1} =
		\mpi^{(\eta)\top}\cdots\mpi^{(1)\top} \mga
		\mpi^{(1)}\cdots\mpi^{(\eta)},
	\end{equation} 
	where\footnote{We use a different definition of \emph{level} compared to \cite{dong2022nested}. In \cite{dong2022nested}, the root has level $\eta$ in and leaf nodes have level 0, and in this paper, the root has level $0$ and leaf nodes have level $\eta$. This is purely for notational convenience in later calculations, so this theorem is otherwise unaffected.}
	\begin{equation}
	\mga \defeq 
	\left[
	\begin{array}{cccc}
		\left( \sum_{H \in \ct(\eta)} \left(\ml^{(H)}_{F_H, F_H}\right)^{-1}\right) & \mzero & \mzero\\
		\mzero &  \ddots & \mzero\\
		\mzero & \mzero & \left( \sum_{H \in \ct(0)} \left(\ml^{(H)}_{F_H, F_H}\right)^{-1}\right)
	\end{array}
	\right],
	\end{equation}
	and for $i = 1, \dots, \eta$,
	\begin{equation}
		\mpi^{(i)} \defeq \mi - \sum_{H \in \ct(i)} \mx^{(H)},
	\end{equation}
	where $\ct(i)$ is the set of nodes at level $i$ in $\ct$, the matrix $\mpi^{(i)}$ is supported on $\bigcup_{H \in \ct(i)} \skel{H}$ and padded with zeros to $n$-dimensions,
	and for each $H \in \cs$,
	\begin{equation}\label{def:mx^(H)}
		\mx^{(H)} \defeq \ml^{(H)}_{\bdry{\region}, F_H} \left( \ml^{(H)}_{F_H, F_H}\right)^{-1}.
	\end{equation}
\end{theorem}

\subsection{Projection operators definition} \label{subsec:proj-op-defn}

Suppose $\cs$ is a separator tree for $G_\ma$.
In this subsection, we define the operator tree $\ct$ based on $\cs$,
followed by the tree operator $\treeop$ and inverse tree operator $\itreeop$ which will be supported on $\ct$.
Finally, we will show that our definitions indeed satisfy
\[
	\mw^{1/2} \mproj_{\vw} = \treeop \itreeop.
\]

Recall that $\cs$ is a constant-degree tree.
The leaf nodes of $\cs$ partition the hyperedges of $G_{\ma}$,
however, we do not have a bound on the number of hyperedges in a leaf node.
In constructing $\ct$, we simply want to modify $\cs$ so that each leaf contains exactly one hyperedge.
Specifically, for each leaf node $H \in \cs$ containing $|E(H)|$ hyperedges,
we construct a complete binary tree $\ct^+_H$ rooted at $H$ with $|E(H)|$ leaves,
assign one hyperedge from $E(H)$ to one new leaf, and attach $\ct^+_H$
at the node $H$.
This construction yields the desired operator tree $\ct$ whose height is within a $\log{|E|}$ factor of $\cs$.

We define the tree operator $\treeop$ on $\ct$ follows: For non-root node $H$ in $\ct$, let
\begin{equation} \label{eq:treeop-defn}
	\treeop_{H} \defeq
	\begin{cases}
		\mi_{\skel{H}} - \mx^{(H) \top} &\text{if $H$ exists in $\cs$}\\
		{\mw_{E(H)}}^{1/2} {\ma_H}^\top &\text{if $H$ is a leaf node in $\ct$} \\
		\mi &\text{else}. \\  
	\end{cases}
\end{equation}
	Note that the first two cases are indeed disjoint by construction.
	We pad zeros to all matrices in order to arrive at the correct overall dimensions.
\begin{lemma}[c.f. \cite{dong2022nested}, Lemma 59]
	Let $\treeop$ be the tree operator as defined above.
	Then
	\begin{equation}
		\treeop = \mw^{1/2} \ma^\top \mpi^{(\eta)\top} \cdots \mpi^{(1)\top}.
	\end{equation}
\qed
\end{lemma}

Next, we establish the query complexity of the tree operator:
\begin{lemma}	\label{lem:treeop-query-complexity}
	Suppose $L$ is the total number of leaf nodes in $\cs$.
	The query complexity of $\treeop$ is 
	\[
		Q(K) = O \left( \rho K + \max_{ \collN : \text{set of $K$ leaves in $\cs$}} \sum_{H \in \mathcal{P}_{\cs}(\collN)} |\skel{H}|^2 \right)
	\]
	for $K \leq L$,
	where $\mathcal{P}_{\cs}(\collN)$ is the set of all nodes in $\cs$ that are ancestors of some node in $\collN$ unioned with $\collN$.
	When $K > L$, then we define $Q(K) = Q(L)$.
\end{lemma}
\begin{proof}
	First, we consider the query time $Q(1)$ for a single edge.
	Let $\vu$ be any vector, and let $H$ be a non-root node in $\ct$. 
	If $H$ is a leaf node, 
	then computing $\treeop_{H} \vu$ and $\vu^\top \treeop_{H}$ both take $O(\rho)$ time.
	If $H$ exists in $\cs$, then
	computing $\treeop_{H} \vu$ takes
	$O \left( |F_H|^2+ |\partial{H}| |F_H| \right) \leq O \left( |\skel{H}|^2 \right)$ time,
	since the bottleneck is naively computing $\ml^{(H)}_{\bdry{\region}, F_H} \left( \ml^{(H)}_{F_H, F_H}\right)^{-1} \vu$.
	Therefore, $Q(1) = O \left( \rho + \max_{H \in \cs} |\skel{H}|^2 \right)$.
	
	For $K > 1$, we can simply bound the query time for $K$ distinct edges by
	\begin{align*}
		Q(K) = O \left( \rho K + \max_{\collN : \text{set of $K$ nodes in $\cs$}} \sum_{H \in \mathcal{H}} |\skel{H}|^2 \right).
	\end{align*}
	Finally, note that we can take the summation over $H \in \mathcal{P}_{\cs}(\collN)$ instead of $H \in \collN$ for an upper bound. 
	In this case, it suffices to take the max over sets of leaf nodes.
\end{proof}

By taking the transpose of $\treeop$, we get an inverse tree operator, and together, they give the projection matrix using \cref{eq:L-inv-factorization}.

\begin{corollary}
	Let $\itreeop \defeq \treeop^\top$ be the inverse tree operator obtained from $\treeop$ by transposing the edge and leaf operators. Then
	\begin{equation}
		\mw^{1/2} \mproj_{\vw} \defeq \mw^{1/2} \mw^{1/2} \ma^\top \ml^{-1} \ma \mw^{1/2} = (\mw^{1/2} \treeop) \mga \itreeop.
	\end{equation}
	\qed
\end{corollary}

\begin{remark} \label{rem:treeop-key-remark}
	Without loss of generality, we have chosen to simplify our presentation and consider $\treeop \itreeop$ in place of $\mw^{1/2} \treeop \mga \itreeop$.
	
	This is possible for two reasons: One, $\mw^{1/2} \treeop$ is a tree operator, which we can in fact maintain in the same time complexity as $\treeop$.
	Two, $\mga$ is a block-diagonal matrix, with a block for each $H \in \ct$ that is indexed by $F_H$. It is straightforward to show we can maintain and apply $\mga \itreeop$ in the same time complexity as $\itreeop$.
\end{remark}

\subsection{Maintenance of projection operators}

So far, we have defined the separator tree $\cs$ for the graph $G_{\ma}$,
which we then used to define the operator tree $\ct$, which supports the tree operator $\treeop$ needed for the IPM framework.
In this subsection, we discuss how to maintain $\ml^{(\region)}, \sc(\ml^{(H)}, \bdry{H})$, and $(\ml^{(H)}_{F_H, F_H})^{-1}$ at each node $H \in \cs$ using the data structure \textsc{DynamicSC} (\cref{alg:dynamicSC}), as the weight vector $\vw$ undergoes changes throughout the IPM.
This will in turn allow us to maintain the tree operator $\treeop$.

We begin with a lemma showing that given a symmetric matrix and a low-rank update, we can compute its new inverse and Schur complement quickly. 

\begin{lemma} \label{lem:low-rank-sc-update}
	Let $\ml' = \ml +\mmu\mv \in \R^{n \times n}$ be a symmetric matrix plus a rank-$K$ update, where $\mmu$ and $\mv^\top$ both have dimensions $n \times K$.
	Given $\ml', \mmu, \mv$, we can compute ${\ml'}^{-1}$ in $O(n^2 K^{\omega-2})$ time.
	
	Additionally, suppose we are also given $\ml^{-1}$ and $\sc(\ml,S)$ for an index set $S$.
	Then we can compute $\sc(\ml', S)$, $\mmu', \mv'$ in $O(n^2K^{\omega-2})$ time, 
	so that $\sc(\ml,S) + \mmu'\mv' = \sc(\ml',S)$, and $\mmu', \mv'^\top$ both have $K$ columns.
\end{lemma}

\begin{proof}
	The Sherman-Morrison formula states
	\[
		{\ml'}^{-1} = \ml^{-1} -\ml^{-1} \mmu (\mi_K + \mv \ml^{-1} \mmu)^{-1} \mv \ml^{-1}.
	\]
	The time to compute this update is dominated by the time required to multiply an $n \times n$ matrix with an $n \times K$ matrix, which is $O(n^2 K^{\omega-2})$.
	
	For the second part of the lemma, recall that the Schur complement is defined to be: 
	\begin{equation} \label{eq:schur-complement}
	\sc(\ml,C) \defeq \ml_{C,C} - \ml_{C,F}\ml_{F,F}^{-1}\ml_{F,C}. 
	\end{equation}
	If we were to naively use this definition of the Schur complement to perform the updates and construct $\mmu'$ and $\mv'^\top$, we will run into an issue where the rank of the new update blows up by a factor of 8, leading to an exponential blowup in the rank as we go up the levels recursively. Instead, we make use of the fact that the inverse of the Schur complement, $\sc(\ml,S)^{-1}$ is exactly the $S,S$ submatrix of $\ml^{-1}$ to control the rank of the updates. 

	We first apply the definition of Schur complement and then use the Sherman-Morrison formula to get
	\begin{align*}
		\sc(\ml',S)^{-1} &= {\ml'^{-1}}_{S,S}\\
		&= {\ml^{-1}}_{S,S} - \left(\ml^{-1} \mmu (\mi_K + \mv\ml^{-1}\mmu )^{-1}\mv\ml^{-1}\right)_{S,S} \\
		&= \sc(\ml,S)^{-1} - \mi_S \ml^{-1}\mmu (\mi_{K}+ \mv\ml^{-1}\mmu )^{-1} \mv \ml^{-1} \mi_S.
	\end{align*}
	This gives us the new rank-$K$ update $\sc(\ml',S)^{-1} = \sc(\ml,S)^{-1} + \mmu^* \mv^*$ with
	\begin{equation*}
		\begin{aligned}
			\mmu^* &= - \mi_S \ml^{-1} \mmu\\
			\mv^* &= (\mi_{K}+ \mv\ml^{-1}\mmu)^{-1}\mv\ml^{-1} \mi_S.
		\end{aligned}
	\end{equation*}
	We can now determine the Schur complement update by applying Sherman-Morrison again:
	\[
		\sc(\ml', S) = \sc(\ml, S) - \sc(\ml, S)\mmu^*(\mi_{K}+ \mv^*\sc(\ml,S) \mmu^*)^{-1} \mv^* \sc(\ml,S).
	\]
	This is a rank-$K$ update $\sc(\ml',S) = \sc(\ml,S) + \mmu' \mv'$ with
	\begin{equation*}
		\begin{aligned}
			\mmu' &= - \sc(\ml, S)\mmu^* \\
			\mv' &= (\mi_{K}+ \mv^*\sc(\ml,S) \mmu^*)^{-1} \mv^* \sc(\ml,S).
		\end{aligned}
	\end{equation*}
	 The time to compute $\mmu^*, \mv^*, \mmu', \mv'$ are all dominated by the time to multiply an $n \times n$ matrix with an $n \times K$ matrix, which is $O(n^2 K^{\omega-2})$.
\end{proof}

Now, we are ready to present the data structure for maintaining the Schur complement matrices along a separator tree.

\begin{algorithm}
	\caption{Data structure to maintain dynamic Schur complements} \label{alg:dynamicSC}
	\begin{algorithmic}[1]
		\State \textbf{data structure} \textsc{DynamicSC}
		\State \textbf{private: member}
		\State \hspace{4mm} Hypergraph $G_{\ma}$ with incidence matrix $\ma$
		\State \hspace{4mm} $\vw \in \R^m$: Dynamic weight vector
		\State \hspace{4mm} $\cs$: Separator tree of height $\eta$. Every node $H$ of $\cs$ stores:
		\State \hspace{8mm} $\elim{\region}$, $\bdry{\region}$: Sets of eliminated vertices and boundary vertices of region $H$
		\State \hspace{8mm} $E(H)$: Set of hyperedges of region $H$ 
		\State \hspace{8mm} $\ml^{(H)}, (\ml^{(H)}_{F_H, F_H})^{-1}, \sc(\ml^{(H)}, \bdry{\region}), $: Matrices to maintain as a function of $\vw$
		\State \hspace{8mm} ${\ml^{(H)}}^{-1}$: Additional inverse matrix to maintain as a function of $\vw$
		\State \hspace{8mm} $\mmu_H, \mv_H$: Low-rank update at $H$, used in \textsc{Reweight}
		\State
		
		\Procedure{\textsc{Initialize}} {$\cs$, $\vw^\init \in \mathbb{R}^m$}
		\State $\cs \leftarrow \cs, \vw \leftarrow \vw^\init$
		\For{level $i=\eta$ to $0$}
		\For{each node $H$ at level $i$}
		\State $\ml^{(H)}, (\ml^{(H)}_{F_H, F_H})^{-1}, \sc(\ml^{(H)}, \bdry{\region}) \leftarrow \mzero, \mzero, \mzero$
		\State \textsc{SchurNode}$(H, \vw)$
		\EndFor
		\EndFor
		\EndProcedure
		\State
		
		\Procedure{\textsc{Reweight}}{$\delta_{\vw} \in \R^{m}$}
		\State $\mathcal{H} \leftarrow$ set of nodes $H$ in $\cs$ where $\delta_{\vw}|_{E(H)} \neq \vzero$
		\For{level $i = \eta$ to $0$}
			\For{each node $\region \in \mathcal{H}$ at level $i$}
				\State \textsc{SchurNode}$(\region, \delta_{\vw})$
			\EndFor
		\EndFor
		\State $\vw \leftarrow \vw + \delta_{\vw}$
		\EndProcedure
		\State
		
		\Procedure{\textsc{SchurNode}} {$\region \in \cs$, $\delta_{\vw} \in \R^m$}
		\If{$H$ is a leaf node} \label{line:sc-case-0}
			\Comment rank of update $\leq \min \{ \nnz(\delta_{\vw}|_{E(H)}), |\skel{H}|\}$
			\State $\ml^{(H)} \leftarrow \ml^{(H)} + \ma_H \diag(\delta_{\vw}|_{E(H)}) \ma_H^\top$
		\ElsIf{$\nnz(\delta_{\vw}|_{E(H)}) \leq |\skel{H}|$} \label{line:sc-case-1}
			\Comment rank of update $\leq \sum_{\text{child $D$}} K_D \leq \nnz(\delta_{\vw}|_{E(H)})$
			\State $\ml^{(H)} \leftarrow \ml^{(H)} + \sum_{\text{child $D$ of $H$}} \mmu_D \mv_D$
		\Else \label{line:sc-case-2}
			\Comment rank of update $\leq |\skel{H}|$
			\State $\ml^{(H)} \leftarrow \sum_{\text{child $D$ of $H$}} \sc(\ml^{(D)}, \bdry{D})$
		\EndIf
		\State Let $K_H \defeq \min \{ \nnz(\delta_{\vw}|_{E(H)}), |\skel{H}|\}$
		\Comment upper bound on the rank of update to $\ml\atH$
		\State Compute $(\ml^{(H)}_{F_H,F_H})^{-1}$ and ${\ml^{(H)}}^{-1}$ by \cref{lem:low-rank-sc-update}
		\State Compute $\sc(\ml^{(H)}, \bdry{H})$ and its rank-$K_H$ update factorization $\mmu_H, \mv_H$ by \cref{lem:low-rank-sc-update} 
		\EndProcedure
	\end{algorithmic}
\end{algorithm}

\begin{lemma} \label{lem:dynamicSC}
	Let $\vw$ be the weights changing at every step of the IPM.
	Let $\cs$ be any separator tree for $G_{\ma}$.
	Recall $G_{\ma}$ has $n$ vertices, $m$ hyperedges, and max hyperedge size $\rho$.
	Then the data structure \textsc{DynamicSC} (\cref{alg:dynamicSC}) correctly maintains the matrices $\ml^{(H)}, (\ml^{(H)}_{F_H, F_H})^{-1},\sc(\ml^{(H)}, \bdry{H})$ at every node $H \in \cs$ dependent on $\vw$ throughout the IPM.
	The data structure supports the following procedures and runtimes:
	\begin{itemize}
		\item $\textsc{Initialize}(\cs, \vw^\init \in \R^m)$: Set $\vw \leftarrow \vw^\init$, and compute all matrices with respect to $\vw$, in time
		\[
		O \left( \sum_{\text{leaf } H \in \cs} |E(H)| \cdot |\skel{H}|^{\omega-1} + \sum_{H \in \cs} |\skel{H}|^\omega \right).
		\]
		\item $\textsc{Reweight}(\delta_{\vw} \in \R^m)$: Update the weight vector to $\vw \leftarrow \vw + \delta_{\vw}$, and update all the maintained matrices with respect to the new weights, in time
		\[
		O \left( \sum_{\text{leaf } H \in \collN} 	\nnz(\delta_{\vw}|_{E(H)}) \cdot |\skel{H}|^{\omega-1}
		+ \sum_{H \in \collN} |F_H \cup \bdry{H}|^2 \cdot K_H^{\omega-2} \right).
		\]
		where $\collN$ is the set of nodes $H$ with $\delta_{\vw}|_{E(H)} \neq \vzero$,  $K_H \defeq \min \{\nnz(\delta_{\vw}|_{E(H)}), |\skel{H}|\}$.
	\end{itemize}
\end{lemma}

\begin{proof} 
	\textsc{Initialize} is a special case of \textsc{Reweight}, where the change in the weight vector is from $\vzero$ to $\vw^\init$, so we focus on a single call of \textsc{Reweight}.
	
	It suffices for \textsc{Reweight} visits only nodes in $\collN$,
	since if none of the edges in a region admits a weight update, then the matrices stored at the node remain the same by definition. 
	Also note that $H \in \mathcal{H}$ implies all ancestors of $H$ are also in $\mathcal{H}$.
	
	\paragraph{Correctness.} 
	We use the superscript $^\new$ on $\ml\atH$ to indicate that it is computed with respect to the new weights, and $^\old$ otherwise.
	Recall that $\ml\atH$ is supported on $\skel{H}$.
	
	We maintain some additional matrices at each node, in order to efficiently compute low-rank updates.
	Specifically, we use helper matrices $\mmu_H, \mv_H$ at $H$, and guarantee that during a single $\textsc{Reweight}(\delta_{\vw})$ call, after $\textsc{SchurNode}(H, \delta_{\vw})$ is run,
	they satisfy $\sc({\ml\atH}^\old, \bdry{H}) = \sc({\ml\atH}^\new, \bdry{H}) + \mmu_H \mv_H$, and $\mmu_H, {\mv_H}^\top$ both have at most $K_H$-many columns.
	
	Now, we show inductively that after \textsc{SchurNode}$(H, \delta_{\vw})$ is run, all matrices at $H$, as well as all matrices at all descendants of $H$, are updated correctly:
	When $H$ is leaf node, recall $\ml\atH$ is defined to be $\ml[H] \defeq \ma_H \mw_{E(H)} {\ma_H}^\top$, so clearly \textsc{SchurNode} updates $\ml\atH$ correctly, and the rank of the update is at most $K_H$. 
	The remaining matrices at $H$ are computed correctly by \cref{lem:low-rank-sc-update}.
	
	Inductively, when $H$ is a non-leaf node, the recursive property of Schur complements (c.f. \cite[Lemma 18]{dong2022nested}) allows us to write ${\ml\atH}^\new = \sum_{\text{child $D$ of $H$}} \sc({\ml^{(D)}}^\new, \bdry{D})$ at every node $H \in \cs$.
	This formula trivially shows that the update ${\ml\atH}^\new - {\ml\atH}^\old$ has rank $|\skel{H}|$ (ie. full-rank).
	Alternatively, if $\nnz(\delta_{\vw}|_{E(H)}) \leq |\skel{H}|$, then by the guarantees on the helper matrices, we have
	\begin{align*}
		{\ml\atH}^\new &= \sum_{\text{child $D$ of $H$}} \sc({\ml^{(D)}}^\new, \bdry{D}) \\
		&= \sum_{\text{child $D$ of $H$}} \sc({\ml^{(D)}}^\old, \bdry{D}) + \mmu_D \mv_D \\
		&={\ml\atH}^\old + \sum_{\text{child $D$ of $H$}} \mmu_D \mv_D.
	\end{align*} 
	This gives a low-rank factorization of the update ${\ml\atH}^\new - {\ml\atH}^\old$ with rank at most $\sum_{\text{child $D$}} K_D$, which we can show by induction is at most $\nnz(\delta_{\vw}|_{E(H)})$.
	Since we have the correct low-rank update to $\ml^{(H)}$, the remaining matrices at $H$ again are computed correctly by \cref{lem:low-rank-sc-update}.
	
	This completes the correctness proof.
	
	\paragraph{Runtime.}
	
	Consider the runtime of the procedure $\textsc{SchurNode}(H, \delta_{\vw})$ at a node $H$:
	If $H$ is a leaf node, then computing the update to $\ml^{(H)}$ involves multiplying a $|\skel{H}| \times \nnz(\delta_{\vw}|_{E(H)})$-sized matrix with its transpose (\cref{line:sc-case-0}).
	Note that if $|\skel{H}| > \nnz(\delta_{\vw}|_{E(H)})$, then this runtime can be absorbed into the runtime expression for the remaining steps of the procedure, since $K_H = \nnz(\delta_{\vw}|_{E(H)})$.
	Otherwise, we use fast matrix multiplication which takes $O(\nnz(\delta_{\vw}|_{E(H)}) \cdot |\skel{H}|^{\omega-1})$ time.
	If $H$ is a non-leaf node, there are two cases for the update to $\ml^{(H)}$ in the algorithm. 
	The first case (\cref{line:sc-case-1}) takes $O\left(|\skel{H}| \cdot (\sum K_D) \right) \leq O(|\skel{H}| \cdot K_H)$ time, and the second case (\cref{line:sc-case-2}) takes $O(|\skel{H}|^2)$ time.
	Computing the other matrices at any node $H$ takes $O \left( |\skel{H}|^2 \cdot {K_H}^{\omega-2}\right)$ time by \cref{lem:low-rank-sc-update}.
	
	The runtime of $\textsc{Reweight}(\delta_{\vw})$ is therefore given by
	\begin{align*}
		&\phantom{{}={}} \sum_{H \in \collN} \textsc{SchurNode}(H, \delta_{\vw}) \text{ time} \\
		&= O \left( \sum_{\text{leaf } H \in \collN}  \nnz(\delta_{\vw}|_{E(H)}) \cdot |\skel{H}|^{\omega-1}
		+ \sum_{H \in \collN} |F_H \cup \bdry{H}|^2 \cdot K_H^{\omega-2} \right).
	\end{align*}
	For \textsc{Initialize}, we further simplify the expression using $\nnz(\delta_{\vw}|_{E(H)}) = |E(H)|$ and $K_H \leq |\skel{H}|$.
\end{proof}

\subsection{Projection operator complexities}

In this subsection, we summarize the runtime complexities for the tree operator, in the special case when $\cs$ is a $(a,b,\lambda)$-separator tree for $G_{\ma}$.
Parametrizing the separator tree this way allows us to write the runtime expressions using geometric series. For non-negative $x$, we use the standard bound $\sum_{i=\ell}^u x^i \leq O(x^\ell + x^u)$. When it is clear $x < 1$, we bound $\sum_{i=\ell}^u x^i \leq O(x^\ell)$.

\begin{lemma} \label{lem:a-b-lambda-tree-complexity}
	Suppose $\cs$ is an $(a,b,\lambda)$-separator tree for $G_{\ma}$ on $n$ vertices, $m$ edges, with max hyperedge size $\rho$,
	where $a \in [0,1]$ and $b \in (0,1)$.
	Let $\eta$ denote the height of $\cs$, and let $L$ denote the number of leaf nodes.
	Let $\treeop$ be the tree operator on $\ct$ as defined in \cref{subsec:proj-op-defn}.
	Then there is a data structure to maintain $\treeop$ as a function of the weights $\vw$ throughout \textsc{Solve}, so that:
	\begin{itemize}
		\item The data structure can be initialize in time
		\begin{equation}\label{eq:treeop-init-time}
		O\left( \rho^{\omega-1} m + \lambda^\omega \cdot \left(1 + (b^{a \omega - 1})^\eta \right) \right).
		\end{equation}
		\item The query complexity of $\treeop$ is
		\begin{equation} \label{eq:treeop-query-time}
		Q(K) = O \left( \rho K + \lambda^2 \left(1 + (\min \{K,L\})^{1 - 2a} \right) \right)
		\end{equation}
		\item When $a  < 1$, the update complexity of $\treeop$ is $U(K) = $
		\begin{equation} \label{eq:treeop-update-time}
		\rho^{\omega-1} K + 
		\lambda^2 \min\{K, \lambda\}^{\omega-2} + 
		\begin{cases}
			\lambda^2 K^{1-2a} &\text{if $K \leq \lambda$} \\
			\lambda^2 K^{1-2a} + 
			\lambda^{\frac{\omega-1}{1 - \alpha}} K^{\frac{1-\alpha \omega}{1-\alpha}} 
			&\text{if $\lambda < K \leq \lambda \cdot b^{(a-1)\eta}$} \\
		    \lambda^{\omega} \cdot b^{(a\omega-1) \eta} &\text{if $K > \lambda \cdot b^{(a-1)\eta}$}.\\
		\end{cases}
		\end{equation}
		When $a = 1$, the update complexity is $U(K) = 
		\rho^{\omega-1} K + \lambda^2 \min\{K, \lambda\}^{\omega-2}$.
	\end{itemize}
\end{lemma}

\begin{proof}
	The data structure we use to maintain $\treeop$ is precisely the data structure \textsc{DynamicSC} with respect to $\cs$.
	
	\paragraph{Initialization time.}
	We use the runtime expression for \textsc{Initialize} in \textsc{DynamicSC} (\cref{lem:dynamicSC}) combined with the parameters of the $(a,b,\lambda)$-separator tree. For any $H$, we have $|\skel{H}| \leq \rho$, so 
	$\sum_{\text{leaf } H \in \cs} |E(H)| \cdot |\skel{H}|^{\omega-1} \leq \rho^{\omega-1} m$. Moreover,
	\begin{align*}
		\sum_{H \in \cs} |\skel{H}|^\omega 
		\leq \sum_{i=0}^{\eta} b^{-i} \left(\lambda \cdot b^{a i} \right)^\omega 
		\leq O(\lambda^\omega) \cdot \left[1 + (b^{a \omega - 1})^\eta \right]. 
	\end{align*}

	\paragraph{Query complexity.}
	
	We substitute the $(a,b,\lambda)$-separator tree bounds in \cref{lem:treeop-query-complexity}, to conclude that the query complexity of $\treeop$ is
	\begin{align*}
	Q(K) &= O \left( \rho K + \sum_{H \in \mathcal{P}_{\cs}(\collN)} (\lambda \cdot b^{ai})^2 \right), \\
	\intertext{
		where $\collN$ is any set of $K$ leaf nodes in $\cs$. 
		We group terms according to their node level, 
		and note that there are $\min\{K, b^{-i}\}$-many terms at any level $i$, so
	}
	&= O \left(\rho K + \sum_{i=0}^\eta \min \{K, b^{-i}\} \cdot (\lambda \cdot b^{ai})^2 \right) \\
	&= O(\rho K) + O( \lambda^2) \cdot \left(
	\sum_{i=0}^{-\log_b K} b^{-i} \cdot b^{2 a i} + 
	K \cdot \sum_{i = -\log_b K}^\eta b^{2 a i}
	\right) \\
	\intertext{Note that $K$ can be at most $L$ in the summation, so we have}
	&= O \left( \rho K +  \lambda^2 (1 + (\min\{K,L\})^{1 - 2a})  \right).
\end{align*}

\paragraph{Update complexity.}
	When $\vw$ changes, we update $\treeop$ by invoking $\textsc{Reweight}(\delta_{\vw})$ in \textsc{DynamicSC}, where $\delta_{\vw}$ denotes the change in $\vw$.
	By \cref{lem:dynamicSC}, the runtime for the fixed $\delta_{\vw}$ is
	\begin{equation} \label{eq:dynamicSC-reweight-time}
	\sum_{\text{leaf $H$}:\; \delta_{\vw}|_{E(H)} \neq \vzero}  \nnz(\delta_{\vw}|_{E(H)}) \cdot |\skel{H}|^{\omega-1} + 
	\sum_{\text{node $H$}:\; \delta_{\vw}|_{E(H)} \neq \vzero} |\skel{H} |^2 \cdot {K_H}^{\omega - 2}.
	\end{equation}
	Hence, the update complexity of $\treeop$ is the max of the above expression taken over all choices of $\delta_{\vw}$. 
	For any leaf node $H$, we upper bound $|\skel{H}|^{\omega-1} \leq \rho^{\omega-1}$, and therefore the first summation is at most $\rho^{\omega-1} K$.
	
	For the second summation, we substitute in the $(a,b,\lambda)$-separator tree bounds, and group terms according to their node level.
	Let $\cs(i)$ denotes all nodes at level $i$ in $\cs$. Then for any $\delta_{\vw}$, we have
	\begin{equation} \label{eq:update-complexity}
		\sum_{\text{node $H$}:\; \delta_{\vw}|_{E(H)} \neq \vzero} |\skel{H} |^2 \cdot {K_H}^{\omega - 2} \leq \sum_{i=0}^\eta \left( (\lambda \cdot b^{ai})^2 \cdot \sum_{H \in \cs(i)} {K_H}^{\omega - 2} \right),
	\end{equation}
	where the $K_H$'s are non-negative integers satisfying $\sum_{H \in \cs(i)} K_H \leq K$ 
	and $K_H \leq |\skel{H}| \leq \lambda \cdot b^{ai}$ for $H \in \cs(i)$.
	We are interested in upper bounding \cref{eq:update-complexity}.
	At any level $i$, there are $b^{-i}$ nodes, and
	the sum is maximized when all the $K_H$'s are equal. Depending on the relationship between $K$ and the level $i$, we have the following three cases:
	\begin{itemize}
		\item If $K \leq b^{-i}$, that is, the total update rank is less than the number of nodes at the level, then the sum is maximized if $K_H = 1$ for $K$-many nodes, and $K_H = 0$ for the rest.
		\item If $b^{-i} < K \leq b^{-i} \cdot (\lambda \cdot b^{ai})$, the sum is upper bounded by setting $K_H = K / b^{-i}$.
		\item If $K > O(b^{-i}) \cdot (\lambda \cdot b^{ai})$,  the sum is upper bounded by setting $K_H = \lambda \cdot b^{ai}$.
	\end{itemize}
	Then, we can bound the summation term in \cref{eq:update-complexity} by
	\begin{align*}
		&\phantom{{}={}} \sum_{\substack{0 \leq i \leq \eta \\ K \leq b^{-i}}} K (\lambda \cdot b^{ai})^2 + 
		\sum_{\substack{0 \leq i \leq \eta :\\ b^{-i} < K \leq \lambda \cdot b^{(a-1)i}}} (\lambda \cdot b^{ai})^2 \cdot b^{-i} \cdot (Kb^i)^{\omega-2} 
		+ \sum_{\substack{0 \leq i \leq \eta : \\ K > \lambda \cdot b^{(a-1)i}}} b^{-i} \cdot (\lambda \cdot b^{ai})^{\omega} \\
		&\leq 
		\lambda^2 K  \sum_{i= -\log_b K}^\eta b^{2ai} + 
		\lambda^2 K^{\omega-2}  \sum_{i=\frac{\log_b(K/\lambda)}{a-1}}^{-\log_b K} b^{(2a + \omega-3)i} 
		+ \lambda^\omega \sum_{i=0}^{\frac{\log_b(K/\lambda)}{a-1}} b^{(a \omega - 1) i}.
	\end{align*}
	We need to further consider different cases for the possible values of $K$, which affects the summation indices.
	If $\log_b(K/\lambda) < 0$, i.e.\ $K < \lambda$, the expression simplifies to
	\[
		\lambda^2 K^{1-2a} + \lambda^2 K^{\omega-2}.
	\]
	If $0 \leq \log_b(K/\lambda)/(a-1) \leq \eta$, i.e.\ $\lambda \leq K \leq \lambda \cdot b^{(a-1)\eta}$, the expression simplifies to
	\[
	\lambda^2 K^{1-2a} + \lambda^{\frac{\omega-1}{1 - \alpha}} K^{\frac{1-\alpha \omega}{1-\alpha}} 
	+  \lambda^\omega.
	\]
	And lastly, if $\log_b(K/\lambda)/(a-1) > \eta$, i.e\ $K > \lambda \cdot b^{(a-1) \eta}$, the expression simplifies to
	\[
	\lambda^\omega +  \lambda^\omega \cdot b^{(a\omega-1)\eta}.
	\]
	We combine the cases to arrive at the overall update complexity, having implicitly assumed that $\alpha < 1$.
	When $\alpha = 1$, the summation in \cref{eq:update-complexity} is maximized when $K_H = \min\{K, \lambda\} \cdot b^i$ for $H$ at level $i$. Then we can upper bound the summation term by
	\[
	\sum_{i=0}^\eta (\lambda \cdot b^{i})^2 \cdot b^{-i} \cdot (\min\{K, \lambda\}  \cdot b^i)^{\omega-2} \leq O(\lambda^2 \min\{K, \lambda\}^{\omega-2}).
	\]
\end{proof}

	\section{Proofs of main theorems}

For our main theorems, it remains to show that we can construct an appropriate $(a,b,\lambda)$-separator tree for $G_{\ma}$ in each of the scenarios:
when $G_{\ma}$ is $n^\alpha$-separable; when $\ma$ is the constraint matrix for a planar $k$-multicommodity flow instance; and when $G_{\ma}$ has a tree decomposition of width $\tau$.
Then, we apply \cref{lem:a-b-lambda-tree-complexity} to the separator tree get the complexity of the tree operator, which we combine with \cref{thm:ripm-main} to conclude the overall IPM running times.

\subsection{Proof of \cref{thm:main}}
First, we show how to construct a separator tree for an $n^\alpha$-separable graph, by modifying the proof from~\cite{Frederickson87}.

\begin{lemma} \label{lem:separable-graph-tree}
	Suppose $G_{\ma}$ is a graph on $n$ vertices and $m$ edges. 
	If $G_{\ma}$ is $n^{\alpha}$-separable for $\alpha < 1$, then $G_{\ma}$ admits an $(\alpha, b, c n^\alpha)$-separator tree, where $b \in (0,1)$ and $c > 0$ are some constants.
	Furthermore, if a balanced vertex separator for $G_{\ma}$ can be computed in $T(n)$ time, then the separator tree can be computed in $\O(T(n))$ time.
\end{lemma}

\begin{proof}
	Let $b' \in (0,1)$ and $c'=1$ (without loss of generality) be the parameters for $G_{\ma}$ being $n^\alpha$-separable.
	In the separator tree construction process, 
	assume inductively that we have constants $b \in (0,1)$ and $c > 0$, both to be chosen later, 
	such that for any node $H$ at level $i$, we have $|V(H)| \leq b^i n$ and $|\bdry{H}| \leq c n^\alpha \cdot b^{\alpha i}$.
	In the base case at the root node, we have $i=0$, and $|V(G_{\ma})| \leq n$ and $|\bdry{G_{\ma}}| = 0 \leq c n^\alpha$.
	
	We show how to construct the nodes at level $i+1$. Let $H$ be an already-constructed node at level $i$. There are three cases:
	\begin{enumerate}
		\item If $H$ satisfies $|V(H)| \leq b^{i+1} n$ and $|\bdry{H}| \leq cn^\alpha \cdot b^{\alpha (i+1)}$, put a copy of $H$ as its only child at level $i+1$.
		\item If $|V(H)| \geq b^{i+1} n$, then assign a weight of 1 to all vertices, find a balanced vertex separator $S(H)$, and partition $H$ accordingly into $H_1$ and $H_2$. Let us consider $H_1$; the analogous holds for $H_2$.
		
		By definition of separability, we know $|V(H_1)| \leq b' \cdot |V(H)| + |V(H)|^\alpha \leq b \cdot |V(H)| \leq {b}^{i+1} n$ as long as $b \in (b',1)$.
		If $|\bdry{H_1}| \leq c \cdot |V(H_1)|^\alpha$, then we can upper bound this expression by $cn^\alpha \cdot {b}^{\alpha (i+1)}$, and we are done.
		
		On the other hand, if $|\bdry{H_1}| > c \cdot |V(H_1)|^\alpha$, 
		then by definition of boundary, we have $|\bdry{H_1}| \leq |\bdry{H}| + |S(H)| \leq (c + 1)n^\alpha \cdot b^{\alpha i}$ using the guarantees at $H$.
		Next, we assign a weight of 1 to vertices in $\bdry{H_1}$ and 0 to all other vertices, find a balanced separator $S(H_1)$ of $H_1$ with respect to these weights, and create two children $D_1, D_2$ of $H_1$ accordingly.
		Then, for $j=1,2$, we have
		\begin{align*}
			|\bdry{D_j}| &\leq b \cdot |\bdry{H_1}| + |V(H_1)|^\alpha \\
			&\leq b(c+1) n^\alpha \cdot b^{\alpha i} + n^\alpha \cdot b^{\alpha (i+1)} \\
			&\leq \left(b^{1-\alpha} \cdot \frac{c+1}{c} + \frac{1}{c} \right) \cdot cn^\alpha \cdot b^{\alpha(i+1)},
		\end{align*}
		As long as $c$ is large enough so the expression in the parentheses to be less than 1.
		In this case, observe that we can add $S(H_1)$ to the balanced separator $S(H)$, and set $D_1$ and $D_2$ directly as the children of $H$.
		
		\item If $|V(H)| \leq b^{(i+1)n}$ and $|\bdry{H}| \geq cn^\alpha \cdot b^{\alpha (i+1)}$, then we apply case 2 with $H_1$ being $H$.
	\end{enumerate}
	
	So we have shown inductively that at the end of this construction, any node $H$ at level $i$ satisfies $|V(H)| \leq b^i n$ and $|\bdry{H}| \leq c n^\alpha \cdot b^{\alpha i}$.
	It follows that $|\skel{H}| \leq |S(H)| + |\bdry{H}| = O(c n^\alpha \cdot b^{\alpha i})$.
	
	Next, we show that there are only $O(b^{-i})$ nodes at level $i$.
	Let $L_i(n)$ denote the total number of boundary vertices with multiplicities, when carrying out the construction starting on a graph of size $n$ and ending when each leaf node $H$ satisfies the level-$i$ assumptions. We can recursively write
	\begin{alignat*}{3}
		L_i(k) &= \sum_{j=1}^4 L_i(b_j k + 3ck^\alpha), \qquad&&\text{if } k > C b^{i}n \\
		B_i(k) &= 1 &&\text{else.}
	\end{alignat*}
	where $\sum b_j = 1$, each $b_j \leq b'$, and $C$ is a positive constant we choose. To see this, note that a node of size $k$ has at most four children in the construction; the separator is of size $3ck^\alpha$ since we may need to compute up to three separators each of size $ck^\alpha$ and take their union; and child $j$ has at most $b_j k$ vertices that are not from the separator. 
	Solving the recursion yields $L_i(k) \leq k/(Cb^in) - \gamma k^\alpha$ for some constant $\gamma > 0$.
	Therefore, there are at most $L_i(n) \leq O(b^{-i})$ nodes at level $i$.
	
	Finally, it is straightforward to see that the separator tree can be computed in $\O(T(n))$ time, since the node sizes decrease by a geometric factor as we proceed down the tree during construction.
\end{proof}

\begin{proof}[Proof of \cref{thm:main}]
	We consider the cases when $\alpha = 1$ and $\alpha < 1$ separately.
	
	All hypergraphs are trivially $n$-separable with max hyperedge size $\rho = n$. In this case, let $\cs$ be the separator tree consisting of simply one node representing $G_{\ma}$, which is a $(1,1/2,n)$-separator tree.
	By \cref{lem:a-b-lambda-tree-complexity}, the tree operator data structure can be initialized in $O(m^\omega)$ time; the query complexity is $Q(K) = O(nK + n^2)$, and the update complexity is $U(K) = O(n^{\omega-1}K + n^2 K^{\omega-2})$.
	
	We apply \cref{thm:ripm-main} to get the overall runtime:
	\[
		\O \left(\sqrt{m} \log( \frac{R}{\eps r}) \cdot \sum_{\ell=0}^{\frac12 \log m} \frac{n 2^{2\ell} + n^2 + n^2 2^{2\ell (\omega-2)}}{2^\ell} \right) =  \O \left(\sqrt{m} n^2 \log( \frac{R}{\eps r}) \right).
	\]
	If $G_\ma$ is $n^\alpha$-separable for $\alpha < 1$, then by \cref{lem:separable-graph-tree}, $G_{\ma}$ admits a $(\alpha,b, cn^\alpha)$-separator tree computable in $\O(n)$ time.
	In this case, $\rho = O(1)$, and $\eta = O(\log_{1/b} n)$.
	Plugging the parameters into \cref{lem:a-b-lambda-tree-complexity}, we get the following  tree operator runtimes:
	
		\begin{itemize}
		\item The data structure can be initialize in $O \left(m +  n^{\alpha \omega} (1 + n^{1-\alpha \omega})\right) \leq O(m + m^{\alpha \omega})$ time. 
		\item The query complexity is $Q(K) \leq O(K + n^{2 \alpha} (1 + K^{1 - 2\alpha}))$.
		\item The update complexity is 
		\[
		U(K) \leq O \left( K + 
			n^{2\alpha} \min\{K, n^{\alpha}\}^{\omega-2} +
				n^{2\alpha} K^{1-2\alpha} + 
				n^{\frac{\alpha (\omega-1)}{1-\alpha}} K^\frac{1-\alpha\omega}{1-\alpha} \cdot \mathbbm{1}_{K \geq n^{\alpha}} \right).
			\]
		\end{itemize}
		We apply \cref{thm:ripm-main} to get the overall runtime.
		 \begin{align*}
		 	&\phantom{{}={}} \O \left(
		 	\sqrt{m} \log(\frac{R}{\eps r}) \right) \cdot 
		 	\sum_{\ell=0}^{\frac12 \log m} 
		 	\frac{ 2^{2\ell} + 
		 		n^{2\alpha} + 
		 		n^{\alpha \omega} \cdot \mathbbm{1}_{2^{2\ell} > n^\alpha} + 
		 		n^{2\alpha} 2^{(1-2\alpha)2\ell} + 
		 		n^{\frac{\alpha (\omega-1)}{1-\alpha}} 2^{\frac{1-\alpha\omega}{1-\alpha} 2\ell} \cdot \mathbbm{1}_{2^{2\ell} > n^\alpha} 
	 		}{2^\ell}
		 	\\
		 	&= \O \left(\sqrt{m} \log( \frac{R}{\eps r}) \right) \cdot 
		 	\left(\sqrt{m} + n^{2\alpha} + 
		 	n^{\alpha \omega - \frac{\alpha}{2}} + 
		 	n^{2\alpha} m^{1-2\alpha - \frac12} + 
		 	n^{\frac{\alpha (\omega-1)}{1-\alpha}} \left( n^{\frac{\alpha(1-\alpha \omega)}{1-\alpha}-\frac{\alpha}2} + m^{\frac{1-\alpha\omega}{1-\alpha} -\frac12} \right) \right) \\
		 	&=  \O \left( \left( m + m^{\frac12 + 2\alpha} \right) \cdot \log( \frac{R}{\eps r}) \right),
		 \end{align*}
	 	where in the last step, we used the fact $\alpha \omega - \frac{\alpha}{2} \leq 2 \alpha$.
\end{proof}

\subsection{Proof of \cref{thm:k-multicommodity-flow}}

Let $G = (V,E)$ denote the planar graph for the original problem, with
$V = \{v_1, \dots, v_n\}$ and $E = \{e_1, \dots, e_m\}$.
First, we write the LP in \cref{eq:k-commodity-LP} in standard form by adding slack variables $\vs \in \R^E$:
\begin{equation} \label{eq:k-mflow-standard-form} \tag{$P'$}
	\begin{split}
	\min \; \sum_{i=1}^k \; & \vc_i^\top \vf_i \\
	\text{s.t} \qquad\quad \mb^\top \vf_i &= \vd_i \qquad \forall i \in [k] \\
	\sum_{i=1}^k \vf_i + \vs &= \vu \\ 
	\vf_i &\geq \vzero \qquad\; \forall i \in [k]\\
	\vs &\geq \vzero
	\end{split}
\end{equation}

Let $\ma$ denote the full constraint matrix of \ref{eq:k-mflow-standard-form}. Then 
\begin{equation} 
\ma =
	\left[
	\begin{array}{cccc|c}
		\mb^\top & \mzero & \cdots & \mzero & \mzero\\
		\mzero & \mb^\top && \vdots & \vdots \\
		\; &\; & \ddots & & \\
		\mzero & \mzero & \cdots & \mb^\top & \mzero\\
		\hline
		\mi & \mi & \cdots & \mi & \mi
	\end{array}
	\right] \in \R^{(kn+m) \times (k+1)m}
\end{equation}
where the top left part of $\ma$ contains $k$ copies of $\mb^\top$ in  block-diagonal fashion, and all the identity matrices are of dimension $m \times m$.
The dual graph of $\mb^\top$ is precisely $G$.
Let $G_\ma$ be the dual graph of $\ma$. 

First, we describe $G_\ma$: 
It contains $k$ independent copies of the vertices $V$,
which we label with $V^i = (v^i_1, \dots, v^i_n)$,
so that $v^i_j$ is a copy of $v_j \in V$.
Additionally, $G_\ma$ contains $m$ vertices $u_1, \dots, u_m$, where the vertex $u_i$ is identified with edge $e_i \in E$.
For each edge $e_i \in E$ with endpoints $v_{i_1}, v_{i_2}$, there are $k$ hyper-edges in $G_\ma$ 
of the form $\{v_{i_1}^\ell, v_{i_2}^\ell, u_i\}$ for $\ell = 1, \dots, k$. Additionally, there are $m$ hyper-edges $f_1, \dots, f_m$ where $f_i$ contains only the vertex $u_i$.

Next, we show how to construct an appropriate separator tree efficiently.

\begin{claim}
	$G_{\ma}$ admits a $(\frac{1}{2}, b, kn^{1/2})$-separator tree that can be computed in $O(kn \log n)$ time.
\end{claim}
\begin{proof}
	Let $G$ be the original planar graph which is $\sqrt{n}$-separable, and let $\tilde \cs$ be the $(\frac12, b, n^{1/2})$-separator tree for $G$ constructed using \cref{lem:separable-graph-tree} in $O(n \log n)$ time by \cite{lipton1979generalized}.
	We show how to construct a $(\frac12, b, kn^{1/2})$-separator tree $\cs$ for $G_{\ma}$ based on $\tilde\cs$. Without loss of generality, we ignore the hyper-edges $f_1, \dots f_m$ in this construction.
	
	Intuitively, $\cs$ will have the same tree structure as $\tilde\cs$, but each node will be larger by a factor of $O(k)$ due to the $k$ copies of $G$ in $G_{\ma}$.
	For each $\tilde H \in \tilde \cs$, we construct a corresponding $H \in \cs$ as follows:
	if $v_j \in \tilde H$, then $v^i_j \in H$ for all $i \in [k]$;
	if $e_j \in E(\tilde H)$, i.e. both endpoints of $e_j$ are in $\tilde H$, add $u_j$ to $H$.
	Since the $k$ copies $v^1_j, \dots, v^k_j$ are always grouped together, we will refer to them together as $v_j$ in $G_{\ma}$ as well.
	
	Let us show that this is indeed a $(\frac12, b, kn^{1/2})$-separator tree.
	Suppose $H$ is a node with children $D_1$ and $D_2$ in $\cs$, corresponding to nodes $\tilde H, \tilde D_1, \tilde D_2$ in $\tilde \cs$. 
	Let $S(H) \defeq V(D_1) \cap V(D_2)$, then $v_j \in S(H)$ iff $v_j \in S(\tilde H)$, and $u_j \in S(H)$ iff $e_j \in E(S(\tilde H))$ for all values of $j$.
	It is straightforward to see that $S(H)$ is indeed a separator of $H$.
	When it comes to the set of boundary vertices, 
	we see $v_j \in \bdry{H}$ iff $v_j \in \bdry{\tilde{H}}$, and
	if $u_j \in \bdry{H}$ with $v_{j_1}, v_{j_2}$ being the two endpoints of $e_j$, then $v_{j_1}, v_{j_2}$ are both in $\bdry{H}$.
	Since $G$ is a planar graph, the number of edges in $\tilde H$ is on the same order as the number of vertices, so we conclude that $|V(H)| \leq O(k) \cdot |V(\tilde H)|$, and similarly, $|\skel{H}| \leq O(k) \cdot |\skel{\tilde H}|$.
	Since node sizes in $\cs$ have increased by a factor of $O(k)$ compared to $\tilde \cs$, we conclude $\cs$ is a $(\frac12, b, kn^{1/2})$-separator tree.

	Finally, we can compute $\tilde \cs$ for $G$ in $O(n \log n)$ time, so we can compute $\cs$ in $O(kn \log n)$ time.
\end{proof}

We reduce our problem to minimum cost multi-commodity circulation problem in order to establish the existence of an interior point in the polytope, before invoking the RIPM in \cref{thm:IPM}.
For each commodity $i\in [k]$, we add extra vertices $s_i$ and $t_i$. Let $\boldsymbol{d}_i$ be the demand vector of the $i$-th commodity. For every vertex $v$ with $\boldsymbol{d}_{i,v}<0$, we add a directed edge from $s_i$ to $v$ with capacity $-\boldsymbol{d}_{i,v}$ and cost 0. 
For every vertex $v$ with $\boldsymbol{d}_{i,v}>0$, we add a directed edge from $v$ to $t_i$ with capacity $\boldsymbol{d}_{i,v}$ and cost 0. Then, we add a directed edge from $t_i$ to $s_i$ with capacity $4 k m M$ and cost $-4 k m M$. The modified graph $G'$ has only $2k$ extra vertices of the form $s_i$ and $t_i$ compared to $G_\ma$, so we can construct a $(\frac{1}{2}, b, kn^{1/2} + 2k)$-separator tree for $G'$ based on the $(\frac12, kn^{1/2})$-separator tree for $G_{\ma}$, where we include the extra vertices at every node of the tree.

To show the existence of the interior point, we remove all the directed edges that no single commodity flow from $s_i$ to $t_i$ can pass for any $i\in[k]$. This can be done by run BFS for $k$ times which takes $O(km)$ time.
For the interior point $\vf$, we construct this finding a circulation $\vf^{(e)}$ that passing through $e$ and $s_i$, $t_i$ for some $i$ with flow value $1/(10km)$ for all the remaining edge $e$. Then, since the capacities are integers, we find a feasible $\vf,\vs$ with value at least $1/(10km)$. This shows the inner radius $r$ of the polytope is at least $1/(10km)$.
For the $L$ and $R$, we note we can bound it by $O(kmM)$.

Let $\ma'$ be the constraint matrix of the reduced problem with dual graph $G'$.
The RIPM in \cref{thm:IPM} invokes the subroutine \textsc{Solve} twice. In the first run, we make a new constraint matrix by concatenating $\ma'$ three times. One can check that the dual graph is $G'$ with each edge duplicated three times, so the corresponding separator tree is straightforward to construct.

Now, we bounding the running time. The tree operator complexities are similar to the analysis in the previous section with an additional factor of $k$ in the expression for $\lambda$. 
The initialization time is $O(km +(kn^{1/2})^{\omega})$.
The query complexity is $Q(K) = O(K + k^2 n)$.
After simplifying, the update complexity is 
\[
U(K) = K + 
\begin{cases}
	k^2 n K^{\omega-2} &\text{if } K \leq kn^{1/2} \\
	(k n^{1/2})^{\omega} &\text{else}.
\end{cases}
\]
Note that the number of variables is $km$. Plugging our choice of $L$, $R$, and $r$, by \cref{thm:ripm-main}, the total runtime simplifies to 
\[
\O \left(k^{2.5}m^{1.5}\log(M/\eps)\right).
\]

\subsection{Proof of \cref{thm:treewidth}}

		First, we show how to construct a $(0, 1/2, O(\tau \log n))$-separator tree $\cs$ for $G_{\ma}$ when we have a tree decomposition of $G_{\ma}$ of width $\tau$.
		At the root of $\cs$, we can use the tree decomposition to compute a balanced separator $S$ of $G_{\ma}$ of size $O(\tau)$ in $\O(n \tau)$ time (c.f.~\cite[Theorem 4.17]{treeLP}), so that the two parts $A$ and $B$ of $G_{\ma} \setminus S$ each have size at most $\frac23 n$.
		We construct two children of the root node on the vertex sets $A \cup S$ and $B \cup S$ respectively, and apply this procedure recursively until the nodes are of size at most $9\tau$. 
		
		\begin{claim}
			There are $O(n/\tau)$-many leaves at the end of this construction. 
		\end{claim}
		\begin{proof}
			Let $L(k)$ denote the number of leaves when starting the construction with a size $k$ subgraph.
			We know $L(k) = 1$ if $k \leq 9\tau$, and $L(k) = L(k_1 + \tau) + L(k_2 + \tau)$
			if $k > 9\tau$, where $k_1 + k_2 + \tau = k \text{ and } k_1, k_2 \leq 2/3k$.
			By induction, we can show that $L(k) \leq 2(k/\tau - 1)$ when $k > 2\tau$, where the balanced separator crucially ensures that the recursion does not reach the base case of $k \leq 2\tau$.
		\end{proof}
		
		The resulting separator tree is binary, so there are at most $2^i$ nodes at level $i$.
		Since there are $L = O(n/\tau)$-many leaves, the height $\eta$ is at most $\eta \leq \log_2(n/\tau)$.
		The boundary of a node $H$ is contained in the union of balanced separators over its ancestors, so $|\skel{H}| \leq \tau \eta \leq O(\tau \log n)$.
		The max hyperedge size of $G_{\ma}$ is $\rho = \tau$.
	
		Using these values, we simplify the complexities in \cref{lem:a-b-lambda-tree-complexity}:
		The initialization time for the tree operator data structure is
		$\O \left(\tau^{\omega-1} m + \tau^\omega \left(1 + n/\tau \right) \right) = \O(\tau^{\omega-1} m)$.
		The query complexity of $\treeop$ is $Q(K) = \O \left(\tau K + \tau^2 \min\{K,L\} \right)$.
		The update complexity of $\treeop$ is 
		\[
		U(K) \leq \tau^{\omega-1} K +
		\begin{cases}
			\tau^2 K &\text{if $K \leq n$} \\
			\tau^{\omega} &\text{if $K > n$} 
		\end{cases}
		\]
		
		Finally, we apply \cref{thm:ripm-main} to get the overall runtime, which is clearly bounded by
		\begin{align*}
			\O \left(\sqrt{m} \log(\frac{R}{\eps r}) \right) \cdot 
			\sum_{\ell=0}^{\frac12 \log m} 
			\frac{\tau^2 2^{2\ell}}{2^\ell}=  \O \left( m \tau^2 \log(R/(\eps r)) \right).
		\end{align*}
		
		To obtain the faster runtime given in \cite{GS22}, we use the data structure restarting trick:
		Recall \textsc{MaintainApprox} guarantees there are $2^{2\ell}$-many coordinate updates to $\ox$ and $\os$ every $2^\ell$ steps, i.e.\ the number of coordinate updates grows superlinearly with respect to the total number of steps taken. 
		By reinitializing \textsc{MaintainApprox} with the exact solution once in a while, we limit the total number of coordinate updates.
		In the proof of \cref{thm:ripm-main}, we showed that running $M$ steps of the RIPM takes
		\[
		\O \left( U(m) + Q(m) + \eta^4 M \log( \frac{R}{\eps r}) \cdot \sum_{\ell=0}^{\log M} \frac{U(2^{2\ell}) + Q(2^{2\ell})}{2^\ell} \right)
		\]
		time,
		where $U(m) + Q(m)$ is the time to initialize the data structures and obtain the final exact solutions.
		There are $N=\sqrt{m}\log m\log(\frac{mR}{\eps r})$-many total IPM steps, and we reinitialize the data structures every $M$ steps. Then the total running time is (ignoring the big-O notation and log factors)
		\begin{align*}
			&\phantom{{}={}} \frac{N}{M} \left(U(m) + Q(m) + M \sum_{\ell=0}^{\log M} \frac{U(2^{2\ell}) + Q(2^{2\ell})}{2^\ell} \right) \\
			&= \frac{\sqrt{m}}{M} \left( \tau^{\omega-1} m + \tau^2 M^2 \right).
		\end{align*}
		The expression is minimized by taking $M = \sqrt{m} \tau^{\frac{\omega-3}{2}}$,
		which gives an overall runtime of
		\[
		\O \left( m\tau^{(\omega+1)/2} \log(R/(\eps r)) \right).
		\]
		\qed

	\subsection*{Acknowledgements}
	We thank the anonymous reviewers for helpful feedback.
	
	\bibliographystyle{alpha}
	\bibliography{main}
	
	\appendix
	
\section{Robust interior point method\label{sec:IPM}}

For completeness, we include the robust interior point method from \cite{dong2022nested}, 
developed in \cite{treeLPArxivV2}, which is a refinement of the methods in
\cite{CohenLS21, van2020deterministic}.
Although there are many other robust interior point methods, we simply
refer to this method as RIPM. Consider a linear program of the
form
\begin{equation}
	\min_{\vx\in\mathcal{P}} \vc^{\top}\vx\quad\text{where}\quad\mathcal{P}=\{\ma \vx=\vb,\; \vl\leq\vx\leq\vu\}\label{eq:LP}
\end{equation}
for some matrix $\ma\in\mathbb{R}^{n \times m}$. 
As with many other IPMs, RIPM follows the central path $\vx(t)$ from an interior point
($t\gg0$) to the optimal solution ($t=0$):
\[
\vx(t)\defeq\arg\min_{\vx\in\mathcal{P}}\vc^{\top}\vx-t\phi(\vx)\quad\text{where }\phi(\vx)\defeq-\sum_{i}\log(\vx_{i}-\vl_{i})-\sum_{i}\log(\vu_{i}-\vx_{i}),
\]
where the term $\phi$ controls how close the solution $\vx_{i}$ can be 
to the constraints $\vu_{i}$ and $\vl_{i}$.
Following the central path exactly is expensive. Instead, RIPM
maintains feasible primal and dual solution
$(\vx, \vs) \in \mathcal{P} \times \mathcal{S}$, where $\mathcal{S}$
is the dual space given by
$\mathcal{S} = \{\vs: \ma^{\top} \vy+\vs=\vc\text{ for some }\vy\}$, and
ensures $\vx(t)$ is an approximate minimizer.
Specifically, the optimality condition for $\vx(t)$ is given by
\begin{align}
	\mu^{t}(\vx,\vs)& \defeq\vs/t+\nabla\phi(\vx) = \vzero\label{eq:mu_t_def}\\
	(\vx,\vs) &\in \mathcal{P} \times\mathcal{S}\nonumber 
\end{align}
where $\mu^{t}(\vx,\vs)$ measures how close $\vx$ is to the minimizer $\vx(t)$.
RIPM maintains $(\vx,\vs)$ such that 
\begin{equation}
	\|\gamma^{t}(\vx,\vs)\|_{\infty}\leq\frac{1}{C\log m}\text{ where }\gamma^{t}(\vx,\vs)_{i}=\frac{\mu^{t}(\vx,\vs)_{i}}{(\nabla^{2}\phi (\vx))_{ii}^{1/2}}, \label{eq:gamma_t_def}
\end{equation}
for some universal constant $C$. The normalization term
$(\nabla^{2}\phi)_{ii}^{1/2}$ makes the centrality measure
$\|\gamma^{t}(\vx,\vs)\|_{\infty}$ scale-invariant in $\vl$ and $\vu$.

The key subroutine $\textsc{Solve}$ takes as input a point close
to the central path
$(\vx(t_{\mathrm{start}}),\vs(t_{\mathrm{start}}))$, and outputs
another point on the central path
$(\vx(t_{\mathrm{end}}),\vs(t_{\mathrm{end}}))$.  Each step of the
subroutine decreases $t$ by a multiplicative factor of
$(1-\frac{1}{\sqrt{m}\log m})$ and moves $(\vx,\vs)$ within
$\mathcal{P}\times\mathcal{S}$ such that $\vs/t+\nabla\phi(\vx)$ is
smaller for the current $t$.  \cite{treeLPArxivV2} proved that even if each
step is computed approximately, \textsc{IPM} still outputs a
point close to $(\vx(t_{\mathrm{end}}),\vs(t_{\mathrm{end}}))$ using
$\widetilde{O}(\sqrt{m}\log(t_{\mathrm{end}}/t_{\mathrm{start}}))$
steps. 
See \cref{alg:IPM_centering} for a simplified version.
\begin{algorithm}
	\caption{Robust Primal-Dual Interior Point Method from \cite{treeLPArxivV2}\label{alg:IPM_centering}}
	\begin{algorithmic}[1]
		
		\Procedure{$\textsc{RIPM}$}{$\ma \in \mathbb{R}^{n \times m}, \vb, \vc,\vl,\vu,\epsilon$}
		
		\State Define $L \defeq \| \vc \|_{2}$ and $R \defeq \|\vu-\vl\|_{2}$
		
		\State Define $\phi_{i}(x)\defeq-\log(\vu_{i}-x)-\log(x-\vl_{i})$ 
		
		\State Define $\mu_i^t(\vx, \vs) \defeq {\vs_i}/{t}+\nabla \phi_i(\vx_i)$
		
		\Statex
		
		\Statex $\triangleright$ Modify the linear program and obtain an
		initial $(\vx, \vs)$ for modified linear program
		
		\State Let $t=2^{21}m^{5}\cdot\frac{LR}{128}\cdot\frac{R}{r}$
		
		\State Compute $\vx_{c}=\arg\min_{\vl\leq\vx\leq\vu}\vc^{\top}\vx+t\phi(\vx)$
		and $\vx_{\circ}=\arg\min_{\ma\vx=\vb}\|\vx-\vx_{c}\|_{2}$
		
		\State Let $\vx=(\vx_{c},3R+\vx_{\circ}-\vx_{c},3R)$ and $\vs=(-t  \nabla\phi(\vx_{c}),\frac{t}{3R+\vx_{\circ}-\vx_{c}},\frac{t}{3R})$
		
		\State Let the new matrix $\ma^{\mathrm{new}} \defeq [\ma;\ma;-\ma]$, the
		new barrier
		\[
		\phi_{i}^{\mathrm{new}}(x)
		=\begin{cases}
			\phi_{i}(x) & \text{if }i\in[m],\\
			-\log x & \text{else}.
		\end{cases}
		\]
		\Statex $\triangleright$ Find an initial $(\vx,\vs)$ for the original
		linear program
		
		\State $((\vx^{(1)},\vx^{(2)},\vx^{(3)}),(\vs^{(1)},\vs^{(2)},\vs^{(3)}))\leftarrow\textsc{Solve}(\ma^{\mathrm{new}},\phi^{\mathrm{new}},\vx,\vs,t,LR)$
		
		\State $(\vx,\vs)\leftarrow(\vx^{(1)}+\vx^{(2)}-\vx^{(3)},\vs^{(1)})$
		
		\Statex 
		
		\Statex $\triangleright$ Optimize the original linear program
		
		\State $(\vx,\vs)\leftarrow\textsc{Solve}(\ma,\phi,\vx,\vs,LR,\frac{\epsilon}{4m})$
		\State \Return $\vx$
		
		\EndProcedure
		
		\Statex
		
		\Procedure{$\textsc{Solve}$}{$\ma,\phi,\vx,\vs,t_{\mathrm{start}},t_{\mathrm{end}}$}
		
		\State Define $\alpha \defeq \frac{1}{2^{20}\lambda}$ and $\lambda \defeq 64\log(256m^{2})$
		
		\State Let $t\leftarrow t_{\mathrm{start}}$, $\ox\leftarrow\vx,\os\leftarrow\vs,\ot\leftarrow t$
		
		\While{$t\geq t_{\mathrm{end}}$}
		
		\State $t\leftarrow\max((1-\frac{\alpha}{\sqrt{m}})t,t_{\mathrm{end}})$
		
		\State Update step size $h=-\alpha/\|\cosh(\lambda\gamma^{\ot}(\ox,\os))\|_{2}$
		where $\gamma$ is defined in \cref{eq:gamma_t_def} \label{line:step_given_begin}
		
		\State Update diagonal weight matrix $\mw=\nabla^{2}\phi(\ox)^{-1}$\label{line:step_given_3}
		
		\State Update step direction $\vv$ where  $\vv_{i}=\sinh(\lambda\gamma^{\ot}(\ox,\os)_i) \cdot  \mu^{\ot}(\ox,\os)_i$\label{line:step_given_end}
		
		\State Implicitly update $\vx, \vs$, with $\mproj_{\vw}\defeq\mw^{1/2}\ma^\top(\ma\mw\ma^\top)^{-1}\ma\mw^{1/2}$
		\begin{align*}
			\vx &\leftarrow \vx+h\mw^{1/2}(\vv - \mproj_{\vw} \vv), \\
			\vs &\leftarrow \vs+\ot h\mw^{-1/2} \mproj_{\vw} \vv
		\end{align*}

		\State Explicitly update $\ox, \os$ such that
		\begin{align*}
			\|\mw^{-1/2}(\ox-\vx)\|_{\infty} &\leq \alpha, \\
			\|\mw^{1/2}(\os-\vs)\|_{\infty} &\leq \ot \alpha
		\end{align*} 
		
		\State If $|\ot-t|\geq\alpha\ot$, update $\ot\leftarrow t$
		
		\EndWhile
		
		\State \Return $(\vx,\vs)$\label{line:step_user_output}
		
		\EndProcedure
		
	\end{algorithmic}
	
\end{algorithm}

RIPM calls \textsc{Solve} twice. The first call to
\textsc{Solve} finds a feasible point by following the central
path of the following modified linear program
\[
\min_{ \substack{\ma (\vx^{(1)}+\vx^{(2)}-\vx^{(3)})=\vb \\
		\vl\leq\vx^{(1)}\leq\vu,\; 
		\vx^{(2)}\geq\vzero,\; 
		\vx^{(3)}\geq\vzero}}
\vc^{(1)\top}\vx^{(1)}+\vc^{(2)\top}\vx^{(3)}+\vc^{(2)\top}\vx^{(3)}
\]
where $\vc^{(1)}=\vc$, and $\vc^{(2)},\vc^{(3)}$ are some positive
large vectors. The above modified linear program is chosen so that we
know an explicit point on its central path, and any approximate
minimizer to this new linear program gives an approximate central path
point for the original problem.
The second call to \textsc{Solve} finds an approximate solution
by following the central path of the original linear program. 


\RIPM*
\begin{proof}
	The number of steps follows from Theorem A.1 in \cite{treeLPArxivV2}, with
	the parameter $\vw_{i}=\nu_{i}=1$ for all $i$. 
	The number of coordinate
	changes in $\mw,\vv$ and the runtime of \cref{line:step_given_begin}
	to \cref{line:step_given_end} follows directly from the formula of
	$\mu^{t}(\vx,\vs)_{i}$ and $\gamma^{t}(\vx,\vs)_{i}$. For the bound
	for $h\|\vv\|_{2}$, it follows from
	\[	h\|\vv\|_{2}\leq\alpha\frac {\|\sinh(\lambda\gamma^{\ot}(\ox,\os))\|_{2}}{\|\cosh(\lambda\gamma^{\ot}(\ox,\os))\|_{2}} \leq \alpha = O\left(\frac{1}{\log m}\right).
	\]
\end{proof}

	\section{Maintaining the implicit representation} 

In this section, we give the general data structure \textsc{MaintainRep}, which implicitly maintains a vector $\vx$ throughout a call of \textsc{Solve} of \cref{alg:IPM_centering}.
We break up the representation into two parts, the first using the inverse tree operator, and the second using the tree operator.

First, we present some of the alternative decomposition properties of the tree operator.
\begin{definition}[Subtree operator]
	\label{defn:subtree-operator}
	Let $\treeop$ be a tree operator on $\ct$.
	Recall $\ct_H$ is the complete subtree of $\ct$ rooted at $H$.
	We define the subtree operator $\treeop^{(H)}$ at each node $H$ to be
	\begin{equation} \label{eq:tree-op-subtree-op}
		\treeop^{(H)} \defeq  \sum_{\text{leaf } L \in \ct_H} \treeop_{L\leftarrow H}.
	\end{equation}
\end{definition}

\begin{corollary}
	Based on the above definitions, we have
	\begin{equation} \label{eq:tree-op-decomp}
		\treeop = \sum_{H \in \ct} \treeop^{(H)} \mi_{F_H}.
	\end{equation}
	Furthermore, if $H$ has children $H_1, H_2$, then
	\begin{equation} \label{eq:tree-op-children-decomp}
		\treeop^{(H)} = \treeop^{(H_1)} \treeop_{H_1} + \treeop^{(H_2)} \treeop_{H_2}.
	\end{equation}
\end{corollary}

The output of $\treeop$ when restricted to $E(H)$ for a node $H \in \ct$
can be written in two parts,
which is useful for our data structures.
The first part involves summing over all nodes in $\ct_H$, ie. descendants of $H$ and $H$ itself, and the second part involves a sum over all ancestors of $H$. 

\begin{lemma} \label{lem:treeop-subtree-ancestor-decomp}
	At any node $H \in \ct$, we have
	\[
	\mi_{E(H)} \treeop = \sum_{D \in \ct_H} \treeop^{(D)} \mi_{F_D} + \treeop^{(H)} \sum_{\text{ancestor $A$ of $H$}} \treeop_{H \leftarrow A} \mi_{F_A}.
	\]
\end{lemma}
\begin{proof}
	We consider the terms in the sum for $\treeop$ that map into to $E(H)$, which is precisely the set of leaf nodes in the subtree rooted at $H$.
	\[
	\mi_{E(H)} \treeop  = \sum_{\text{leaf } L \in \ct_H} \sum_{A : L \in \ct_A} \treeop_{L \leftarrow A} \mi_{F_A} .
	\]
	The right hand side involves a sum over the set $\{ (L, A) \;:\; \text{leaf } L \in \ct_H, L \in \ct_A\}$. Observe that $(L,A)$ is in this set if and only if $A$ is a descendant of $H$, or $A = H$, or $A$ is an ancestor of $H$. Hence, the summation can be written as
	\[
	\sum_{\text{leaf $L \in \ct_H$}} \sum_{\text{node $A \in \ct_H$}} \treeop_{L \leftarrow A}\mi_{F_A} + 
	\sum_{\text{leaf $L \in \ct_H$}} \sum_{\text{ancestor $A$ of $H$}} \treeop_{L \leftarrow A} \mi_{F_A} .
	\]
	The first term is precisely the first term in the lemma statement.
	For the second term, we can use the fact that $A$ is an ancestor of $H$ to expand $\treeop_{L \leftarrow A} = \treeop_{L \leftarrow H} \treeop_{H \leftarrow A}$. Then, the second term is
	\begin{align*}
		&\phantom{{}={}} \sum_{\text{leaf $L \in \ct_H$}} \sum_{\text{ancestor $A$ of $H$}} \treeop_{L \leftarrow H} \treeop_{H \leftarrow A} \mi_{F_A}  \\
		&= \sum_{\text{leaf $L \in \ct_H$}} \treeop_{L \leftarrow H} \left( \sum_{\text{ancestor $A$ of $H$}} \treeop_{H \leftarrow A} \mi_{F_A}  \right) \\
		&= \treeop^{(H)} \left(\sum_{\text{ancestor $A$ of $H$}} \treeop_{H \leftarrow A} \mi_{F_A}  \right),
	\end{align*}
	by definition of $\treeop^{(H)}$.
\end{proof}

Now, we consider the cost of applying the inverse tree operator and the tree operator.
\begin{lemma} \label{lem:itreeop-time}
	Let $\itreeop : \R^E \mapsto \R^V$ be an inverse tree operator on $\ct$ with query complexity $Q$.
	Given $\vv \in \R^E$, we can compute $\itreeop \vv$ as well as 
	$\vy_H \defeq \sum_{\text{leaf } L \in \ct_H} \itreeop_{H \leftarrow L} \vv$ for all $H \in \ct$ 
	in $O(Q(\eta K))$ time, where $K = \nnz(\vv)$ and $\eta$ is the height of $\ct$.
\end{lemma}
\begin{proof}
	Recall the definition
	\[
		\itreeop \vv \defeq \sum_{\text{leaf $L$}} \left( \sum_{H :\; L \in \ct_H} 
		\mi_{F_H} \itreeop_{H \leftarrow L} \right) \vv.
	\]
	At a leaf node $L$, if we have $\vv_e = 0$ for all $e \in E(L)$, then we can ignore the term for $L$ in the outer sum. So we can reduce $\itreeop \vv$ to consist of at most $K$ terms in the outer sum.
	We can further rearrange the order of applying the edge operators so that each edge operator is applied at most once,
	and this naturally gives the values for all non-zero $\vy_H$'s.
	We bound the overall runtime loosely by $O(Q(\eta K))$.
\end{proof}

Unlike the inverse tree operator, the tree operator is applied downwards along a tree, and therefore we do not have non-trivial bounds on total number of edge operators applied.
Instead, we have a more general bound:

\begin{lemma} \label{lem:treeop-time}
	Let $\treeop : \R^V \mapsto \R^E$ be a tree operator on $\ct$ with query complexity $Q$.
	Given $\vz \in \R^V$, we can compute $\treeop \vv$ in $O(Q(|E|))$ time.
\end{lemma}
\begin{proof}
	We simply observe that we can compute $\treeop \vv$ by applying each edge operator at most once.
	Since the leaf nodes partition the set $E$, we know in $\ct$, there are $O(|E|)$ edge operators in total, so the overall time is at most $O(Q(|E|))$.
\end{proof}

With the appropriate partial computations taking advantage of the decomposition of $\itreeop$, we can maintain $\itreeop \vv$ efficiently for dynamic $\itreeop$ and $\vv$. Specifically, we use the following property:

\begin{lemma}
	Given a vector $\vv \in \R^E$, let $\vy_H \defeq \sum_{\text{leaf } L \in \ct_H} \itreeop_{H \leftarrow L} \vv$ for each $H \in \ct$. If $H$ has children $H_1, H_2$, then
	\begin{equation} \label{eq:itreeop-children-decomp}
		\vy_H = \itreeop_{H_1} \vy_{H_1} + \itreeop_{H_2} \vy_{H_2}.
	\end{equation}
	Furthermore,
	\begin{equation} \label{eq:itreeop-y}
	\sum_{H \in \ct} \mi_{F_H}  \vy_H = \itreeop \vv.
	\end{equation}
\qed
\end{lemma}

\begin{lemma} \label{lem:itreeop-diff-time}
	Let $\itreeop : \R^E \mapsto \R^V$ be an inverse tree operator with query complexity $Q$.
	Let $\itreeop^\new$ be $\itreeop$ with $K$ updated edge operators.
	Suppose we know $\itreeop \vv$, and we know $\vy_H \defeq \sum_{\text{leaf } L \in \ct_H} \itreeop_{H \leftarrow L} \vv$ at all nodes $H$, then we can compute $(\itreeop^\new - \itreeop) \vv$ and the $\vy_H^{\new}$'s in $O(Q(\eta K))$ time.
\end{lemma}
\begin{proof}
	Observe that for a node $H \in \ct$, if no edge operator in $\ct_H$ was updated, then $\vy_H$ remains the same.
	We use \cref{eq:itreeop-children-decomp} to compute $\vy_H^\new$ up the tree for the $O(\eta K)$-many nodes that admit changes,
	and then \cref{eq:itreeop-y} to extract the change $(\itreeop^\new - \itreeop) \vv$.
\end{proof}

Now we are ready for the complete data structure involving the inverse tree operator.

\begin{algorithm}
	\caption{Dynamic data structure to maintain cumulative $\itreeop \vv$}\label{algo:inverse-tree-operator}
	\begin{algorithmic}[1]
		\State \textbf{dynamic data structure \textsc{InverseTreeOp}}
		\State \textbf{member:}
		\State \hspace{4mm} $\ct$: tree supporting $\itreeop$ with edge operators on the edges
		\State \hspace{4mm} $\vw \in \R^m$: dynamic weight vector
		\State \hspace{4mm} $\vv \in \R^n$: dynamic vector
		\State \hspace{4mm} $c, \zprev, \zsum \in \R^n$: coefficient, result vectors
		\State \hspace{4mm} $\vy_H \in \R^{n}$ for each $H \in \ct$: sparse partial computations
		\State
		\Procedure{\textsc{Initialize}}{$\ct, \vw^{\init}, \vv^{\init}$}
		\State $\vw \leftarrow \vw^{\init}, \vv \leftarrow \vv^{\init}, c \leftarrow 0, \zsum \leftarrow \vzero$
		\State Initialize $\itreeop$ on $\ct$ based on $\vw$
		\State Compute $\itreeop \vv$ and $\vy_H$'s, set $\zprev \leftarrow \itreeop \vv$
		\EndProcedure
		\State
		\Procedure{\textsc{Reweight}}{$\delta_{\vw}$}
		\State $\vw^\new \leftarrow \vw + \delta_{\vw}$
		\State Let $\itreeop^\new$ be the new tree operator using $\vw^\new$
		\State $\vz' \leftarrow (\itreeop^\new - \itreeop) \vv$, and update $\vy_H$'s \Comment{\cref{lem:itreeop-diff-time}}
		\State $\zprev \leftarrow \zprev + \vz'$
		\State $\zsum \leftarrow \zsum - c \cdot \vz'$
		\State $\vw \leftarrow \vw^\new, \itreeop \leftarrow \itreeop^\new$
		\EndProcedure
		\State
		\Procedure{\textsc{Move}}{$h$, $\delta_{\vv}$}
		\State Compute $\vz' \defeq \itreeop \delta_{\vv}$ and the $\vy'_H \defeq \sum_{\text{leaf $L \in \ct_H$}} \itreeop_{H \leftarrow L} \delta_{\vv}$ for each node $H$ \Comment{\cref{lem:itreeop-time}}
		\State $\zprev \leftarrow \zprev + \vz'$, and $\vy_H \leftarrow \vy_H + \vy'_H$ for each node $H$
		\State $\zsum \leftarrow \zsum - c \vz'$
		\State $c \leftarrow c + h$
		\State $\vv \leftarrow \vv + \delta_{\vv}$
		\EndProcedure
	\end{algorithmic}
\end{algorithm}

\begin{theorem}[Inverse tree operator data structure]
	Let $\vw \in \R^m$ be the weights changing at every step of \textsc{Solve},
	and let $\vv \in \R^n$ be a dynamic vector.
	Suppose $\itreeop : \R^m \mapsto \R^n$ is an inverse tree operator dependent on $\vw$ supported on $\ct$ with query complexity $Q$ and update complexity $U$.
	Let $\eta$ be the height of $\ct$.
	Then the data structure \textsc{InverseTreeOp} (\cref{algo:inverse-tree-operator}) maintains $\vz^{(k)} \defeq \sum_{i=1}^k h^{(i)} \itreeop^{(i)} \vv^{(i)}$ so that at the end of each step $k$,
	the variables in the algorithm satisfy
	\begin{itemize}
		\item $\vz = c \zprev + \zsum$,
		\item $\zprev = \itreeop \vv$, and
		\item $\vy_H = \sum_{\text{leaf } L \in \ct_H} \itreeop_{H \leftarrow L} \vv$ for all nodes $H$.
	\end{itemize}
	The data structure is initialized via \textsc{Initialize} in $O(U(m) + Q(m))$ time.
	At step $k$, there is one call \textsc{Reweight}$(\delta_{\vw})$ taking $O(U(K) + Q(\eta K))$ time, where $K = \nnz(\delta_{\vw})$, 
	followed by one call of \textsc{Move}$(h, \delta_{\vv})$ taking $O(Q(\eta \cdot \nnz(\delta_{\vv})))$ time.
\end{theorem}

\begin{proof}
	In the data structure, we always maintain $\zprev$ and the $\vy_H$'s together. Specifically, at every step, we update the $\vy_H$'s up the tree using the recursive property \cref{eq:itreeop-children-decomp} only at the necessary nodes, and from the $\vy_H$'s, we get $\zprev = \sum_H \mi_{F_H} \vy_H$.
	
	Consider \textsc{Initialize}. At the end of the function, the variables satisfy
	\[
		\vz \defeq c \zprev + \zsum = 0 \cdot \itreeop \vv + \vzero = \vzero,
	\]
	and $\zprev = \itreeop \vv$, as required.
	
	Let us consider \textsc{Reweight}. Let the superscript $^\new$ denote the value of an algorithm variable at the end of the function, and let no superscript denote the value at the start.
	\begin{align*}
		\vz^\new &= c^\new \zprev^\new + \zsum^\new \\
		&= c (\zprev + \vz') + \zsum - c \vz' \\
		&= c \zprev + \zsum, \\
		\text{and } \zprev^\new &= \zprev + (\itreeop^\new - \itreeop) \vv \\
		&= \itreeop^\new \vv,
	\end{align*}
	as required.
	Similarly, let us consider \textsc{Move}:
	\begin{align*}
		\vz^\new &= c^\new \zprev^\new + \zsum^\new \\ 
		&= (c + h) (\zprev + \vz') + \zsum - c  \vz' \\
		&= c \zprev + \zsum + h \zprev,\\
		\text{and } \zprev^\new &= \zprev + \itreeop (\vv^\new - \vv) \\
		&= \itreeop \vv + \itreeop (\vv^\new - \vv) \\
		&= \itreeop \vv^\new,
	\end{align*}
	which is exactly the update we want to make to $\vz$, and the invariant we want to maintain.
	
	The runtimes follow directly from \cref{lem:itreeop-time,lem:itreeop-diff-time}.
\end{proof}


Next, we present the tree operator data structure, which is significantly more involved compared to the inverse tree operator.
Applying the tree operator involves going down the tree to the leaves, which is too costly to do at every step.
To circumvent the issue, we use lazy computations.

\begin{algorithm}
	\caption{Dynamic data structure to maintain cumulative $\treeop \vz$}
	\label{algo:tree-operator}
	\begin{algorithmic}[1]
		\State \textbf{dynamic data structure \textsc{TreeOp}}
		\State \textbf{member:}
		\State \hspace{4mm} $\ct$: tree supporting $\treeop$
		\State \hspace{4mm} $\vw \in \R^m$: dynamic weight vector
		\State \hspace{4mm} $\vz\in \R^n$: dynamic vector
		\State \hspace{4mm} $\vu_H$ for each $H \in \ct$: lazy pushdown computation vectors
		
		\State
		\Procedure{\textsc{Initialize}}{$\ct, \vw^{\init}, \vz^{\init}, \vx^{\init}$}
		\State $\vw \leftarrow \vw^{\init}, \vz \leftarrow \vz^{\init}$
		\State Initialize $\treeop$ on $\ct$ based on $\vw$
		\State $\vu_H \leftarrow \vzero$ for each non-leaf $H \in \ct$
		\State $\vu_H \leftarrow \vx^{\init}|_{E(H)}$ for each leaf $H \in \ct$
		\EndProcedure
		\State
		\Procedure{\textsc{Reweight}}{$\delta_{\vw}$}
		\State $\vw \leftarrow \vw + \delta_{\vw}$
		\State Let $\treeop^\new$ be the new tree operator wrt the new weights
		\State Let $\collN$ be all nodes $H$ where $\treeop_H$ changed
		\For{$H \in \pathT{\collN}$ going down the tree level by level}
		\State \textsc{Pushdown}$(H)$
		\EndFor
		\For{$H \in \pathT{\collN}$ going down the tree level by level}
		\State $\vu_H \leftarrow c \vz|_{F_H}$
		\State \textsc{Pushdown}$(H)$
		\EndFor
		\State $\treeop \leftarrow \treeop^\new$
		\For{$H \in \pathT{\collN}$ going down the tree level by level}
			\State $\vu_H \leftarrow - c \vz|_{F_H}$
			\State \textsc{Pushdown}$(H)$
		\EndFor
		\EndProcedure
		\State
		\Procedure{Move}{$\delta_{\vz}$}
		\State $\vz \leftarrow \vz + \delta_{\vz}$
		\EndProcedure
		\State
		\Procedure{\textsc{Exact}}{}
		\For{$H \in \ct$ going down the tree level by level}
			\State $\vu_H \leftarrow \vu_H + \vz|_{F_H}$
			\State \textsc{Pushdown}$(H)$
		\EndFor
		\State \Return $\vx$ defined by $\vx|_{E(H)} \defeq \vu_H$ at each leaf $H \in \ct$
		\EndProcedure
	\end{algorithmic}
\end{algorithm}

\begin{algorithm}
	\caption{Dynamic data structure to maintain cumulative $\treeop \vz$, con't}
	\begin{algorithmic}[1]
	\State \textbf{dynamic data structure \textsc{TreeOp}}
	\Procedure{\textsc{Pushdown}}{$H \in \ct$}
	\For{each child $D$ of $H$}
	\State $\vu_{D} \leftarrow \vu_{D} + \treeop_D \vu_H$
	\EndFor
	\State $\vu_{H} \leftarrow \vzero$
	\EndProcedure
\end{algorithmic}
\end{algorithm}

\begin{theorem}[Tree operator data structure]
	 \label{thm:tree-operator-data-struct}
	Let $\vw \in \R^m$ be the weights changing at every step of \textsc{Solve}.
	Suppose $\treeop : \R^n \mapsto \R^m$ is a tree operator dependent on $\vw$ supported on $\ct$ with query complexity $Q$ and update complexity $U$.
	Let $\vz \in \R^n$ be the vector maintained by \cref{algo:inverse-tree-operator},
	so that at the end of step $k$,
	$\vz = \sum_{i=1}^k h^{(i)} \itreeop^{(i)} \vv^{(i)}$.
	Then the data structure \textsc{TreeOp} (\cref{algo:tree-operator}) implicitly maintains $\vx$
	so that at the end of step $k$,
	\[
	\vx^{(k)} = \vx^{\init} + \sum_{i=1}^k \treeop^{(i)} \itreeop^{(i)} \vv^{(i)}.
	\]
	
	The data structure is initialized via \textsc{Initialize} in $O(U(m))$ time.
	At step $k$, there is one call to \textsc{Reweight}$(\delta_{\vw})$ taking $O(U(K) + Q(\eta K))$ time, where $K = \nnz(\delta_{\vw})$, followed by one call to \textsc{Move}$(\delta_{\vz})$ taking $O(\nnz(\delta_{\vz}))$ time.
	At the end of \textsc{Solve}, $\vx$ is returned via \textsc{Exact} in $O(Q(m))$ time.
\end{theorem}

\begin{proof}
	We will show that the data structure maintains the implicit representation via the identity
	\begin{equation} \label{eq:tree-operator-invariant}
		\vx = c \treeop \vz + \sum_{H \in \ct} \treeop^{(H)} \vu_H,
	\end{equation}
	where the RHS expression refers to the state of the variables at the end of step $k$ during the algorithm.

	At a high level, the variables $\treeop$ and $\vz$ in the data structure at step $k$ represent the latest $\treeop^{(k)}$ and $\vz^{(k)}$. 
	We need to introduce additional vectors $\vu_H$ at every node $H$ which intuitively stores lazy computations at node $H$, in order to take advantage of the tree structure of $\treeop$. 
	The function \textsc{Pushdown} performs the accumulated computation at $H$, and moves the result to its children nodes to be computed lazily at a later point.
	The next claim describes this process formally.
	
	\begin{claim} \label{claim:pushdown}
		Let $H \in \ct$ be a non-leaf node. \textsc{Pushdown}$(H)$ does not change the value of the implicit representation in \cref{eq:tree-operator-invariant}. 
		Also, at the end of the procedure, $\vu_{H} = \vzero$.
	\end{claim}
	\begin{proof}
		For any variable in the algorithm, we add the superscript $^\new$ to mean its state at the end of \textsc{Pushdown}; if there is no superscript, then it refers to the state at the start.
		
		We show the claim for when $H$ has two children $H_1, H_2$. 
		Note that $\treeop$ and $\vz$ are not touched by \textsc{Pushdown}, so we may ignore the term $c \treeop \vz$ in \cref{eq:tree-operator-invariant}. 
		Then,
		\begin{align*}
			&\phantom{{}={}} \sum_{H' \in \ct} {\treeop^{(H')}} \vu_{H'}^\new \\
			&= \treeop^{(H)} \vu_H^\new + \sum_{i=1,2} \treeop^{(H_i)} \vu_{H_i}^\new + \sum_{H' \in \ct, H' \neq H, H_1, H_2} \treeop^{(H')} \vu_{H'} \tag{expand terms} \\
			&= \sum_{i=1,2}  \treeop^{(H_i)} (\vu_{H_i} + \treeop_{H_i} \vu_H) + 
			\sum_{H' \in \ct, H' \neq H, H_1, H_2} \treeop^{(H')} \vu_{H'} \tag{substitute values} \\
			&= \sum_{i=1,2}  \treeop^{(H_i)} \treeop_{H_i} \vu_H + 
			\sum_{H \in \ct, H' \neq H} \treeop^{(H')} \vu_{H'} \\
			&= \treeop^{(H)} \vu_H + \sum_{H \in \ct, H' \neq H} \treeop^{(H')} \vu_{H'} \tag{by \cref{eq:tree-op-children-decomp}} \\
			&= \sum_{H' \in \ct} \treeop^{(H')} \vu_{H'},
		\end{align*}
		so the implicit representation of $\vx$ has not changed in value.
	\end{proof}
	This claim can be generalized from $H \in \ct$ to $\mathcal{H} \subseteq \ct$; we omit the full details.
	Next, we show that the implicit representation of $\vx$ by \cref{eq:tree-operator-invariant} is correctly maintained
	after reweight.
	
	\begin{claim}
		After the $k$-th call \textsc{Reweight}, the value of $\vx$ is unchanged, while the value of $\treeop$ is updated to $\treeop^{(k)}$ which is a function of $\vw^{(k)}$.
	\end{claim}
	\begin{proof}
		We begin by observing that if $H \notin \pathT{\collN}$, then $\treeop^{(H)\new} = \treeop^{(H)}$ by definition, as there are no edges in $\ct_H$ with updated operators.
		
		At a high level, we traverse the subtree $\pathT{\collN}$ three rounds and perform \textsc{Pushdown} at every node.
		During the first round, we simply push down the current $\vu_H$ values at each node $H$. 
		By \cref{claim:pushdown}, we know this does not change the value of the implicit representation.
		
		During the second round, we first initialize $\vu_H \leftarrow c \vz|_{F_H}$ at each node $H \in \pathT{\collN}$, and then perform \textsc{Pushdown}. Since \textsc{Pushdown} does not affect the value of the implicit representation, we can use the initial change in $\vu_H$ to determine the overall change in the implicit representation.
		Crucially, note that we perform \textsc{Pushdown} using the old tree operator.
		So, the change in value of the implicit representation is given by
		\begin{align*}
			+ c \sum_{H \in \pathT{\collN}} \treeop^{(H)} \vz|_{F_H}.
		\end{align*}
	
		After the second round of \textsc{Pushdown}, we update the tree operator $\treeop$ to $\treeop^\new$.
		Note that $\treeop^{(H)}$ changes if and only if $H \in \pathT{\collN}$, and in this case, $\vu_H = \vzero$.
		So, updating the tree operator at this point induces a change in the value of the implicit representation of
		\begin{align*}
			c \treeop^\new \vz -  c \treeop \vz = c \sum_{H \in \ct} \left(\treeop^{(H) \new} - \treeop^{(H)} \right)\vz|_{F_H}
			= c \sum_{H \in \pathT{\collN}} \left(\treeop^{(H) \new} - \treeop^{(H)} \right)\vz|_{F_H}.
		\end{align*}
		
		During the third round, we initialize $\vu_H \leftarrow - c \vz|_{F_H}$ at each node $H \in \pathT{\collN}$ and perform \textsc{Pushdown}.
		Similar to the first round, the change to the value of the implicit representation induced by this round is given by
		\begin{align*}
			- c \sum_{H \in \pathT{\collN}} \treeop^{(H) \new} \vz|_{F_H}.
		\end{align*}
		The sum of the changes from each of the three rounds is exactly $\vzero$, so we conclude the value of the implicit representation did not change.
	\end{proof}
	
	Finally, we consider the other functions.

	For \textsc{Initialize}, we see that by substituting the values assigned during \textsc{Initialize} and applying the definition from \cref{eq:tree-op-decomp}, we have
	\[
		c \treeop \vz + \sum_{H \in \ct} {\treeop^{(H)}} \vu_H = 
		\vx^\init + \treeop \vz,
	\]
	where $\treeop$ is the initial $\treeop^{\init}$ and $\vz$ is the initial $\vz^{\init}$, which is exactly how we want to initialize $\vx$.
	
	For \textsc{Move}, we see the value of $\vx$ is incremented by $\treeop^{(k)} (\vz^{(k)} - \vz^{(k-1)})$ after the step $k$.
	By definition of $\vz$, we know $\vz^{(k)} - \vz^{(k-1)} = h^{(k)} \itreeop^{(k)} \vv^{(k)}$, so we conclude \textsc{Move} correctly makes the update $h^{(k)} \treeop^{(k)} \itreeop^{(k)} \vv^{(k)}$.
	
	For \textsc{Exact}, we calculate the value of $\vx$ explicitly by performing the computation $\sum_{H \in \ct} \treeop^{(H)} (\vu_H + \vz|_{F_H})$ using a sequence of \textsc{Pushdown}'s down the tree.
	The final answer $\vx$ is stored in parts in the $\vu_H$'s along the leaf nodes.
	
	Note that by definition of the query complexity of $\treeop$, \textsc{Pushdown} uses $O(Q(1))$ time. 
	The remaining runtimes are straightforward.
\end{proof}

Finally, 
we combine \cref{algo:inverse-tree-operator} and \cref{algo:tree-operator} to get the overall data structure \textsc{MaintainRep} for maintaining $\vx$ throughout \textsc{Solve} as given by \cref{eq:simplified-x-update}.
We omit the pseudocode implementation.

\begin{proof}[Proof of \cref{thm:maintain_representation}]
	We use one copy of \textsc{InverseTreeOp}, which maintains $\vz \defeq c \zprev + \zsum$.
	We want to use \textsc{TreeOp} to maintain $\vz$ which is given in two terms by \textsc{InverseTreeOp}.
	To resolve this, we can simply use two copies of the data structure and track the two terms in $\vz$ separately; 
	then we correctly maintain $\vx$. 
	During \textsc{Solve}, at step $k$, we first call \textsc{Reweight} and \textsc{Move} in \textsc{InverseTreeOp}, followed by \textsc{Reweight} and \textsc{Move} in each copy of \textsc{TreeOp}.
	The runtimes follow in a straightforward manner.
\end{proof}

	\section{Maintaining vector approximation}
\label{sec:sketch}

We include this section for completeness; all techniques are from \cite{dong2022nested}.

Recall at every step of the IPM, we want to maintain approximate vectors $\ox, \os$ so that
\begin{align*}
	\norm{ \mw^{-1/2} (\ox - \vx)}_\infty \leq \delta \quad \text{and} \quad 
	\norm{ \mw^{1/2} (\os- \vs)}_\infty \leq \delta'
\end{align*}
for some error tolerances $\delta$ and $\delta'$.

In the previous section, we showed how to use \textsc{MaintainRep} to maintain $\vx$ implicitly throughout \textsc{Solve} in the IPM.
In this section, we give a data structure to efficiently maintain an approximate vector $\ox$ to the $\vx$ from \textsc{MaintainRep}, so that at every step,
\[
	\norm{\mw^{-1/2} \left(\ox - \vx\right)}_\infty \leq \delta.
\]
In the remainder of this section, we crucially assume that $\vw$ is a function of $\ox$ coordinate-wise, which is indeed satisfied by the RIPM framework.

\begin{remark}
	In our problem setting, we do not have full access to the exact vector $\vx$. 
	The algorithms in the next two subsections however will refer to $\vx$ for readability and modularity. 
	We observe that access to $\vx$ is limited to two types: 
	accessing the JL-sketches of specific subvectors, and accessing exact coordinates and other specific subvectors of sufficiently small size.
\end{remark}

In \cref{subsec:sketch_vector_to_change}, we reduce the problem of maintaining $\ox$ to detecting coordinates of $\vx$ with large changes. 
In \cref{subsec:sketch_change_to_sketch}, we detect coordinates of $\vx$ with large changes using a sampling technique on a constant-degree tree, where Johnson-Lindenstrauss sketches of subvectors of $\vx$ are maintained at each node the tree. 
In \cref{subsec:sketch_maintenance}, we show how to compute and maintain the necessary collection of JL-sketches on the operator tree $\ct$; in particular, we do this efficiently with only an implicit representation of $\vx$.
Finally, we put the three parts together to prove \cref{thm:soln-approx}.

For notational simplicity, we use $\md \defeq \mw^{-1/2}$. 
Recall we use the superscript $^{(k)}$ to denote the variable at the end of step $k$; that is, $\md^{(k)}$ and $\vx^{(k)}$ are the values of $\md$ and $\vx$ at the end of step $k$. Step 0 is the initialization step.

\subsection{Reduction to change detection} 
\label{subsec:sketch_vector_to_change}

In this section, we show that in order to maintain an approximation $\ox$ to some vector $\vx$, it suffices to detect coordinates of $\vx$ that change a lot. 

We make use of dyadic intervals. At step $k$, for each $\ell$ such that $k \equiv 0 \bmod 2^\ell$, we find the index set $I_{\ell}^{(k)}$ that contains all coordinates $i$ of $\vx$ such that $\vx_i^{(k)}$ changed significantly compared to $\vx_i^{(k-2^\ell)}$, that is, compared to $2^\ell$ steps ago. Formally:

\begin{definition} \label{defn:sketch-index-set}
	At step $k$ of the IPM, for each $\ell$ such that $k \equiv 0 \bmod 2^\ell$, we define 
	\begin{align*}
		I_{\ell}^{(k)} \defeq &\; \{i\in[n]: \md_{ii}^{(k)} \cdot|\vx_{i}^{(k)}-\vx_{i}^{(k-2^{\ell})}|\geq\frac{\delta}{2\left\lceil \log m\right\rceil }, \\
		&\;\text{and $\ox_i$ has not been updated after the $(k-2^{\ell})$-th step}\}.
	\end{align*}
\end{definition}

We show how to find the sets $I_\ell^{(k)}$ with high probability in the next section.
Assuming the correct implementation, we have the following data structure for maintaining the desired approximation $\ox$:

\begin{algorithm}
	\caption{Data structure \textsc{AbstractMaintainApprox}, Part 1}
	\label{algo:maintain-vector}
	
	\begin{algorithmic}[1]
		\State \textbf{data structure} \textsc{AbstractMaintainApprox} 
		\State \textbf{private : member}
		\State \hspace{4mm} $\ct$: constant-degree rooted tree with height $\eta$ and $m$ leaves 
		\Comment leaf $i$ corresponds to $\vx_i$
		\State \hspace{4mm} $\sketchlen \defeq \Theta(\eta^{2}\log(\frac{m}{\rho}))$: sketch dimension
		\State \hspace{4mm} $\mphi \sim \mathbf{N}(0,\frac{1}{w})^{w \times m}$: JL-sketch matrix
		\State \hspace{4mm} $\delta>0$: additive approximation error
		\State \hspace{4mm} $k$: current step
		\State \hspace{4mm} $\ox \in \R^{m}$: current valid approximate vector
		\State \hspace{4mm} $\{\vx^{(j)} \in \R^{m}\}_{j=0}^{k}$: list of previous inputs 
		\State \hspace{4mm} $\{\md^{(j)}\in\R^{m\times m}\}_{j=0}^{k}$: list of previous diagonal scaling matrices 
		
		\State
		\Procedure{\textsc{Initialize}}{$\ct, \vx\in\R^{m}, \md\in\R_{>0}^{m\times m}, \rho>0, \delta>0$}
		\State $\ct \leftarrow \ct$, $\delta\leftarrow\delta$, $k \leftarrow 0$
		\State $\ox\leftarrow\vx, \vx^{(0)} \leftarrow \vx, \md^{(0)} \leftarrow \md$
		\State sample $\mphi \sim \mathbf{N}(0,\frac{1}{w})^{w \times m}$
		\EndProcedure
		
		\State
		\Procedure{\textsc{Update}}{$\vx^\new \in\R^{m}, \md^\new \in\R_{>0}^{m\times m}$}
		\State $k\leftarrow k+1$, $\vx^{(k)}\leftarrow\vx^\new$, $\md^{(k)}\leftarrow\md^\new$
		\EndProcedure
		
		\State
		\Procedure{\textsc{Approximate}}{}
		\State $I \leftarrow \emptyset$
		
		\ForAll {$0 \leq \ell < \left\lceil \log m\right\rceil $ such that $k \equiv 0\bmod2^{\ell}$}
		
		\State $I_{\ell}^{(k)} \leftarrow \textsc{FindLargeCoordinates}(\ell)$
		\label{line:set-I_ell^k}
		\State $I \leftarrow I\cup I_{\ell}^{(k)}$
		\EndFor
		
		\If{$k \equiv 0 \bmod 2^{\left\lceil \log m\right\rceil}$} 
		\State $I\leftarrow[m]$ \label{line:set-ox-to-x}
		\Comment Update $\ox$ in full every $2^{\left\lceil \log m\right\rceil}$ steps
		\EndIf
		\State $\ox_{i}\leftarrow\vx_{i}^{(k)}$ for all $i\in I$ \label{line:update-ox}
		\State \Return $\ox$
		
		\EndProcedure
		\algstore{approx-vector-break-point}
	\end{algorithmic}
\end{algorithm}
\addtocounter{algorithm}{-1}

\begin{lemma}[Approximate vector maintenance] \label{lem:maintain-approx}
	Suppose \textsc{FindLargeCoordinates}$(\ell)$ is a procedure in \textsc{AbstractMaintainApprox} that correctly computes the set $I_\ell^{(k)}$ at the $k$-th step.
	Then the deterministic data structure \textsc{AbstractMaintainApprox} in \cref{algo:maintain-vector} maintains an approximation $\ox$ of $\vx$ with the following procedures:
	\begin{itemize}
		\item \textsc{Initialize}$(\ct, \vx$, $\md$, $\rho > 0$, $\delta>0)$: 
		Initialize the data structure at step 0 with tree ${\ct}$, initial vector $\vx$, initial diagonal scaling matrix $\md$, target additive approximation error $\delta$, and success probability $1 - \rho$. 
		
		\item \textsc{Update}$(\vx^\new$, $\md^\new$): 
		Increment the step counter and update vector $\vx$ and diagonal scaling matrix $\md$. 
		
		\item \textsc{Approximate}:
		Output a vector $\ox$ such that $\|\md (\vx-\ox)\|_{\infty}\leq\delta$ for the latest $\vx$ and $\md$ with probability at least $1 - \rho$.
	\end{itemize}
	At step $k$, the procedure \textsc{Update} is called, followed by \textsc{Approximate}.
	Suppose $\|\md^{(k)}(\vx^{(k)}-\vx^{(k-1)})\|_2 \leq\beta$
	\emph{for all steps $k$}, and $\md$ is a function of $\ox$ coordinate-wise.
	Then, for each $\ell \geq 0$, the data structure updates $O(2^{2\ell}(\beta/\delta)^{2}\log^{2}m)$ coordinates of $\ox$ every $2^\ell$ steps.
\end{lemma}

\begin{proof}[Proof of \cref{lem:maintain-approx}]
	The failure case arises from \textsc{FindLargeCoordinates}.
	Assuming \textsc{FindLargeCoordinates} returns the correct set of coordinates,
	we prove the correctness of \textsc{Approximate}. 
	
	Fix some coordinate $i\in[m]$ and fix some step $k$. 
	Suppose the latest update to $\ox_i$ is $\ox_i\leftarrow \vx_i^{(k')}$ at step $k'$. 
	So $\md_{ii}^{(d)}$ is the same for all $k\ge d >k'$, 
	and $i$ is not in the set $I_{\ell}^{(d)}$ returned by \textsc{FindLargeCoordinates} for all $k\ge d>k'$. 
	Since we set $\ox\leftarrow\vx$ every $2^{\left\lceil \log m\right\rceil }$ steps by \cref{line:set-ox-to-x}, 
	we know $k-2^{\left\lceil \log m\right\rceil }\leq k'<k$.
	Using dyadic intervals, we can define a sequence $k_0, k_1, \dots, k_s$ with $s\leq2\left\lceil \log m\right\rceil$, 
	where $k'=k_{0}<k_{1}<k_{2}<\cdots<k_{s}=k$, 
	each $k_{j+1}-k_{j}$ is a power of $2$, 
	and $(k_{j+1}-k_{j}) \mid k_{j+1}$.
	Hence, we have 
	\[
	\vx_{i}^{(k)}-\ox_{i}^{(k)}=\vx_{i}^{(k_{s})}-\ox_{i}^{(k_{0})}=\vx_{i}^{(k_{s})}-\vx_{i}^{(k_{0})} = \sum_{j=0}^{s-1} \left(\vx_{i}^{(k_{j+1})}-\vx_{i}^{(k_{j})}\right).
	\]
	We know that $\md_{ii}^{(d)}$ is the same for all $k\ge d >k'$. By the guarantees of \textsc{FindLargeCoordinates}, we have
	\[
	\md_{ii}^{(k)} \cdot|\vx_{i}^{(k_{j+1})}-\vx_{i}^{(k_{j})}|= 
	\md_{ii}^{(k_{j+1})} \cdot|\vx_{i}^{(k_{j+1})} -\vx_{i}^{(k_{j})}|\leq\frac{\delta}{2\left\lceil \log m\right\rceil }
	\] for all $0\le j<s$. 
	Summing over all $j=0,1,\ldots,s-1$ gives 
	\[
	\md_{ii}^{(k)} \cdot|\vx_{i}^{(k)}-\ox_{i}^{(k)}|\leq\delta.
	\]
	Hence, we have $\|\md (\vx-\ox)\|_{\infty}\leq\delta$.
	
	Fix some $\ell$ with $k \equiv 0 \bmod 2^\ell$. We bound the number of coordinates in $I^{(k)}_\ell$.
	For any $i\in I_{\ell}^{(k)}$, we know $\md_{ii}^{(j)}=\md_{ii}^{(k)}$ for all $j>k-2^{\ell}$ because $\ox_{i}$ did not change in the meanwhile. By definition of $I_\ell^{(k)}$, we have
	\[
	\md_{ii}^{(k)} \cdot\sum_{j=k-2^{\ell}}^{k-1}|\vx_{i}^{(j+1)}-\vx_{i}^{(j)}| \geq 
	\md_{ii}^{(k)} \cdot|\vx_{i}^{(k)}-\vx_{i}^{(k-2^{\ell})}|\geq\frac{\delta}{2\left\lceil \log m\right\rceil }.
	\]
	
	Using $\md_{ii}^{(j)}=\md_{ii}^{(k)}$ for all $j>k-2^{\ell}$ again, the above inequality yields 
	\begin{align*}
		\frac{\delta}{2\left\lceil \log m\right\rceil } &\leq\sum_{j=k-2^{\ell}}^{k-1} \md_{ii}^{(j+1)} |\vx_{i}^{(j+1)}-\vx_{i}^{(j)}| \\
		&\leq \sqrt{2^{\ell}\sum_{j=k-2^{\ell}}^{k-1} {\md_{ii}^{(j+1) 2}} |\vx_{i}^{(j+1)}-\vx_{i}^{(j)}|^{2}}. \tag{by Cauchy-Schwarz}
	\end{align*}
	Squaring and summing over all $i\in I_{\ell}^{(k)}$ gives 
	\begin{align*}
		\Omega \left(\frac{2^{-\ell}\delta^{2}}{\log^{2}m}\right)|I_{\ell}^{(k)}| &\leq \sum_{i\in I_{\ell}^{(k)}}\sum_{j=k-2^{\ell}}^{k-1}\md_{ii}^{(j+1) 2}|\vx_{i}^{(j+1)}-\vx_{i}^{(j)}|^{2} \\
		&\leq
		\sum_{i=1}^{m}\sum_{j=k-2^{\ell}}^{k-1}\md_{ii}^{(j+1) 2}|\vx_{i}^{(j+1)}-\vx_{i}^{(j)}|^{2} \\
		&\leq2^{\ell}\beta^{2},
	\end{align*}
	where we use $\|\md^{(j+1)}(\vx^{(j+1)}-\vx^{(j)})\|_2 \leq \beta$
	at the end. Hence, we have 
	\[
	|I_{\ell}^{(k)}|=O(2^{2\ell}(\beta/\delta)^{2}\log^{2}m).
	\]
	In other words, for each $\ell \geq 0$, we update $|I_{\ell}^{(k)}|$-many coordinates of $\ox$ at step $k$ when $k \equiv 0 \mod 2^\ell$.
	So we conclude that
	for each $\ell \geq 0$, we update $O(2^{2\ell}(\beta/\delta)^{2}\log^{2}m)$-many coordinates of $\ox$ every $2^\ell$ steps.
\end{proof}

\subsection{From change detection to sketch maintenance}
\label{subsec:sketch_change_to_sketch}

Now we discuss the implementation of \textsc{FindLargeCoordinates}$(\ell)$ to find the set $I_{\ell}^{(k)}$ in \cref{line:set-I_ell^k} of \cref{algo:maintain-vector}.
We accomplish this by repeatedly sampling a coordinate $i$ with probability proportional to $\left({\md_{ii}^{(k)}} (\vx_{i}^{(k)}-\vx_{i}^{(k-2^{\ell})}) \right)^{2}$,
among all coordinates $i$ where $\ox_i$ has not been updated since $2^\ell$ steps ago.
With high probability, we can find all indices in $I_\ell^{(k)}$ in this way efficiently.
To implement the sampling procedure, we make use of a data structure based on segment trees~\cite{CLRS} along with sketching based on the Johnson-Lindenstrauss lemma.

Formally, we define the vector $\vq \in \R^m$ where $\vq_i \defeq \md^{(k)}_{ii} (\vx^{(k)}_i-\vx^{(k-2^{\ell})}_i)$ if $\ox_i$ has not been updated after the $(k-2^\ell)$-th step, and $\vq_i = 0$ otherwise. 
Our goal is precisely to find all large coordinates of $\vq$.

Let $\ct$ be a constant-degree rooted tree with $m$ leaves, where leaf $i$ represents coordinate $\vq_i$, which we call a \emph{sampling tree}.
For each node $u \in \ct$, we define $E(u) \subseteq [m]$ to be the set of indices of leaves in the subtree rooted at $u$.
We make a random descent down $\ct$, in order to sample a coordinate $i$ with probability proportional to $\vq_i^2$.
At a node $u$, for each child $u'$ of $u$, the total probability of the leaves under $u'$ is given precisely by $\norm{\vq|_{E(u')}}_2^2$. We can estimate this by the Johnson-Lindenstrauss lemma using a sketching matrix $\mphi$. Then we randomly move from $u$ down to child $u'$ with probability proportional to the estimated value.
To tolerate the estimation error, when reaching some leaf node representing coordinate $i$, we accept with probability proportional to the ratio between the exact probability of $i$ and the estimated probability of $i$. If $i$ is rejected, we repeat the process from the root again independently. 

\begin{algorithm}
	\caption{Data structure \textsc{AbstractMaintainApprox}, Part 2}
	\label{algo:change-detection}
	\begin{algorithmic}[1]
		\renewcommand{\thealgorithm}{}
		\algrestore{approx-vector-break-point}
		
		\Procedure{FindLargeCoordinates}{$\ell$}
		\State \LeftComment $\overline{\md}$: diagonal matrix such that 
		\[
		\overline{\md}_{ii}=\begin{cases}
			\md_{ii}^{(k)} & \text{\text{if }}\text{$\ox_i$ has not been updated after the  $(k-2^{\ell})$-th step}\\
			0 & \text{otherwise.}
		\end{cases}
		\]
		\State \LeftComment $\vq \defeq \overline{\md} (\vx^{(k)}-\vx^{(k-2^{\ell})})$ \Comment{vector to sample coordinates from}
		\State
		\State $I \leftarrow \emptyset$ \Comment{set of candidate coordinates}
		\For{$M_{\ell} \defeq \Theta(2^{2\ell}(\beta/\delta)^{2}\log^{2}m\log(m/\rho))$ iterations}\label{line:outer-loop}
		
		\State \LeftComment Sample coordinate $i$ of $\vq$ w.p. proportional to $\vq_{i}^{2}$ by random descent down $\ct$ to a leaf
		\While{\textbf{true}} \label{line:sample-loop}
		\State $u\leftarrow\textrm{root}(\ct)$, $p_{u}\leftarrow1$ \label{line:sample-start}
		\While{$u$ is not a leaf node} 
		\State Sample a child $u'$ of $u$ with probability 
		\[
		\mathbf{P} (u\rightarrow u')\defeq\frac{\|\mphi_{E(u')} \vq\|_{2}^{2}}{\sum_{\text{child $u''$ of $u$}} \|\mphi_{E(u'')}\vq\|_{2}^{2}}
		\] \label{line:sample-a-child}
		\Comment let $\mphi_{E(u)} \defeq \mphi\mi_{E(u)}$ for each node $u$
		\State $p_{u}\leftarrow p_{u}\cdot \mathbf{P} (u\rightarrow u')$
		\State $u\leftarrow u'$
		\EndWhile 
		\State \textbf{break} with probability $p_{\mathrm{accept}}\defeq \norm{\vq|_{E(u)}}^{2} /(2\cdot p_{u}\cdot\|\mphi\vq\|_{2}^{2})$ \label{line:sample-filter} 
		\EndWhile
		
		\State $I\leftarrow I\cup E(u)$ \label{line:sample-one-cor}
		\EndFor
		
		\State \Return $\{i\in I \;:\; \vq_{i} \geq\frac{\delta}{2\left\lceil \log m\right\rceil }\}$.
		\EndProcedure
	\end{algorithmic}
\end{algorithm}

\begin{lemma} 
	\label{lem:change-detection}
	Assume that $\|{\md^{(k+1)}} (\vx^{(k+1)}-\vx^{(k)})\|_2 \leq \beta$ for all IPM steps $k$. 
	Let $\rho < 1$ be any given failure probability, and let $M_{\ell} \defeq \Theta(2^{2\ell}(\beta/\delta)^{2}\log^{2}m\log(m/\rho))$ be the number of samples \cref{algo:change-detection} takes.
	Then with probability $\geq1-\rho$, during the $k$-th call of \textsc{Approximate},  \cref{algo:change-detection} finds the set $I_{\ell}^{(k)}$ correctly. 
	Furthermore, the while-loop in \cref{line:sample-loop} happens only $O(1)$ times in expectation per sample.  
\end{lemma}

\begin{proof}The proof is similar to Lemma 6.17 in \cite{treeLPArxivV2}. We
	include it for completeness. 
	For a set $S$ of indices, let $\mi_{S}$ be the $m \times m$ diagonal matrix that is one on $S$ and zero otherwise. 
	
	We first prove that \cref{line:sample-filter} breaks with probability
	at least $\frac{1}{4}$. By the choice of $\sketchlen$, Johnson--Lindenstrauss
	lemma shows that $\|\mphi_{E(u)}\vq\|_{2}^{2}=(1\pm\frac{1}{9\eta})\|\mi_{E(u)}\vq\|_{2}^{2}$
	for all $u \in \ct$ with probability at least $1-\rho$. 
	Therefore, the probability we move from a node $u$ to its child node $u'$ is given by
	\[
	\mathbf{P} (u\rightarrow u')= \left(1\pm\frac{1}{3\eta} \right) \frac{\|\mi_{E(u')}\vq\|_{2}^{2}}{\sum_{u''\text{ is a child of }u}\|\mi_{E(u'')}\vq\|_{2}^{2}}
	=
	\left(1\pm\frac{1}{3\eta}\right)\frac{\|\mi_{E(u')}\vq\|_{2}^{2}}{\|\mi_{E(u)}\vq\|_{2}^{2}}.
	\]
	Hence, the probability the walk ends at a leaf $u \in \ct$ is given by
	\[
	p_{u}= \left(1\pm\frac{1}{3\eta} \right)^{\eta} \frac{\|\mi_{u}\vq\|_{2}^{2}}{\|\vq\|_{2}^{2}}=(1\pm\frac{1}{3\eta})^{\eta}\frac{\norm{\vq|_{E(u)}}^{2}}{\|\vq\|_{2}^{2}}.
	\]
	Now, $p_{\mathrm{accept}}$ on \cref{line:sample-filter} is at least
	\[
	p_{\mathrm{accept}}=\frac{\norm{\vq|_{E(u)}}^{2}} {2\cdot p_{u}\cdot\|\mphi\vq\|_{2}^{2}}
	\geq 
	\frac{\norm{\vq|_{E(u)}}^{2}} {2\cdot(1+\frac{1}{3\eta})^{\eta} \frac{\norm{\vq|_{E(u)}}^{2}}{\|\vq\|_{2}^{2}}\cdot\|\mphi\vq\|_{2}^{2}}
	\geq
	\frac{\|\vq\|_{2}^{2}}{2\cdot(1+\frac{1}{3\eta})^{\eta}\|\mphi\vq\|_{2}^{2}}
	\geq 
	\frac{1}{4}.
	\]
	On the other hand, we have that $p_{\mathrm{accept}}\leq\frac{\|\vq\|_{2}^{2}}{2(1-\frac{1}{3\eta})^{\eta}\|\mphi\vq\|_{2}^{2}}<1$
	and hence this is a valid probability.
	
	Next, we note that $u$ is accepted on \cref{line:sample-filter} with probability
	\[
	p_{\mathrm{accept}} p_{u}=\frac{\norm{\vq|_{E(u)}}^{2}}{2\cdot\|\mphi\vq\|_{2}^{2}}.
	\]
	Since $\|\mphi\vq\|_{2}^{2}$ remains the same in all iterations, this probability is proportional to $\norm{\vq|_{E(u)}}^{2}$.
	Since the algorithm repeats when $u$ is rejected, on \cref{line:sample-one-cor},
	$u$ is chosen with probability exactly $\norm{\vq|_{E(u)}}^{2}/\|\vq\|^{2}$.
	
	Now, we want to show the output set is exactly $\{i \in [n] :|\vq_{i}|\geq\frac{\delta}{2\left\lceil \log m\right\rceil }\}$.
	Let $S$ denote the set of indices where $\ox$ did not update between the $(k-2^\ell)$-th step and the current $k$-th step.
	Then
	\begin{align*}
		\|\vq\|_{2} &=\|\mi_{S} \md^{(k)} (\vx^{(k)}-\vx^{(k-2^{\ell})})\|_{2}\\
		& \leq\sum_{i=k-2^{\ell}}^{k-1}\|\mi_{S} \md^{(k)} (\vx^{(i+1)}-\vx^{(i)})\|_{2}\\
		& =\sum_{i=k-2^{\ell}}^{k-1}\|\mi_{S} \md^{(i+1)} (\vx^{(i+1)}-\vx^{(i)})\|_{2}\\
		& \leq\sum_{i=k-2^{\ell}}^{k-1}\| \md^{(i+1)} (\vx^{(i+1)}-\vx^{(i)})\|_{2} \\
		&\leq 2^{\ell}\beta,
	\end{align*}
	where we used $\mi_S \md^{(i+1)} = \mi_S \md^{(k)}$, because $\ox_{i}$ changes whenever $\md_{ii}$ changes at a step. Hence,
	each leaf $u$ is sampled with probability at least $\norm{\vq|_{E(u)}}^{2}/(2^{\ell}\beta)^{2}$.
	If $|\vq_{i}|\geq\frac{\delta}{2\left\lceil \log m\right\rceil }$,
	and $i \in E(u)$ for a leaf node $u$, then the coordinate $i$ is not in $I$ with probability at most
	\[
	\left(1-\frac{\norm{\vq|_{E(u)}}^{2}}{(2^{\ell}\beta)^{2}} \right)^{M_{\ell}}
	\leq \left(1-\frac{1}{2^{2\ell+2}(\beta/\delta)^{2}\left\lceil \log m\right\rceil ^{2}}\right)^{M_{\ell}}\leq\frac{\rho}{m},
	\]
	by our choice of $M_{\ell}$. Hence, all $i$ with $|\vq_{i}|\geq\frac{\delta}{2\left\lceil \log m\right\rceil }$
	lies in $I$ with probability at least $1-\rho$. This proves that
	the output set is exactly $I_{\ell}^{(k)}$ with probability at least
	$1-\rho$. 
\end{proof} 

\begin{remark}
	
	In \cref{algo:change-detection}, we only need to compute $\|\mphi_{E(u)}\vq\|_{2}^{2}$ for $O(M_{\ell})$ many nodes $u \in \ct$. 
	Furthermore, 
	the randomness of the sketch is not leaked and we can use the same
	random sketch $\mphi$ throughout the algorithm. 
	This allows us to efficiently maintain $\mphi_{E(u)}\vq$ for each $u \in \ct$ throughout the IPM.
	
\end{remark}
	
\subsection{Sketch maintenance} \label{subsec:sketch_maintenance}

In \textsc{FindLargeCoordinates} in the previous subsection, we assumed the existence of a constant degree sampling tree $\ct$, 
and for the dynamic vector $\vq$, 
the ability to access $\mphi_{E(u)} \vq$ at each node $u \in \ct$ and $\vq|_{E(u)}$ at each leaf node $u$.

In this section, we consider when the required sampling tree is the operator tree $\ct$ supporting a tree operator $\treeop$, and
the vector $\vq$ is $\vx \defeq \treeop \vz + \sum_{H \in \ct} \treeop^{(H)} \vu_H$,
where each of $\treeop, \vz$ and the $\vu_H$'s undergo changes at every IPM step. 
We present a data structure that implements two features efficiently on $\ct$:
\begin{itemize}
	\item access $\vx|_{E(H)}$ at every leaf node $H$,
	\item access $\mphi_{E(H)} \vx$ at every node $H$, where $\mphi_{E(H)}$ is $\mphi$ restricted to columns given by $E(H)$.
\end{itemize}

\begin{algorithm}
	\caption{Data structure for maintaining $\mphi \vx$, Part 1}
	\label{algo:maintain-sketch}
	\begin{algorithmic}[1]
		\State \textbf{data structure} \textsc{MaintainSketch} 
		\State \textbf{private : member}
		\State \hspace{4mm} $\ct$ : rooted constant degree tree, where at every node $H$, there is
		\State \hspace{12mm} $\ms^{(H)} \in \R^{\sketchlen \times |\skel{H}|}$ : sketched subtree operator $\mphi \treeop^{(H)}$
		\State \hspace{12mm} $\vt^{(H)} \in \R^{\sketchlen}$ : sketched vector $\mphi \sum_{H' \in \ct_H} \treeop^{(H')}\vu_{H'}$
		\State \hspace{4mm} $\mphi\in\R^{\sketchlen\times m}$ : JL-sketch matrix
		\State \hspace{4mm} $\treeop \in \R^{m \times n}$ : dynamic tree operator on $\ct$
		\State \hspace{4mm} $\vu_H$ at every $H \in \ct$ : dynamic vectors
		
		\State
		\Procedure{\textsc{Initialize}}{tree $\ct$, $\mphi\in\R^{\sketchlen\times m}$, tree operator $\treeop$, $\vu_H$ for each $H \in \ct$}
		\State $\mphi \leftarrow\mphi, \ct \leftarrow \ct, \treeop \leftarrow \treeop, \vu_H \leftarrow \vu_H$ for each $H \in \ct$
		\State $\ms^{(H)} \leftarrow \vzero$, $\vt^{(H)} \leftarrow \vzero$ for each $H \in \ct$
		\State {$\textsc{Update}$}({$V(\ct)$}) 
		\EndProcedure
		
		\State
		\Procedure{\textsc{Update}}{$\mathcal{H} \defeq $ set of nodes admitting implicit representation changes}
		\For{$H \in \mathcal{P}_{\ct}(\mathcal{H})$ going up the tree level by level}
		\State $\ms^{(H)} \leftarrow \sum_{\text{child $D$ of $H$}} \ms^{(D)} \treeop_{D}$ \label{line:collect}
		\State $\vt^{(H)} \leftarrow \ms^{(H)}\vu_H + \sum_{\text{child  $D$ of $H$}} \vt^{(D)}$ \label{line:collect2} \label{line:collecty}
		\EndFor
		\EndProcedure
		
		\State
		\Procedure{$\textsc{SumAncestors}$}{$H \in \ct$}
		\If {\textsc{Update} has not been called since the last call to \textsc{SumAncestors}$(H)$}
		\State \Return the result of the last \textsc{SumAncestors}$(H)$
		\EndIf
		\If {$H$ is the root} 
		\Return $\vzero$
		\EndIf
		\State \Return $\treeop_{H} (\vu_P + \textsc{SumAncestors}(P))$ 
		\Comment $P$ is the parent of $H$
		\EndProcedure
		\State
		\Procedure{$\textsc{Estimate}$}{$H \in \ct$}
		\State Let $\vy$ be the result of \textsc{SumAncestors}$(H)$
		\State \Return $\ms^{(H)}\vy + \vt^{(H)}$
		\EndProcedure
		\State
		\Procedure{$\textsc{Query}$}{leaf $H \in \ct$}
		\State \Return $\vu_H + \textsc{SumAncestors}(H)$
		\EndProcedure
	\end{algorithmic}
\end{algorithm}

\begin{lemma} \label{lem:maintain-sketch}
	Let $\ct$ be a constant degree rooted tree with height $\eta$ supporting tree operator $\treeop$ with query complexity $Q$.
	Let $\sketchlen = \Theta(\eta^{2}\log(\frac{m}{\rho}))$ be as defined in \cref{algo:maintain-vector},
	and let $\mphi\in\R^{\sketchlen\times m}$ be a JL-sketch matrix.
	Then \textsc{MaintainSketch} (\cref{algo:maintain-sketch}) is a data structure that maintains $\mphi \vx$, 
	where $\vx$ is implicitly represented by
	\[
		\vx \defeq \treeop \vz + \sum_{H \in \ct} \treeop^{(H)} \vu_H.
	\]
	The data structure supports the following procedures: 
	\begin{itemize}
		\item \textsc{Initialize}(operator tree $\ct$, implicit $\vx$):
		Initialize the data structure and compute the initial sketches in $O(Q(\sketchlen m))$ time.
		\item \textsc{Update}($\collN \subseteq \ct$):
		Update all the necessary sketches in $O(\sketchlen\cdot Q(\eta |\collN|))$ time, where
		$\collN$ is the set of all nodes $H$ where $\vu_H$ or $\treeop_{H}$ changed.
		\item \textsc{Estimate}$(H \in \ct)$: Return $\mphi_{E(H)} \vx$.
		\item \textsc{Query}$(H \in \ct)$: Return $\vx|_{E(H)}$.
	\end{itemize}

	If we call \textsc{Query} on $N$ nodes, the total runtime is $O(Q(\sketchlen \eta N))$.

	If we call \textsc{Estimate} along a sampling path (by which we mean starting at the root, calling estimate at both children of a node, and then recursively descending to one child until reaching a leaf), and then we call \textsc{Query} on the resulting leaf, and we repeat this $N$ times with no updates during the process,
	then the total runtime of these calls is $O(Q(\sketchlen \eta N))$.
\end{lemma}

We note that $\treeop \vz = \sum_{H \in \ct} \treeop^{(H)} \vz|_{F_H}$.
For simplicity, it suffices to give the algorithm for sketching the simpler $\vx \defeq \sum_{H \in \ct} \treeop^{(H)} \vu_H$.

\begin{proof}
	Let us consider the correctness of the data structure, starting with the helper function \textsc{SumAncestors}.
	We implement it using recursion and memoization as it is crucial for bounding subsequent runtimes.
	
	\begin{claim}
	\textsc{SumAncestors}$(H \in \ct)$ returns $\sum_{\text{ancestor $A$ of $H$}} \treeop_{H \leftarrow A}  \vu_{A}$.
	\end{claim}
	\begin{proof}
	At the root, there are no ancestors, hence we return the zero matrix. 
	When $H$ is not the root, suppose $P$ is the parent of $H$. Then we can recursively write 
	\[
	\sum_{\text{ancestor $A$ of $H$}} \treeop_{H \leftarrow A} \vu_A = \treeop_{H} \left(\vu_P + \sum_{\text{ancestor $A$ of $P$}} \treeop_{P \leftarrow A} \vu_A \right).
	\]
	The procedure implements the right hand side, and is therefore correct.
	\end{proof}
	
	Assuming we correctly maintain $\ms^{(H)} \defeq \mphi \treeop^{(H)}$ and $\vt^{(H)} \defeq \mphi \sum_{H' \in \ct_H} \treeop^{(H')} \vu_{H'}$ at every node $H$, 
	\textsc{Estimate} and \textsc{Query} return the correct answers by the tree operator decomposition given in \cref{lem:treeop-subtree-ancestor-decomp}.
	
	For \textsc{Update}, note that if a node $H$ is not in $\collN$ and it has no descendants in $\collN$, then by definition, the sketches at $H$ are not changed.
	Hence, it suffices to update the sketches only at all nodes in $\mathcal{P}_{\ct}(\mathcal{H})$. 
	We update the nodes from the bottom of $\ct$ upwards, so that when we are at a node $H$, 
	all the sketches at its descendant nodes are correct. 
	Therefore, by definition, the sketches at $H$ is also correct.
	
	Now we consider the runtimes:
	
	\paragraph{\textsc{Initialize}:}
	It sets the sketches to $\vzero$ in $O(\sketchlen m)$ time, and then calls \textsc{Update} to update the sketches everywhere on $\ct$. 
	By the correctness runtime of \textsc{Update}, this step is correct and runs in $\O(Q(\sketchlen m))$ time.
	
	\paragraph{\textsc{Update}(set of nodes $\mathcal{H}$ admitting implicit representation changes):}
	
	First note that $|\mathcal{P}_{\ct}(\collN)| \leq \eta |\collN|$.
	For each node $H \in \collN$ with children $D_1, D_2$, 
	\cref{line:collect} multiplies each row of $\ms^{(D_1)}$ with $\treeop_{(D_1, H)}$,  
	each row of $\ms^{(D_2)}$ with $\treeop_{D_2}$, and sums the results.
	Summing over $\sketchlen$-many rows and over all nodes in $\mathcal{P}_{\ct}(\mathcal{H})$,
	we see the total runtime of \cref{line:collect} is $O(Q(\sketchlen \eta |\collN|))$.
	
	\cref{line:collect2} multiply each row of $\ms^{(H)}$ with a vector and then performs a constant number of additions of $\sketchlen$-length vectors. Since $\ms^{(H)}$ is computed for all $H \in \mathcal{P}_\ct(\collN)$ in $O(Q(\sketchlen \eta |\collN|))$ total time, this must also be a bound on their number of total non-zero entries. Since each $\ms^{(H)}$ is used once in \cref{line:collect2} for a matrix-vector multiplication, the total runtime of \cref{line:collect2} is $O(Q(\sketchlen \eta |\collN|))$. 
	
	All other lines are not bottlenecks.
	
	\paragraph{Overall \textsc{Estimate} and \textsc{Query} time along $N$ sampling paths:}
	We show that if we call \textsc{Estimate} along $N$ sampling paths
	each from the root to a leaf, and we call $\textsc{Query}$
	on the leaves, the total cost is $O(Q(\sketchlen \eta N))$:
	
	Suppose the set of nodes visited is given by $\mathcal{H}$, then $|\mathcal{H}| \leq \eta N$.
	Since there is no update, and \textsc{Estimate} is called for a node only after it is called for its parent, 
	we know that $\textsc{SumAncestors}(H)$ is called exactly once for each $H \in \mathcal{H}$.
	Each $\textsc{SumAncestor}(H)$ multiplies a unique edge operator $\treeop_{(H,P)}$ with a vector. 
	Hence, the total runtime of \textsc{SumAncestors} is $Q(|\mathcal{H}|)$.
	
	Finally, each \textsc{Query} applies a leaf operator to the output of a unique \textsc{SumAncestors} call, so the overall runtime is certainly bounded by $O(Q(|\mathcal{H}|))$.
	Similarly, each \textsc{Estimate} multiplies $\ms^{(H)}$ with the output of a unique \textsc{SumAncestors} call.
	This can be computed as $\sketchlen$-many vectors each multiplied with the \textsc{SumAncestors} output. Then two vectors of length $\sketchlen$ are added. Summing over all nodes in $\collN$, the overall runtime is $O(Q(\sketchlen|\mathcal{H}|)) = O(Q(\sketchlen \eta N))$.
	
	\paragraph{\textsc{Query} time on $N$ leaves:}
	Since this is a subset of the work described above, the runtime must also be bounded by $O(Q(\sketchlen \eta N))$.
	
\end{proof}

\subsection{Proof of \crtcref{thm:soln-approx}}

We combine the previous three subsections for the overall approximation procedure.
It is essentially \textsc{AbstractMaintainApprox} in \cref{algo:maintain-vector}, 
with the abstractions replaced by a data structure implementation.
We omit the pseudocode and simply describe the functions.
\solnApprox*

\begin{proof}
We apply \cref{lem:maintain-approx} using $\vx$ maintained by \textsc{MaintainRep} and $\md \defeq \mw^{-1/2}$ from \textsc{Solve}.
We create $O(\log m)$ copies of \textsc{MaintainSketch} (\cref{lem:maintain-sketch}),
so that for each $0 \leq \ell \leq O(\log m)$, 
we have one copy $\texttt{sketch}_{\ell,x}$ which maintains sketches of $\mphi \omd \vx^{(k)}$ at step $k$, 
and one copy $\texttt{sketch}_{\ell}$ which maintains sketches of $\mphi \omd \vx^{(k - 2^\ell)}$ at step $k \geq 2^{\ell}$,
where $\omd$ is defined so $\omd_{i,i} = \md_{i,i}$ if $\ox_i$ has not been updated after the $k-2^\ell$-th step, and $\omd_{i,i} = 0$ otherwise (as needed in \cref{algo:change-detection}).
Note that $\omd$ can be absorbed into the tree operator in the implicit representation of $\vx$, so \cref{lem:maintain-sketch} does indeed apply.

To access skeches of the vector
$\vq \defeq \omd (\vx^{(k)} - \vx^{(k - 2^\ell)})$
as needed in \textsc{FindLargeCoordinates} in \cref{algo:change-detection}, 
we can simply access the corresponding sketch in $\texttt{sketch}_{\ell,x}$ and $\texttt{sketch}_{\ell}$, and then take the difference.

We now describe each procedure in words, and then prove their correctness and runtime. 
\paragraph{\textsc{Initialize}$(\ct, \vx, \md, \rho, \delta)$:}

This procedure implements the initialization of \textsc{AbstractMaintainApprox} to approximate the dynamic vector $\vx$ which is given implicitly.
The initialization steps described in \cref{algo:maintain-vector} takes $O(\sketchlen m)$ time.
Then, we initialize the $O(\log m)$ copies of \textsc{MaintainSketch} in $O(Q(\sketchlen m) \log m)$ time by \cref{lem:maintain-sketch}. 

\paragraph{\textsc{Update}$(\vx^\new, \md^\new)$:}
To implement \textsc{Update}, it suffices to update all the sketching data structures.
Let us fix $\ell$, and consider the update time for $\texttt{sketch}_{\ell,x}$ and $\texttt{sketch}_{\ell}$.

\cref{lem:maintain-approx} shows that throughout \textsc{Solve},
there are $O(2^{2\ell}(\beta/\delta)^{2}\log^{2}m)$-many coordinate updates to $\ox$ every $2^\ell$ steps.
Since $\md$ is a function of $\ox$ coordinate-wise, $\ox_{i}= \vx_{i}^{(k-1)}$ for all $i$ where $\md_{ii}^{(k)}\neq\md_{ii}^{(k-1)}$ by \cref{line:update-ox}.
The diagonal matrix $\omd$ is the same as $\md$,
except $\omd_{ii}$ is temporarily zeroed out for $2^\ell$ steps after $\overline{\vx}_i$ changes at a step. 
So, the overall number of coordinate changes to $\omd$ is $O(2^{2\ell})$-many every $2^\ell$ steps.

Let $S^{(k)}$ denote the number of nodes $H$ where $\treeop_{H}$ or $\vu_H$ in the implicit representation of $\vx$ changed at step $k$. Additionally, since 
the sketching data structures maintain some variant of $\omd \vx$ (where $\omd$ is viewed as absorbed in the tree operator), every coordinate change in $\omd$ implies an edge operator update.
Now we apply \cref{lem:maintain-sketch} to conclude that the total time for all \textsc{Update} calls for $\texttt{sketch}_{\ell,x}$ and $\texttt{sketch}_{\ell}$ over $N$ steps is:
\begin{align*}
	&\phantom{{}={}} O(1) \cdot \left(\sum_{k=1}^N Q \left(\sketchlen \eta S^{(k)} \right) + 
	\frac{N}{2^\ell} \cdot Q(\sketchlen \eta \cdot 2^{2\ell}) \right) 
	\leq O(\sketchlen \eta) \cdot \left( \sum_{k=1}^N Q(S^{(k)}) + 
	\frac{N}{2^\ell} \cdot Q(2^{2\ell}) \right) .
\end{align*}
We then sum this over all $\ell$ to get the total update time for the sketching data structures.

\paragraph{\textsc{Approximate}:}

There are two operations to be implemented in the subroutine \textsc{FindLargeCoordinates}$(\ell)$:
Accessing $\mphi_{E(u)} \vq$ at a node $u$, and accessing $\vq|_{E(u)}$ at a leaf node $u$.
For the first, we call $\texttt{sketch}_{\ell,x}.\textsc{Estimate}(u) - \texttt{sketch}_{\ell}.\textsc{Estimate}(u)$.
For the second, we call $\texttt{sketch}_{\ell,x}.\textsc{Query}(u) - \texttt{sketch}_{\ell}.\textsc{Query}(u)$.

To set $\ox_i$ as $\vx^{(k)}_i$ for a single coordinate at step $k$ as needed in \cref{line:update-ox},
we find the leaf node $H$ containing the edge $e$, and call $\texttt{sketch}_{0,x}.\textsc{Query}(H)$. 
This returns the sub-vector $\vx^{(k)}|_{E(H)}$, from which we can extract $\vx^{(k)}_{i}$ and set $\ox_i$ to be the value. This line is not a bottleneck in the runtime.

We compute the total runtime over $N$ \textsc{Approximate} calls.
For every $\ell \geq 0$, we call \textsc{FindLargeCoordinates}($\ell$) once every $2^\ell$ steps, for a total of $N/2^\ell$ calls.
In a single call,
$M_\ell \defeq \Theta(2^{2\ell} (\beta/\delta)^2 \log^2 m \log (mN /\rho))$ sampling paths are explored in the $\texttt{sketch}_{\ell}$ and $\texttt{sketch}_{\ell,x}$ data structures by \cref{lem:change-detection}, where a sampling path correspond to one iteration of the while-loop. This takes a total of $O(Q (\sketchlen \eta M_{\ell}))$ time by \cref{lem:maintain-sketch}.
Therefore, for every fixed $\ell$, the total time for all \textsc{FindLargeCoordinates}$(\ell)$ calls is
\[
	\frac{N}{2^\ell} \cdot O \left(Q(\sketchlen \eta M_{\ell}) \right).
\]
The total time for all \textsc{LargeCoordinates} calls is obtained by summing over all values of $\ell = 0, \dots, \log N$.
To achieve overall failure probability at most $\rho$, it suffices to set the failure probability of each call to be $O(\rho/N)$.
\end{proof}

We sum up the initialization time, update and approximate time for all values of $\ell = 0, \dots, \log N$ and over $N$ total steps of \textsc{Solve}, to get the overall runtime of the data structure:
\begin{align*}
	&\phantom{{}={}} \O(Q(\sketchlen m)) + 
	O(\sketchlen \eta) \sum_{k=1}^N Q(S^{(k)})  
	+ O(\sketchlen \eta) \sum_{\ell=0}^{\log N} \frac{N}{2^\ell} \left( Q(2^{2\ell}) + Q(M_{\ell}) \right) \\
	&= \O(\eta^3 (\beta/\delta)^2 \log^3(mN/\rho)) \left( Q(m) + \sum_{k=1}^N Q(S^{(k)}) +  \sum_{\ell=0}^{\log N} \frac{N}{2^\ell} \cdot Q(2^{2\ell}) \right).
\end{align*}

\end{document}